\pdfoutput=1

\documentclass[10pt]{article}

\usepackage{amsfonts,amsmath,amssymb,amsthm}
\usepackage[pdftex]{graphicx}
\usepackage[hmargin=1in,vmargin=1in]{geometry}
\usepackage{natbib,setspace,enumerate,mathtools}
\usepackage{multirow}
\usepackage{booktabs}
\allowdisplaybreaks
\usepackage[toc,title,titletoc,header]{appendix}
\usepackage{adjustbox}
\usepackage[colorlinks,citecolor=blue]{hyperref}
\usepackage{etoolbox} % This package goes here to add \qed to remarks automatically.
\usepackage{threeparttable}
\def\qed{\rule{2mm}{2mm}}
\def\independent{\perp \!\!\! \perp}

\newtheorem{theorem}{Theorem}[section]
\newtheorem{lemma}{Lemma}[section]
\newtheorem{definition}{Definition}[section]

\theoremstyle{definition}
\newtheorem{example}{Example}[section]
\newtheorem{remark}{Remark}[section]
\newtheorem{assumption}{Assumption}[section]

\AtEndEnvironment{remark}{~\qed}
\AtEndEnvironment{example}{~\qed}

\DeclareMathOperator*{\var}{Var}
\DeclareMathOperator*{\cov}{Cov}

\newcommand{\mycomment}[1]{}

\begin{document}

\author{
Jizhou Liu \\
HSBC Business School\\
Peking University\\
\url{jizhou.liu@phbs.pku.edu.cn}
}

\bigskip

\title{Inference for Two-stage Experiments under Covariate-Adaptive Randomization \thanks{I thank Yuehao Bai, Max Farrell, Christian Hansen, Tetsuya Kaji, Azeem Shaikh, Max Tabord-Meehan, and Panagiotis Toulis for helpful comments.}}

\maketitle

\vspace{-0.3in}

\begin{spacing}{1.2}
\begin{abstract}
This paper studies inference in two-stage randomized experiments under covariate-adaptive randomization. In the initial stage of this experimental design, clusters (e.g., households, schools, or graph partitions) are stratified and randomly assigned to control or treatment groups based on cluster-level covariates. Subsequently, an independent second-stage design is carried out, wherein units within each treated cluster are further stratified and randomly assigned to either control or treatment groups, based on individual-level covariates. Under the homogeneous partial interference assumption, I establish conditions under which the proposed difference-in-``average of averages'' estimators are consistent and asymptotically normal for the corresponding average primary and spillover effects and develop consistent estimators of their asymptotic variances. Combining these results establishes the asymptotic validity of tests based on these estimators. My findings suggest that ignoring covariate information in the design stage can result in efficiency loss, and commonly used inference methods that ignore or improperly use covariate information can lead to either conservative or invalid inference. Then, I apply these results to studying optimal use of covariate information under covariate-adaptive randomization in large samples, and demonstrate that a specific generalized matched-pair design achieves minimum asymptotic variance for each proposed estimator. Finally, I discuss covariate adjustment, which incorporates additional baseline covariates not used for treatment assignment. The practical relevance of the theoretical results is illustrated through a simulation study and an empirical application.
\end{abstract}
\end{spacing}
 

\hypersetup{pageanchor=false}
\thispagestyle{empty} 
\newpage
\hypersetup{pageanchor=true}
\setcounter{page}{1}

\section{Introduction}\label{sec:intro}
This paper considers the problem of inference in two-stage randomized experiments under covariate-adaptive randomization. Here, a two-stage randomized experiment refers to a design where clusters (e.g., households, schools, or graph partitions) are initially randomly assigned to either a control or treatment group. Subsequently, random assignment of units within each treated cluster to either treatment or control is carried out based on a pre-determined treated fraction. Covariate-adaptive randomization refers to randomization schemes that first stratify according to baseline covariates and then assign treatment status so as to achieve ``balance'' within each stratum. Two-stage randomized experiments are widely used in social science (see for example \cite{Duflo2003,Haushofer2016,McKenzie2021}), and discussed by statisticians (see for example \cite{hudgens2008}), as a general approach to causal inference with interference; that is, when one individual's treatment status affects outcomes of other individuals. Moreover, practitioners often use covariate information to design more efficient two-stage experiments \citep[see for example][]{Duflo2003,Ichino2012,Beuermann2015,Muralidharan2015,Hidrobo2016,Rogers2018,Kinnan2020,Banerjee2021,Malani2021}. However, to the best of my knowledge, there has not yet been any formal analysis on covariate-adaptive randomization in two-stage randomized experiments. Accordingly, this paper establishes general results about estimation and inference for two-stage designs under covariate-adaptive randomization. Subsequently, I propose and examine the optimality of two-stage designs with ``matched tuples'', i.e. a generalized matched-pair design (see \cite{Bai-optimal} and \cite{matched-tuple}).

This paper examines covariate-adaptive randomization for two-stage experiments within a comprehensive framework that encompasses matched tuples designs, stratified block randomization, and complete randomization as special cases. The framework relies on finely stratified randomization (see \cite{cytrybaum2021} and \cite{bai2024efficiency}), which involves grouping clusters into homogeneous strata of size $k$ and then assigning treatment entirely at random within each stratum.\footnote{The terms ``cluster'' and ``stratum'' are both used in the literature to describe groupings of units, which can lead to confusion. Here, a cluster is defined as a pre-determined group of units (e.g., households, schools, or graph partitions), and a stratum as a group of clusters that share similar baseline cluster-level covariates.} Within this framework, I propose a set of difference-in-``average of averages'' estimators and analyze the statistical inference for four parameters of interest: equally-weighted and size-weighted primary effects, and equally-weighted and size-weighted spillover effects, under the assumption of homogeneous partial interference, where interference is confined within clusters. I establish conditions under which these four estimators are asymptotically normal and construct consistent estimators of their corresponding asymptotic variances. These results collectively validate the asymptotic validity of tests based on these estimators.

This paper then considers the asymptotic properties of a commonly recommended inference procedure based on a linear regression with cluster-robust standard errors. My findings suggest that the corresponding $t$-test is generally valid but conservative. I also demonstrate that in the first stage of cluster-level assignment, covariate information about clusters is important for both designing efficient experiments and consistently estimating variances under covariate-adaptive randomization. However, in the second stage of unit-level assignment, while individual-level covariate information is useful for improving efficiency, it is not required for the proposed inference method. Specifically, I show that consistent variance estimators can be constructed using only the cluster-level covariates from the first stage design, regardless of the use of individual-level covariates in the second stage.

Next, I apply the results to study optimal use of covariate information in two-stage designs. Here, by ``optimal'', I mean designs that achieve the minimum asymptotic variances within the class of designs considered in the paper. For all estimands of interest, the designs in the first and second stage affect the efficiency independently. Thus, I am able to identify optimal designs in the first and second stage separately and use them together as the optimal two-stage design. My result shows that, at each stage, the asymptotically optimal design is a matched tuples design where clusters or units are matched based on an index function (similar to \cite{Bai-optimal}) that is specific to the given estimator. In a simulation study, the results demonstrate that properly designed two-stage experiments utilizing the optimality results outperform other designs. However, the efficiency gain achieved through proper second-stage randomization is significantly lower compared to the first stage under my simulation specifications. 

In the empirical literature, it is common to match or stratify on a small set of covariates expected to be most predictive of outcomes and to adjust for other pre-treatment covariates ex-post. Building on \cite{mp-cluster} and \cite{bai2023}, I propose a covariate-adjusted version of my estimator and discuss the conditions under which this estimator enhances asymptotic efficiency compared to the unadjusted version.

Finally, this paper evaluates the proposed inference method against various regression-based methods commonly used in empirical literature in a simulation study and empirical application. The simulation study confirms the asymptotic exactness of the  inference results and highlights that statistical inference based on various ordinary least squares regressions could either be too conservative or invalid. Specifically, my result verifies that the commonly used regression with cluster-robust standard errors is conservative, while the other regression-based methods examined in the paper, such as regressions with strata fixed effects or heteroskedasticity-robust standard errors, are generally invalid. In the empirical application, I demonstrate the proposed inference method based on the experiment conducted in \cite{foos2017} and compare it with regression-based methods. The empirical findings are consistent with the results of the simulation study.

The analysis of data from two-stage randomized experiments and experiments under covariate-adaptive randomization has received considerable attention, but most work has focused on only one of these two features at a time. Previous work on the analysis of two-stage randomized experiments includes \cite{Hirano2010}, \cite{liu2014}, \cite{Rigdon2015}, \cite{Baird2018}, \cite{basse2018}, \cite{basse2019}, \cite{Imai2021}, \cite{imai2022}, \cite{VAZQUEZBARE2022} and \cite{Gonzalo-two-stage}. Recent work on the analysis of covariate-adaptive experiments includes \cite{Bugni2018}, \cite{yichong2021}, \cite{bai-inference}, \cite{Bai-optimal}, \cite{mp-cluster}, \cite{matched-tuple}, \cite{yichong2022}, \cite{bai2023}, \cite{cytrybaum2021} and \cite{bai-attrition}. In fact, both \cite{basse2018} and \cite{Imai2021} applied their inference methods, which do not account for covariate information, to two-stage experiments under covariate-adaptive randomization.\footnote{\cite{basse2018} analyzes the empirical application from \cite{Rogers2018}, whose design involves stratification on school, grade, and prior-year absences. \cite{Imai2021} analyzes the empirical application from \cite{Kinnan2020}, whose design involves matching villages (clusters) and households into small blocks.} My framework of analysis follows closely \cite{Bugni2022}, in which they formalize cluster randomized experiments in a super population framework. 

This paper contributes to the methodology for a growing number of empirical papers using two-stage experiments with covariate-adaptive randomization. For instance, \cite{Muralidharan2015}, \cite{Hidrobo2016}, \cite{foos2017}, \cite{Rogers2018} and  \cite{Banerjee2021} conducted two-stage randomized experiments that stratify clusters or units into a small number of large strata according to their baseline covariates, typically known as stratified block randomization. \cite{Duflo2003}, \cite{Ichino2012}, \cite{Beuermann2015}, \cite{Kinnan2020} and  \cite{Malani2021}  conducted two-stage randomized experiments in which clusters or units are matched into small strata according to their baseline covariates, commonly known as matched pairs, matched triplets or matched tuples designs. 

The rest of the paper is organized as follows. Section \ref{sec:setup} describes the setup and notation. Section \ref{sec:main-results} presents the main results. Section \ref{sec:optimality} discusses the optimality of matched tuples designs. Section \ref{sec:covariate-adjustment} introduces the covariate-adjusted estimator. Section \ref{sec:simulation} examines the finite sample behavior of various experimental designs through simulations. Section \ref{sec:empirical} illustrates the proposed inference methods in an empirical application based on the experiment conducted in \cite{foos2017}. Finally, I conclude with recommendations for empirical practice in Section \ref{sec:conclude}.

\section{Setup and Notation}\label{sec:setup}
Let $Y_{i,g}$ and $X_{i,g}$ denote the observed outcome and individual baseline covariates of the $i$th unit in the $g$th cluster, respectively. Denote by $Z_{i,g}$ the indicator for whether the $i$th unit in the $g$th cluster is treated or not. Let $C_g$ denote the observed baseline covariates for the $g$th cluster, $N_g$ denote the size of the $g$th cluster, $H_g$ denote the target fraction of units treated in the $g$th cluster, and $G$ the number of observed clusters. In addition, define $\mathcal{M}_g$ as the (possibly random) subset of $\{1,...,N_g\}$ corresponding to the observations within the $g$th cluster that are sampled by the researcher. Let $M_g=|\mathcal{M}_g|$ denote the number of units in set $\mathcal{M}_g$. In other words, the researcher randomly assigns treatments to all $N_g$ units in the $g$th cluster but only observes or conducts analysis on a subset of units sampled from the $g$th cluster \citep[see for example][]{Beuermann2015,Muralidharan2015,Haushofer2016,Hidrobo2016,Aramburu2019,Haushofer2019,Banerjee2021,Malani2021}. Denote by $P_G$ the distribution of the observed data
\begin{equation*}
    V^{(G)} :=  \left(\left(Y_{i,g}, X_{i,g}, Z_{i,g}:  i \in \mathcal{M}_g\right), H_g, C_g, N_g: 1 \leq g \leq G\right)~.
\end{equation*}
This paper considers a setup where units are partitioned into a large number of clusters. In this context, the paper studies a two-stage randomized experiment with binary treatment in both stages. In the first stage, a fraction of $\pi_1$ clusters are randomly assigned to the treatment group, while the remaining clusters are assigned to the control group with no treated units. Then, conditional on the assignment in the first stage, a fraction of $\pi_2$ individuals from treated clusters are assigned to the treatment group, while the remaining units are assigned to the control group. Such a binary design is widely used in empirical literature \citep[see, e.g.,][]{Duflo2003,Ichino2012,Haushofer2016,foos2017,Haushofer2019}. Moreover, while some experiments have multiple treated fractions, researchers often analyze them as binary designs \citep[see, e.g.,][]{Beuermann2015,basse2018,Imai2021}.

The two-stage experiment closely resembles the split-plot design (see \cite{shi2022rerandomization,zhao2022splitplot}), where \(H_g\) represents the whole-plot (cluster-level) randomization and \(Z_{i,g}\) represents the subplot (within-cluster) randomization. In split-plot designs, \(H_g\) and \(Z_{i,g}\) are usually assumed to be independent and represent two binary factors of treatment. However, in two-stage designs, \(H_g\) represents the intended treated fraction and thus does not correspond to a real treatment; it is correlated with \(Z_{i,g}\) through the relation \(H_g = \sum_{1 \leq i \leq N_g} Z_{i,g} / N_g\).\footnote{Strictly speaking, the equality holds up to a finite sample error, i.e. $\lfloor H_g N_g \rfloor =  \sum_{1 \leq i \leq N_g} Z_{i,g}$.} This distinction indicates that it could be a promising direction for future research to develop a general framework that allows dependence between the first-stage and second-stage randomizations, encompassing both split-plot and two-stage designs.

%In other words, I allow ``imperfect balance'' for the first-stage design but impose ``perfect balance'' for the second stage design. \footnote{The term “balance” is often used in different ways. Here, “balance” refers to the extent to which the of fraction of treated units within a strata differs from the target proportion. See Remark 2.1 of \cite{Bugni2018} for more detailed discussion.}

%In order to formalize our assumption later on independence between first and second stage designs, I denote by $\{Z_{i,g}(h): h \in \mathcal{H}\}$ the potential individual-level treatment under different treated fractions $h\in\mathcal{H}$. Similar to the classic potential outcome model, the (observed) individual treatment indicator and potential individual-level treatment indicator are related to cluster-level treatment assignment by the relationship
%\begin{equation}\label{eqn:potential_treatment}
%    Z_{i,g} = \sum_{h \in \mathcal{H}}  Z_{i,g}(h) I\{H_g=h\} \text{ for } 1\leq i \leq N_g~.
%\end{equation}

\subsection{Potential Outcomes and Interference}
In this section, I provide assumptions on the interference structure that assume no interference across clusters and exchangeable/homogeneous interference within clusters. Let $Y_{i,g}(\mathbf{z}, n)$ denote the potential outcome of the $i$th unit in the $g$th cluster, where $n$ denotes the cluster size and $\mathbf{z}$ denotes a realized vector of assignment for all units in all clusters, i.e., $\mathbf{z}=((z_{i,g}: 1\leq i\leq n): 1\leq g\leq G)$, where $z_{i,g} \in \{0,1\}$ denotes a realized assignment for the $i$th unit in the $g$th cluster. Following previous work \citep[see, for example,][]{hudgens2008, basse2018, basse2019, Forastiere2021, Imai2021}, I assume the following about potential outcomes.
\begin{assumption}[Homogeneous partial interference]\label{as:interference}
\begin{equation*}
    Y_{i,g}(\mathbf{z}, n) = Y_{i,g}(\mathbf{z}^\prime, n) \text{ w.p.1 } \text{ if } z_{i,g} = z_{i,g}^\prime \text{ and }\sum_{1\leq j\leq n} z_{j,g} =  \sum_{1\leq j\leq n} z_{j,g}^\prime \text{ for any } 1 \leq i \leq n, 1 \leq g \leq G~,
\end{equation*}
where $\mathbf{z}$ and $\mathbf{z}^\prime$ are any realized vectors of assignment, and $z_{i,g}, z_{i,g}^\prime$ are the corresponding individual treatment indicators for $i$-th unit in $g$-th cluster.
\end{assumption}
Under Assumption \ref{as:interference}, potential outcomes can be simplified as $Y_{i,g}(z,n,n_1)$ where $n_1$ denotes the number of treated units in the cluster. Following this notation, we define
\begin{equation*}
    Y_{i,g}(z, h) := \sum_{n \geq 1} Y_{i,g}(z, n, \lfloor n h \rfloor) I\{N_g = n\} 
\end{equation*}
to be the potential outcome under the individual treatment status $z \in \{0,1\}$ and the cluster target treated fractions $h \in \mathcal{H} \subseteq [0, 1]$, where $\mathcal{H}$ is a pre-determined set of treated fractions.\footnote{For example, when the cluster size is $3$ and the target treated fraction is $0.5$, there will be one treated unit in the cluster. Other rounding approaches, like the ceiling function, to handle fractional numbers of treated units can also be easily accommodated.} As mentioned before, this paper considers binary treatments, i.e. $\mathcal{H} = \{0, \pi_2\}$, throughout the paper.\footnote{Extending the designs to accommodate multiple treatment fractions is technically straightforward. Related work can be found in \cite{Bugni2019}.} Furthermore, the (observed) outcome and potential outcomes are related to treatment assignment by the relationship $Y_{i,g} = Y_{i,g}(Z_{i,g}, H_{g})$. Denote by $Q_G$ the distribution of
\begin{equation*}
    W^{(G)}:= \left(\left(\left(Y_{i,g}(z, h):z \in \{0, 1\}, h\in \mathcal{H}\right), X_{i,g}: 1 \leq i \leq N_g \right), \mathcal{M}_g, C_g, N_g: 1 \leq g \leq G\right) ~.
\end{equation*}

\subsection{Distribution and Sampling Procedure}
The distribution $P_G$ of observed data and its sampling procedure can be described in three steps. First, $\{(\mathcal{M}_g, C_g, N_g): 1\leq g \leq G\}$ are i.i.d samples from a population distribution. Second, potential outcomes and baseline individual covariates are sampled from a conditional distribution $R_G(\mathcal{M}^{(G)}, C^{(G)}, N^{(G)})$, which is defined as follows:
\begin{equation*}
    \left(\left(\left(Y_{i,g}(z, h):z \in \{0, 1\}, h\in \mathcal{H}\right), X_{i,g}: 1 \leq i \leq N_g \right): 1\leq g \leq G\right) \mid \mathcal{M}^{(G)}, C^{(G)}, N^{(G)}~.
\end{equation*}
Finally, $P_G$ is jointly determined by the relationship $Y_{i,g} = Y_{i,g}(Z_{i,g}, H_{g})$ together with the assignment mechanism, which will be described in Section \ref{sec:main-results}, and $Q_G$, which is described in the first two steps. Note that \(A^{(G)}\) denotes the vector \((A_1, \dots, A_G)\) for any random variable \(A\), and \(X_g\) represents the vector \((X_{i,g}: 1 \leq i \leq N_g)\). The following assumption states my requirements on $Q_G$ using this notation.
\begin{assumption}\label{ass:Q_G}The distribution $Q_G$ is such that
\begin{enumerate}
    \item[(a)] $\{(\mathcal{M}_g, C_g, N_g): 1\leq g \leq G\}$ is an i.i.d. sequence of random variables.
    \item[(b)]  For some family of distributions $\{R(m,c,n):(m,c,n) \in \text{supp}(\mathcal{M}_g, C_g, N_g) \}$,
    \begin{equation*}
        R_G(\mathcal{M}^{(G)}, C^{(G)}, N^{(G)}) = \prod_{1\leq g \leq G} R(\mathcal{M}_g, C_g, N_g)~,
    \end{equation*}
    where $R(\mathcal{M}_g, C_g, N_g)$ denotes the distribution of $\left(\left(Y_{i,g}(z, h):z \in \{0, 1\}, h\in \mathcal{H}\right), X_{i,g}: 1 \leq i \leq N_g \right)$ conditional on $\{\mathcal{M}_g, C_g, N_g\}$.
    \item[(c)] $P\left\{|\mathcal{M}_g| \geq 2 \right\}=1$ and $E[N_g^2] < \infty$.
    \item[(d)] For some constant $C < \infty$, $P\left\{E[Y_{i,g}^2(z,h)\mid N_g, C_g, X_g] \leq C \text{ for all } 1\leq i \leq N_g \right\}=1$ for all $z \in \{0, 1\}$ and $h \in \mathcal{H}$ and $1 \leq g \leq G$.
    \item[(e)] $\mathcal{M}_g \perp \left(\left(Y_{i,g}(z, h):z \in \{0, 1\}, h\in \mathcal{H}\right): 1 \leq i \leq N_g\right) \mid C_g, N_g, X_g$ for all $1 \leq g \leq G$.
    \item[(f)] For all $z\in\{0,1\}, h\in\mathcal{H}$ and $1 \leq g \leq G$,
    \begin{equation*}
        E\left[\frac{1}{M_g} \sum_{i \in \mathcal{M}_g} Y_{i,g}(z,h) \mid N_g \right] = E\left[\frac{1}{N_g}\sum_{1\leq i \leq N_g} Y_{i,g}(z, h) \mid N_g \right] \text{w.p.1}~.
    \end{equation*}
\end{enumerate}
\end{assumption}
The sampling procedure of a cluster randomized experiment used in this paper closely follows that formalized by \cite{mp-cluster} and \cite{Bugni2022}. Assumption \ref{ass:Q_G} is essentially the same as Assumption 2.2 in \cite{Bugni2022}, which formalizes the sampling procedure of i.i.d. clusters (Assumptions \ref{ass:Q_G} (a)-(b)) and imposes mild regularity conditions (Assumptions \ref{ass:Q_G} (c)-(d)). Furthermore, Assumption \ref{ass:Q_G} (e) accommodates a second-stage sampling process within a given cluster that may depend on cluster-level and individual-level covariates as well as cluster sizes. This flexibility permits $\mathcal{M}_g$ to be potentially determined through stratified sampling, as discussed in \cite{cytrybaum2021}. Finally, Assumption \ref{ass:Q_G} (f) is a high-level assumption that ensures the extrapolation from the observations that are sampled to those that are not sampled.

\subsection{Parameters of Interest and Estimators}\label{subsec:estimand-estimator}

In the context of the sampling framework described above, this paper considers four parameters of interest, including primary and spillover effects that are equally or (cluster) size-weighted. For different choices of (possibly random) weights $\omega_g$, $1\leq g\leq G$ satisfying $E[\omega_g]=1$, we define the average \textit{primary effects} and \textit{spillover effects} under general weights as follows.
\begin{definition}
Define the weighted average primary effect under weight $\omega_g$ as follows:
\begin{equation}
    \theta^{P}_\omega(Q_G) := E\left[ \omega_g\left( \frac{1}{N_g} \sum_{1\leq i \leq N_g} Y_{i,g}(1,\pi_2) - Y_{i,g}(0,0)\right)\right]~,
\end{equation}
and the weighted average spillover effect as:
\begin{equation}
    \theta^S_\omega(Q_G) := E\left[\omega_g\left( \frac{1}{N_g} \sum_{1\leq i \leq N_g} Y_{i,g}(0,\pi_2) - Y_{i,g}(0,0)\right)\right]~.
\end{equation}
\end{definition}

Denote by $\theta^{P}_1(Q_G)$ and $\theta^{S}_1(Q_G)$ the equally-weighted cluster-level average primary and spillover effects with $\omega_g=1$, and $\theta^{P}_2(Q_G)$ and $\theta^{S}_2(Q_G)$ the size-weighted cluster-level average primary and spillover effects with $\omega_g=N_g/E[N_g]$. The consideration of weighted estimands is motivated by the \textit{non-ignorability} of cluster sizes. According to \cite{Bugni2022}, cluster sizes are considered \textit{ignorable} if the individual-level average treatment effect is independent of the cluster size. Formally, this is expressed as:
\begin{equation}\label{eqn:ignorable}
    P\{E[Y_{i,g}(z,\pi_2) - Y_{i,g}(0,0)\mid N_g] = E[Y_{i,g}(z,\pi_2) - Y_{i,g}(0,0)] \text{ for all } 1\leq i \leq N_g\} = 1
\end{equation}
for all \(1\leq g \leq G\) and \(z \in \{0, 1\}\). Cluster sizes are \textit{non-ignorable} whenever (\ref{eqn:ignorable}) is not satisfied. When cluster sizes are non-ignorable, different weights can lead to distinct parameters. The selection between these two types of estimands—equally weighted or size-weighted—depends on the analytical focus: whether the primary interest is on the clusters themselves or the individuals within these clusters. For instance, in assessing the impact of an educational program on students' academic performance, if policymakers are concerned with improvements at the school level, equally weighted estimands are appropriate. Conversely, if the focus is on student-level outcomes, then size-weighted estimands become relevant. 

The primary effects $\theta^{P}_1(Q_G)$ and $\theta^{P}_2(Q_G)$ are the differences in the averaged potential outcomes of treated units from treated clusters and control units from control clusters. In contrast, the spillover effects $\theta^{S}_1(Q_G)$ and $\theta^{S}_2(Q_G)$ are the differences in the averaged potential outcomes of control units from treated clusters and control units from control clusters. In many empirical settings, the estimation and comparison of primary and spillover effects play a crucial role in addressing important research questions \citep[see for example][]{Duflo2003}.

In summary, the formulas for the four parameters of interest are listed in Table \ref{table:estimands}. These estimands have been proposed and studied in previous literature \citep[see, e.g.,][]{hudgens2008, toulis2013, basse2018, Imai2021}, but mostly in a finite population framework. This paper adopts the terminology ``primary'' and ``spillover'' effects from \cite{basse2018}, which are respectively referred to as ``total'' and ``indirect'' effects in \cite{hudgens2008}. Previous works on interference have also studied other estimands, such as direct effects and overall effects \citep[see, e.g.,][]{hudgens2008,Wager2021,Imai2021}, but I do not explore these estimands further in this paper.

    \begin{table}[ht!]
    \centering
    \setlength{\tabcolsep}{8pt}
    \begin{tabular}{ll}
    \toprule
    Parameter of interest & Formula \rule[2ex]{0pt}{-2ex} \\ \midrule
    Equally-weighted primary effect      &   $ \theta^{P}_1(Q_G) := E\left[  \frac{1}{N_g} \sum_{1\leq i \leq N_g} Y_{i,g}(1,\pi_2) - Y_{i,g}(0,0)\right] $    \rule[-3.5ex]{0pt}{2ex}  \\
    Equally-weighted spillover effect & $ \theta^{S}_1(Q_G) := E\left[  \frac{1}{N_g} \sum_{1\leq i \leq N_g} Y_{i,g}(0,\pi_2) - Y_{i,g}(0,0)\right] $ \rule[-3.5ex]{0pt}{2ex} \\
    Size-weighted primary effect    & $ \theta^{P}_2(Q_G) := E\left[  \frac{1}{E[N_g]} \sum_{1\leq i \leq N_g} Y_{i,g}(1,\pi_2) - Y_{i,g}(0,0)\right] $ \rule[-3.5ex]{0pt}{2ex} \\
    Size-weighted spillover effect      &   $ \theta^{S}_2(Q_G) := E\left[  \frac{1}{E[N_g]} \sum_{1\leq i \leq N_g} Y_{i,g}(0,\pi_2) - Y_{i,g}(0,0)\right] $ \\
    \bottomrule
    \end{tabular}
    \caption{Parameters of interest}
    \label{table:estimands}
    \end{table}

For estimating the four parameters of interest, I propose the following estimators analogous to the difference-in-“average of averages” estimator in \cite{Bugni2022}:
\begin{align*}
    \hat{\theta}^P_1 &= \frac{1}{G_T}\sum_{1\leq g \leq G} I\{H_g = \pi_2 \} \bar{Y}_{g}^{1}  - \frac{1 }{G_C} \sum_{1\leq g \leq G} I\{H_g = 0 \} \bar{Y}_{g}^{1} \\
    \hat{\theta}^S_1 &= \frac{1}{G_T}\sum_{1\leq g \leq G} I\{H_g = \pi_2 \} \bar{Y}_{g}^{0}  - \frac{1 }{G_C} \sum_{1\leq g \leq G} I\{H_g = 0 \} \bar{Y}_{g}^{0} \\
    \hat{\theta}^P_2 &= \frac{1}{N_T} \sum_{1\leq g \leq G} I\{H_g = \pi_2 \} N_g \bar{Y}_{g}^1  - \frac{1}{N_C} \sum_{1\leq g \leq G}I\{H_g = 0 \}N_g \bar{Y}_{g}^1  \\
    \hat{\theta}^S_2 &= \frac{1}{N_T} \sum_{1\leq g \leq G} I\{H_g = \pi_2 \} N_g \bar{Y}_{g}^{0}  - \frac{1 }{N_C} \sum_{1\leq g \leq G} I\{H_g = 0 \}N_g \bar{Y}_{g}^0~,
\end{align*}
% where $G_1 = \sum_{1\leq g \leq G} I\{H_g = \pi_2\}$, $G_0 = \sum_{1\leq g \leq G} I\{H_g = 0\}$, and $N_1 = \sum_{1\leq g \leq G} I\{H_g = \pi_2 \} N_g$, $N_0 = \sum_{1\leq g \leq G} I\{H_g = 0 \} N_g$ and
where $G_T = \sum_{1\leq g \leq G} I\{H_g = \pi_2\}$, $G_C = \sum_{1\leq g \leq G} I\{H_g = 0\}$, and 
$N_T = \sum_{1\leq g \leq G} I\{H_g = \pi_2 \} N_g$, $N_C = \sum_{1\leq g \leq G} I\{H_g = 0 \} N_g$ and
\begin{align*}
    \bar Y_g^z &= \frac{1}{M_g^z} \sum_{i\in\mathcal{M}_g} Y_{i,g}  I\{H_g = \pi_2, Z_{i,g}=z\} + \frac{1}{M_g} \sum_{i\in\mathcal{M}_g} Y_{i,g} I\{H_g =0\} ~,
\end{align*}
where $M_g^z = \sum_{i\in\mathcal{M}_g} I\{ Z_{i,g} = z \}$ with $z \in \{0,1\}$.

By definition, the ``first/individual average'' $\bar{Y}_{g}^{1}$ from the primary effect estimator is taken over all treated units within the $g$-th cluster if the cluster is treated, and all control units within the $g$-th cluster if the cluster is assigned to control. When it comes to estimating spillover effects, the ``first/individual average'' $\bar{Y}_{g}^{0}$ is taken over all control units within the $g$-th cluster if the cluster is treated, and all control units within the $g$-th cluster if the cluster is assigned to control. Then, the ``second/cluster average'' is a cluster-level average of $\bar{Y}_{g}^{1}$ or $\bar{Y}_{g}^{0}$ taken within groups of treated and untreated clusters as featured in a usual difference-in-means estimator. 

The proposed estimators can be obtained from ordinary least squares regressions using different weighting schemes. Let \( L_{i,g} = I\{H_g = \pi_2\}(1 - Z_{i,g}) \) denote the indicator for untreated units within treated clusters. Consider the following linear model for an ordinary least squares regression:
\begin{equation}\label{eqn:ols-full}
    Y_{i,g} = \alpha + \beta_1 Z_{i,g} + \beta_2 L_{i,g} + \epsilon_{i,g}~.
\end{equation}
Note that the estimators \(\hat{\theta}_1^P\) and \(\hat{\theta}_1^S\) may be obtained by estimating coefficients \(\beta_1\) and \(\beta_2\) from a weighted least squares regression of equation (\ref{eqn:ols-full}) using weights \(1/M_g\). Similarly, \(\hat{\theta}_2^P\) and \(\hat{\theta}_2^S\) may be derived using weights \(N_g/M_g\) (see Appendix \ref{sec:app-wols} for formal derivations). Moreover, the unweighted least squares regression produces the ``sample'' size-weighted estimators. Taking \(\beta_1\) as an example:
\begin{align*}
    \hat{\beta}_1 = \frac{1}{M_1} \sum_{1 \leq g \leq G} I\{H_g = \pi_2\} M_g \bar{Y}_{g}^1  - \frac{1}{M_0} \sum_{1 \leq g \leq G} I\{H_g = 0\} M_g \bar{Y}_{g}^1~,
\end{align*}
where \( M_1 = \sum_{1 \leq g \leq G} I\{H_g = \pi_2\} M_g \), and \( M_0 = \sum_{1 \leq g \leq G} I\{H_g = 0\} M_g \). These ``sample'' size-weighted estimators are identical to \(\hat{\theta}_2^P\) and \(\hat{\theta}_2^S\) when outcomes of all units from each cluster are observed or the number of observed units is proportional to the cluster size, i.e., \( M_g / N_g = c \) for \( 0 < c \leq 1 \).

My estimators are closely related to those studied in previous methodological literature. For example, equally-weighted estimators $\hat{\theta}^P_1$ and $\hat{\theta}^S_1$ are identical to the household-weighted estimators from \cite{basse2018}, which are closely related to the estimators in \cite{hudgens2008}. $\hat{\theta}^P_1$ and $\hat{\theta}^S_1$ may also be obtained through the ``household-level regression'' proposed in \cite{basse2018}, which is equivalent to running two separate ordinary least squares regressions of $\bar Y_{g}^1$ on a constant and $I\{H_g=\pi_2\}$, and $\bar Y_{g}^0$ on a constant and $I\{H_g=0\}$. Size-weighted estimators $\hat{\theta}^P_2$ and $\hat{\theta}^S_2$ are closely related to the individual-weighted estimator proposed by \cite{basse2018}. In previous studies such as \cite{basse2018}, \cite{VAZQUEZBARE2022}, and \cite{Gonzalo-two-stage}, researchers have investigated estimators obtained through a widely used saturated regression in multi-treatment experiments, similar to the least squares regression described by equation (\ref{eqn:ols-full}).

In empirical literature, various regression estimators are used for estimating primary and spillover effects. One widely used estimator is described in equation (\ref{eqn:ols-full}) \citep[see, e.g.,][]{Haushofer2016,Haushofer2019}. Another estimator that produces the same set of estimators is through the alternative regression $Y_{i,g} = a + b_1 Z_{i,g} + b_2 I\{H_g=\pi_2\} + u_{i,g}$ \citep[see, e.g.,][]{Duflo2003,Ichino2012},
where the estimators are related to those from (\ref{eqn:ols-full}) as follows: $\hat\beta_1 = \hat b_1 + \hat b_2$ and $\hat \beta_2 = \hat b_2$. Some empirical works use either or both of the two separate regressions: $Y_{i,g} = \alpha + \beta_1 Z_{i,g} + \epsilon_{i,g}$ and $Y_{i,g} = \alpha + \beta_2 L_{i,g} + \epsilon_{i,g}$ \citep[see, e.g.,][]{Beuermann2015,Hidrobo2016,Aramburu2019}. In many cases, estimators obtained from regressions with fixed effects are reported along with those without fixed effects \citep[see, e.g.,][]{Ichino2012}. Section \ref{sec:simulation-inference} will examine the validity of statistical tests based on regressions with and without fixed effects.

\section{Main Results}\label{sec:main-results}
In this section, I investigate the asymptotic properties of the estimators presented in Section \ref{subsec:estimand-estimator} within a finely stratified randomization framework. Specifically, in the first stage, clusters are partitioned into a large number of small strata of a fixed size, with the assignment mechanism being a completely randomized design (also known as a permuted block design) independently applied within each stratum. Formally, consider \( n \) strata of size \( k \) (each stratum consisting of \( k \) clusters), formed by matching clusters according to a function \( S:\text{supp}((C_g,N_g)) \rightarrow \mathbf{R}^\ell \). Denote by \( S^{(G)}=(S_1, \dots, S_G) \) the vector of variables used for matching, where \( S_g = S(C_g, N_g) \). Within each stratum, \( l \) clusters are randomly selected and assigned to the treatment group.\footnote{Extending the setup to a more general framework with varying stratum sizes and heterogeneous treatment fractions is indeed possible; see Section 3.2 of \cite{cytrybaum2021}.} Specifically, \( G = nk \) and \( \pi_1 = l/k \), where \( 0 < l < k \), and \( l \) and \( k \) are mutually prime. Furthermore, I consider a second-stage stratification on units from a given cluster. Denote by \( B_g = (B_{i,g}: 1 \leq i \leq N_g) \) the vector of strata on units in the \( g \)th cluster, constructed from observed baseline covariates \( X_{i,g} \) for the \( i \)th unit using a function \( B:\text{supp}(X_{i,g}) \rightarrow \mathcal{B}_g \).\footnote{Asymptotics are not considered in the second-stage design; thus, the second stage could employ finely stratified designs like matched-pair, or those with coarse stratification such as stratified block randomization.}

\begin{example}
    \cite{Duflo2003} conducted such a finely stratified experiment involving 330 university departments, each averaging 30 staff employees. In the first stage, these departments (clusters with an average size of 30) were grouped into triplets (small strata of size 3) based on their cluster-level covariates. Within each triplet, two departments were randomly chosen to be part of the treated group. Specifically, this design has \( k=3 \), \( l=2 \), \( \pi_1 = 2/3 \), and \( G=3n = 330 \). In the second stage, individuals from these treated departments were randomly selected to receive treatments.
\end{example}

To start with, I describe my assumptions on the treatment assignment mechanism in the first stage. Formally, let
\begin{equation*}
    \lambda_j = \lambda_j(S^{(G)}) \subseteq \{1,\dots, G\},~ 1\leq j \leq n
\end{equation*}
denote $n$ sets each consisting of $k$ elements that form a partition of $\{1,\dots, G\}$.
%Let $r = r_n(S^{(G)})$ be a permutation of $\{1,\dots,nk\}$ potentially dependent on $S^{(G)}$. The $n$ blocks are represented by
%\begin{equation}\label{eqn:blocks-lambdaj}
%     \lambda_j:=\{r((j - 1)k+1),\dots, r(kj)\} ~.
%\end{equation}

I assume treatment status is assigned to clusters as follows:
\begin{assumption}\label{ass:assignment-match}
Treatments are assigned so that $W^{(G)} \independent H^{(G)} | S^{(G)}$ and, conditional on $S^{(G)}$,
\[ \{(I\{H_{i}=\pi_2\}): i \in \lambda_j): 1 \leq j \leq n\} \]
are i.i.d.\ and each uniformly distributed over all permutations of $\left\{z\in \{0,1\}^{k}: \sum_{z=1}^k z_{i}= l \right\}$.
\end{assumption}
Assumption \ref{ass:assignment-match} formally describes the assignment mechanism of a two-stage experiment with finely stratified randomization in the first stage. Further, units in each pair are required to be ``close'' in terms of their stratification variable $S_g$ in the following sense:
\begin{assumption} \label{as:close}
The strata used in determining treatment status satisfy
\[ \frac{1}{n} \sum_{1 \leq j \leq n} \max_{i,k\in\lambda_j} \|S_{i} - S_{k}\|^2 \stackrel{P}{\to} 0~. \]
\end{assumption}
The validity of the variance estimators relies on the following condition that the distances between units in adjacent blocks are considered ``close'' in relation to their baseline covariates:
\begin{assumption}\label{ass:variance-estimator}
The strata used in determining treatment status satisfy
        \[ \frac{1}{n} \sum_{1 \leq j \leq \lfloor n / 2 \rfloor } \max_{i\in\lambda_{2j-1}, k\in \lambda_{2j}} \|S_{i} - S_{k}\|^2 \stackrel{P}{\to} 0~. \]
\end{assumption}

\begin{remark}\label{remark:finely-stratified-framework}
    Following \cite{cytrybaum2021}, Assumptions \ref{ass:assignment-match}-\ref{ass:variance-estimator} underpin the finely stratified randomization framework that enables unified asymptotics and inference for a wide variety of different designs. When \( S_g \) is continuous, this framework aligns with \textit{matched tuples designs}, essentially generalized matched-pair designs. Blocking algorithms that satisfy Assumptions \ref{ass:assignment-match}-\ref{ass:variance-estimator} have been thoroughly discussed in recent literature on matched pairs/tuples designs \citep[see, for example,][]{bai-inference, Bai-optimal, matched-tuple, cytrybaum2021}.\footnote{For instance, when $\text{dim}(S_g) = 1$ and clusters/units are matched into blocks by ordering them according to the values of $S_g$ and grouping the adjacent clusters/units, Theorem 4.1 of \cite{bai-inference} shows that Assumptions \ref{ass:assignment-match}-\ref{ass:variance-estimator} are satisfied as long as $E[S_g^2] < \infty$.} When \( S_g \) is categorical, the framework corresponds to \textit{stratified block randomization (SBR)}. Intuitively, consider a finely stratified design that first stratifies on \( S_g \) and then groups clusters arbitrarily into small strata of size \( k \) within each large stratum. Such a design is referred to as ``coarse stratification'' by \cite{cytrybaum2021} and is shown to be equivalent to SBR in Proposition 9.15 of \cite{cytrybaum2021_old}. When \( S_g \) is a constant and thus provides no information at all, the framework equates to a \textit{completely randomized design}.
\end{remark}

The next step is to formalize the assumption of independence between the first and second stage designs. To begin with, I utilize the notation $\{Z_{i,g}(h): h \in \mathcal{H}\}$, representing the ``potential treatment'' for various treated fractions $h\in\mathcal{H}$, and relate the (observed) individual treatment indicator and potential individual-level treatment indicator as follows:
\begin{equation}\label{eqn:potential_treatment}
    Z_{i,g} = \sum_{h \in \mathcal{H}}  Z_{i,g}(h) I\{H_g=h\} \text{ for } 1\leq i \leq N_g~.
\end{equation}
The underlying motivation for this ``potential outcome style'' notation becomes evident when considering that in two-stage experiments, the realized treatment assignment in the first stage is almost always correlated with that in the second stage (e.g., $H_g = \frac{1}{N_g}\sum_{1\leq i \leq N_g} Z_{i,g}$). Yet, the ``potential'' individual-level treatment assignment, for any specified target treated fraction, can be independent of the cluster-level assignment of that target treated fraction. This is similar to the classic potential outcome model, where treatment assignment is independent of potential outcomes but likely correlates with observed outcomes.

Then, my requirements on the treatment assignment mechanism for the second stage are summarized in the following assumption:
\begin{assumption}\label{as:assignment2}
The treatment assignment mechanism for the second-stage is such that
\begin{itemize}
    \item[(a)] $(((Z_{i,g}(h):h\in\mathcal{H}): 1\leq i\leq N_g): 1\leq g \leq G) \perp H^{(G)}$,
    \item[(b)] $W^{(G)} \perp (((Z_{i,g}(h):h\in\mathcal{H}): 1\leq i\leq N_g): 1\leq g \leq G) \mid(B_g: 1\leq g \leq G)$,
    \item[(c)] For all $1\leq g\leq G$, $E[Z_{i,g}(h)\mid B_g] = \frac{1}{M_g} \sum_{i \in \mathcal{M}_g} Z_{i,g}(h) = h + o_P(1) $.
\end{itemize}
\end{assumption}
Assumption \ref{as:assignment2} (a) rules out any confounders between the first-stage and second-stage treatment assignments, which is typically satisfied in most two-stage experiments. Assumption \ref{as:assignment2} (b) is analogous to Assumption \ref{as:assignment1} (a). Assumption \ref{as:assignment2} (c) requires that the marginal assignment probability for each stratum and the realized treated fraction in the observed subset of units both equal the intended treated fraction 
$h$, up to a finite sample error that diminishes as cluster size increases.\footnote{In the proof of the main results, I only need \( E[Z_{i,g}(h) \mid B_g] = \frac{1}{M_g} \sum_{i \in \mathcal{M}_g} Z_{i,g}(h) \) to hold for unbiasedness. However, in practice, these two quantities, along with the treated fraction for the entire cluster, need to align with the intended treated fraction \( h \) so that they are consistent with the notations of the potential outcomes \( Y_{i,g}(z,h) \).} An example of this could be (individual-level) stratified block randomization, where the treated fraction remains constant across all strata, with observed units drawn from a random subset of these strata.

Finally, I impose the following assumption on $Q_G$ in addition to Assumption \ref{ass:Q_G}:
\begin{assumption} \label{as:Q_G-lip}
The distribution $Q_G$ is such that
\begin{enumerate}
%\item $E[N_g | S_g = s]$ are Lipschitz in $s$.
\item[(a)] $E[\bar{Y}_g^r(z,h) N_g^\ell | S_g = s]$ is Lipschitz in $s$ for $(z,h) \in \mathcal \{(0,0),(0,\pi_2),(1,\pi_2)\}$ and $r,\ell\in\{0,1,2\}$.
\item[(b)] For some $C < \infty$, $P\{E[N_g^2|S_g] \le C\} = 1$
%\item $E[\bar{Y}_g^2(z,h) N_g^2 ] < \infty$.
\end{enumerate}
\end{assumption}
Assumption \ref{as:Q_G-lip}(a) is a smoothness requirement analogous to Assumption 3(ii) in \cite{Bai-optimal} ensuring that units within clusters which are “close” in terms of their baseline covariates are suitably comparable. Assumption \ref{as:Q_G-lip}(b) imposes an additional restriction on the distribution of cluster sizes beyond what is stated in Assumption \ref{ass:Q_G}(c).

The following theorem derives the asymptotic behavior of estimators for equally-weighted and size-weighted effects.\footnote{Throughout the paper, $V_1(1)$ and $V_2(1)$ denote the variances of primary effects, while $V_1(0)$ and $V_2(0)$ represent the variances of spillover effects. In other words, the notation $z\in \{0, 1\}$ (as in $V_1(z)$) represents the individual's own treatment status.}

\begin{theorem}\label{thm:matched-group}
Suppose Assumption \ref{as:interference} holds, $Q_G$ satisfies Assumptions \ref{ass:Q_G} and \ref{as:Q_G-lip} and the treatment assignment mechanism satisfies Assumptions \ref{ass:assignment-match}-\ref{as:close} and \ref{as:assignment2}. Then, as $n \to \infty$,
\begin{align}
    &\sqrt{G} \left( \hat\theta^{P}_1 - \theta^P_1(Q_G) \right) \xrightarrow{d} \mathcal{N}(0, V_1(1))~, \\ &\sqrt{G} \left( \hat\theta^{S}_1 - \theta^S_1(Q_G) \right) \xrightarrow{d} \mathcal{N}(0, V_1(0))~,\\
    &\sqrt{G} \left( \hat\theta^{P}_2 - \theta^P_2(Q_G) \right) \xrightarrow{d} \mathcal{N}(0, V_2(1))~, \\ &\sqrt{G} \left( \hat\theta^{S}_2 - \theta^S_2(Q_G) \right) \xrightarrow{d} \mathcal{N}(0, V_2(0))~,
\end{align}
where, for $z\in \{0,1\}$,
\begin{align}\label{eqn:V1}
    \begin{split}
    V_1(z) &= \frac{1}{\pi_1} \operatorname{Var}\left[\bar{Y}_{g}(z,\pi_2)\right] + \frac{1}{1-\pi_1} \operatorname{Var}\left[\bar{Y}_{g}(0,0)\right]\\
    &\quad\quad\quad - \pi_1(1-\pi_1)  E\left[\left( \frac{1}{\pi_1}m_{z,\pi_2}\left(S_{g}\right) + \frac{1}{1-\pi_1}m_{0,0}\left(S_{g}\right) \right)^2\right]\\
    \end{split}
\end{align}
and
\begin{align}\label{eqn:V2}
    \begin{split}
    V_2(z) &= \frac{1}{\pi_1} \var[\tilde Y_g(z,\pi_2)] + \frac{1}{1-\pi_1} \var[\tilde Y_g(0,0)]\\
    &\quad\quad\quad - \pi_1(1-\pi_1) E\left[\left(\frac{1}{\pi_1}E[\tilde Y_g(z,\pi_2)\mid S_g] + \frac{1}{1-\pi_1}E[\tilde Y_g(0,0)\mid S_g] \right)^2\right]
    \end{split}
\end{align}
with
\begin{align}
    \bar{Y}_g(1,\pi_2) &= \frac{1}{M_g^{1}} \sum_{i \in \mathcal{M}_g} Y_{i,g}(1,\pi_2)Z_{i,g}(\pi_2)\\
    \bar{Y}_g(0,\pi_2) &= \frac{1}{M_g^{0}} \sum_{i \in \mathcal{M}_g} Y_{i,g}(0,\pi_2)(1-Z_{i,g}(\pi_2))\\
    \bar{Y}_g(0,0) &= \frac{1}{M_g} \sum_{i \in \mathcal{M}_g} Y_{i,g}(0,0)\\
    m_{z,h}\left(S_{g}\right) &= E[\bar{Y}_g(z,h)\mid S_g] - E[\bar{Y}_g(z,h)] \label{eqn:m(S)}
\end{align}
and 
\begin{align}\label{eqn:Y(z,h)}
    & \tilde Y_g(z,h) = \frac{N_g}{E[N_g]} \left ( \bar Y_g(z,h) - \frac{E[\bar Y_g(z,h) N_g]}{E[N_g]} \right )
\end{align}
for $(z,h) \in \{(1,\pi_2),(0,\pi_2),(0,0)\}$.
\end{theorem}

\begin{remark}
Note that the asymptotic variance $V_2(z)$ has the same form as $V_1(z)$, with $\tilde Y_g(z,h)$ replacing $\bar Y_g(z,h)$. Intuitively, $\tilde Y_g(z,h)$ is a demeaned and cluster size weighted version of $\bar Y_g(z,h)$. Moreover, $V_1(z)$ and $V_2(z)$ correspond exactly to the asymptotic variance of the difference-in-means estimator for matched-pair experiments with individual-level ``one-stage'' assignment, as in \cite{bai-inference} and \cite{Bai-optimal}. Additionally, $V_2(z)$ has a similar form to the asymptotic variance in a cluster randomized trial with matched pairs, as derived in \cite{mp-cluster}. In fact, when $\pi_1=1/2$ and $\pi_2 = 1$, my result collapses exactly to theirs.

In a special case where covariate information is not used to construct strata, the asymptotic variance of the estimated equally-weighted effects can be expressed as follows:
\begin{equation}\label{eqn:variance-no-covariate}
    V_1(z) = \frac{1}{\pi_1} \operatorname{Var}\left[\bar{Y}_g(z,\pi_2)\right]+\frac{1}{1-\pi_1} \operatorname{Var}\left[\bar{Y}_{g}(0,0) \right]~,
\end{equation}
which is equivalent to the identifiable parts of the variance derived in \cite{basse2018} under the finite population framework. The asymptotic variance of partial population designs from \cite{Gonzalo-two-stage} is also closely related to (\ref{eqn:variance-no-covariate}) under binary settings. Specifically, \cite{Gonzalo-two-stage} provides an alternative expression of $\operatorname{Var}\left[\bar{Y}_g(z,\pi_2)\right]$ with intra-cluster variances and correlations. Therefore, inference methods based on (\ref{eqn:variance-no-covariate}), including \cite{basse2018} and \cite{Gonzalo-two-stage}, are generally conservative under covariate-adaptive randomization.
\end{remark}

\begin{remark}\label{remark:how-second-stage-affects-variance}
    It's worth noting that the setup of the first-stage design has a clear impact on the asymptotic variance $V_1(z)$, as evidenced by the third term in equation (\ref{eqn:V1}). Furthermore, the second-stage design also influences the asymptotic variance $V_1(z)$, albeit more implicitly, via the distribution of $Z_{i,g}(\pi_2)$. Specifically, the first term in equation (\ref{eqn:V1}) depends on $\operatorname{Var}\left[\bar{Y}_{g}(z,\pi_2)\right]$, which is directly tied to the second-stage design. In contrast, the second and third terms do not depend on $Z_{i,g}(\pi_2)$ (for further details, see Remark \ref{remark:how-second-stage-affects-variance-app}). Thus, the efficacy of designing the first stage versus the second stage can be disentangled into distinct components. More importantly, as I show in Appendix \ref{subsubsec:variance-improvement-bound}, if there exists $M>0$ such that $M_g \geq M$ for all $g$, then the effect of second-stage designs on $V_1(z)$ is $O(1/M)$, while the effect of the first-stage is $O(1)$. This characterization could be beneficial for practitioners seeking to assess the relative importance of first-stage design versus second-stage design in optimizing efficiency gains. If possible, a calibrated simulation study using pilot or observational data can be used to estimate the relative efficiency gain obtained at each stage.
\end{remark}

Theorem \ref{thm:matched-group} implies that covariate information is important to establish asymptotically exact inference for the four estimands of interest under covariate-adaptive randomization. Many empirical studies rely on statistical inference based on the regression in equation (\ref{eqn:ols-full}) with HC2 cluster-robust standard errors. While this procedure is also proposed in \cite{basse2018} and \cite{Gonzalo-two-stage}, the regression coefficients it produces generally do not provide consistent estimates for the estimands in Table \ref{table:estimands}. As discussed in Section \ref{subsec:estimand-estimator}, if all units in each cluster are sampled ($N_g = M_g$) or the number of sampled units is proportional to cluster size ($M_g/N_g=c$ for $0< c < 1$), this procedure yields consistent point estimates for size-weighted effects but may still be conservative (see Appendix \ref{sec:app-wols}). Therefore, I aim to develop asymptotically exact inference methods based on my theoretical results.

To begin with, I introduce consistent variance estimators for the asymptotic variances from Theorem \ref{thm:matched-group}. To estimate $V_1(z)$ and $V_2(z)$, I follow the construction of ``pairs of pairs'' in \cite{bai-inference} and \cite{matched-tuple}, and replace the individual outcomes with the averaged outcomes $\bar Y_g^{z}$ (as defined in Section \ref{subsec:estimand-estimator}) and adjusted averaged outcomes $\tilde{Y}_g^z$ , respectively. The definition of the adjusted average outcomes is given as follows:
\begin{align*}
   \tilde{Y}_g^z = \frac{N_{g}}{ \frac{1}{G}\sum_{1\leq g \leq G} N_g }\left(\bar{Y}_{g}^z-\frac{ \frac{1}{G_g} \sum_{1\leq j \leq G} \bar{Y}_{j}^z I\{H_g = H_j\} N_j}{\frac{1}{G} \sum_{1\leq j \leq G} N_{j}} \right)~,
\end{align*}
where $G_g = \sum_{1\leq j \leq G} I\{H_g = H_j\}$. Here, I present the construction of variance estimator $\hat V_1(z)$ for $V_1(z)$. Similarly, $\hat V_2(z)$ can be constructed by simply replacing $\bar Y_{g}^z$ with $\tilde Y_{g}^z$, and thus details are omitted. Let $\hat \Gamma^z_n(h) = \frac{1}{n k(h)} \sum_{1\leq g \leq G} \bar Y_g^z I \{H_g = h\}$ where $k(h) = \sum_{i\in\lambda_j} I\{H_i=h\}$ denotes the number of units under assignment $H_i = h$ in the $j$-th strata. In the setup of binary treatment, it becomes that $k(\pi_2)=l$ and $k(0)=k-l$. Finally, my estimator for $V_1(z)$ is then given by
\begin{align}\label{eqn:Vhat1}
    \hat V_1(z) = \frac{1}{\pi_1} \hat{\mathbb{V}}_{1,n}^z(\pi_2) + \frac{1}{1-\pi_1} \hat{\mathbb{V}}_{1,n}^z(0) + \hat{\mathbb{V}}_{2,n}^z(\pi_2, \pi_2) + \hat{\mathbb{V}}_{2,n}^z(0, 0) - 2 \hat{\mathbb{V}}_{2,n}^z(\pi_2, 0)
\end{align}
with
\begin{align*}
&\hat{\mathbb{V}}_{1,n}^z(h)= \hat{\mathbb{E}}\left[\operatorname{Var}\left[\bar Y_{g}(z,h) \mid S_g \right]\right] := \left(\hat\sigma^{z}_n(h)\right)^2 - (\hat{\rho}_n^z(h) - (\hat{\Gamma}_n^z(h))^2 )\\
&\hat{\mathbb{V}}_{2,n}^z(h, h') = \hat{\operatorname{Cov}}\left[ E\left[\bar Y_{g}(z,h) \mid S_g\right], E\left[\bar Y_{g}(z, h') \mid S_g\right] \right]  := \hat{\rho}_n^z(h, h') - \hat{\Gamma}_n^z(h)\hat{\Gamma}_n^z(h')~,
\end{align*}
where 
\begin{align*}
\hat \rho_n^z(h) &:= \frac{2}{n} \sum_{1 \leq j \leq \lfloor n / 2 \rfloor} \frac{1}{k^2(h)}\Big ( \sum_{i \in \lambda_{2j-1}} \bar Y_i^z I \{H_i = h\} \Big ) \Big ( \sum_{i \in \lambda_{2j}} \bar Y_i^z I \{H_i = h\} \Big )  \\
\hat \rho_n^z(h, h^\prime) &:= \frac{1}{n} \sum_{1 \leq j \leq n} \frac{1}{l(k-l)} \Big ( \sum_{i \in \lambda_j} \bar Y_i^z I \{H_i = h\} \Big ) \Big ( \sum_{i \in \lambda_j} \bar Y_i^z I \{H_i = h^\prime\} \Big )  \\
\left(\hat\sigma^{z}_n(h)\right)^2 &:= \frac{1}{n k(h)} \sum_{1 \leq g \leq G} (\bar Y_g^z - \hat \Gamma_n^z(h))^2 I \{H_g = h\}~.
\end{align*} 

Based on the variance estimators, I propose the ``adjusted'' $t$-test with the aforementioned variance estimators as my method of inference throughout the rest of the paper. As an example, the ``adjusted'' $t$-test for equally-weighted primary effect, i.e. $H_0: \theta_1^P(Q_G) = \theta_{0}$, is given by
\begin{equation}\label{eqn:adjust-t-test}
    \phi_G(V^{(G)})=I\left\{\left|\sqrt{G}\left(\hat{\theta}_1^P-\theta_{0}\right)/\hat{V}_1(1)\right|>z_{1-\frac{\alpha}{2}}\right\}~,
\end{equation}
where $z_{1-\frac{\alpha}{2}}$ represents $1-\frac{\alpha}{2}$ quantile of a standard normal random variable.

The subsequent analysis yields consistency results for the estimators $\hat V_1(z)$ and $\hat V_2(z)$ and validity results for the adjusted $t$-test:
\begin{theorem}\label{thm:variance-estimator-mt}
Suppose Assumption \ref{as:interference} holds, \(Q_G\) satisfies Assumptions \ref{ass:Q_G} and \ref{as:Q_G-lip}, and the treatment assignment mechanism satisfies Assumptions \ref{ass:assignment-match}-\ref{as:assignment2}. Then, as \(n \to \infty\), \(\hat V_1(z) \xrightarrow{P} V_1(z)\) and \(\hat V_2(z) \xrightarrow{P} V_2(z)\) for \(z \in \{0,1\}\). As a consequence, provided that \(V_1(z) > 0\) and \(V_2(z) > 0\) for \(z \in \{0,1\}\), \(\phi_G(V^{(G)})\) satisfies
\begin{equation*}
    \lim_{n \rightarrow \infty} E\left[\phi_G(V^{(G)})\right] = \alpha ~,
\end{equation*}
under the null hypothesis and a significance level \(\alpha \in (0,1)\).
\end{theorem}

\begin{remark}\label{remark:large-strata-variance-estimator}
    Although the ``pairs of pairs'' variance estimators were initially developed for matched-pair designs (see \cite{abadie2008} and \cite{bai-inference}), they are consistent under both Stratified Block Randomization (SBR) and complete randomization. Nevertheless, it is beneficial to explore alternative variance estimators inspired by the ``large strata'' asymptotic framework proposed by \cite{Bugni2018}. This framework involves a fixed number of large strata where the number of units within each stratum grows indefinitely. Formal results for two-stage experiments under covariate-adaptive randomization using this framework are detailed in Appendix \ref{sec:large-strata}. Notably, I construct variance estimators \(\hat{V}_3(z)\) and \(\hat{V}_4(z)\) (detailed in equations (\ref{eqn:V3(z)hat}) and (\ref{eqn:V4(z)hat}) in Appendix \ref{sec:large-strata}) that are particularly well-suited for large strata experiments such as SBR (see Remark \ref{remark:compare-variance-app}). This framework also enables an analytical examination of a broader spectrum of experimental designs, including Efron's biased coin design and other sequential randomizations.
\end{remark}

Note that the variance estimator $\hat V_1(z)$ (or $\hat V_2(z)$) depends on the assignment mechanism in the first stage through the strata indicator $S_g$, but not on the assignment mechanism in the second stage. This means that valid statistical inference based on $\phi_G(V^{(G)})$ does not require knowledge of the assignment mechanism in the second stage. We can see this by observing that the first term in equations (\ref{eqn:V1}), which is the only term affected by the second-stage design, can be consistently estimated by the first term in equation (\ref{eqn:Vhat1}). My approach leverages the cluster-level averaged outcomes and benefits from large samples of clusters, without explicitly modeling intra-cluster correlations as done in the previous literature \citep[see, for example,][]{Gonzalo-two-stage}.

\section{Optimal Stratification for Two-stage Designs}\label{sec:optimality}
In this section, I introduce two optimality results related to two-stage randomized experiments, as discussed in Sections \ref{sec:main-results}. The first result provides insights into the optimal design for the initial stage, while the second addresses the optimal design for the second stage, taking into account additional assumptions about the assignment mechanism and covariance among unit outcomes within clusters. These findings indicate that particular finely stratified designs maximize statistical precision when estimating parameters outlined in Table \ref{table:estimands}.

First, I present a result that identifies the optimal functions for matching in the first-stage, targeting various parameters of interest.
\begin{theorem}\label{thm:optimal-first-stage}
$V_1(z)$ is minimized when $S_g=E\left[  \frac{\bar{Y}_g(z,\pi_2)}{\pi_1} + \frac{\bar{Y}_g(0,0) }{1-\pi_1}\mid C_g, N_g \right]$. Meanwhile, $V_2(z)$ is minimized when $S_g= E\left[\frac{\tilde Y_g(z,\pi_2)}{\pi_1} + \frac{\tilde Y_g(0,0)}{1-\pi_1} \mid C_g,N_g\right] $.
\end{theorem}

A direct implication of Theorem \ref{thm:optimal-first-stage} is that it characterizes the optimal functions to match on within the class of finely stratified designs. These functions are referred to as ``index function'' in \cite{Bai-optimal}. As noted in Remark \ref{remark:finely-stratified-framework}, when $S_g$ is categorical, finely stratified designs correspond to stratified block randomization, which implies that the optimal finely stratified designs is also asymptotically optimal among all large strata designs described in Appendix \ref{sec:large-strata}. It is important to note that when discussing the optimal design for the first stage, we are comparing different first-stage designs for any fixed second-stage design (and vice versa for the second-stage design optimality).

\begin{remark}\label{remark:first-stage-optimality}
    Based on the optimality results in Theorem \ref{thm:optimal-first-stage}, I recommend choosing covariates for matching based on the parameters of interest. For example, for the size-weighted estimands \(\theta_2^P\) and \(\theta_2^S\), matching on \(N_g\) is essential in experiments with widely varying cluster sizes, as this will be highly predictive of the scale of the outcome \(\tilde{Y}_g(z,h)\). When cluster covariates \(C_g\) are aggregations of individual-level covariates \(X_{ig}\), it may be beneficial to consider whether to match on averages or sums. For instance, matching on \(C_g = \sum_{i=1}^{N_g} X_{ig}\) might be appropriate for \(\theta_2^P\) and \(\theta_2^S\), while using the normalized mean \(C_g = N_g^{-1} \sum_{i=1}^{N_g} X_{ig}\) might be better for other estimands.
\end{remark}

The subsequent discussion examines the optimality of finely stratified designs in the second stage of the experiment. The second-stage randomization is formalized in the following assumptions.
\begin{assumption}\label{ass:assignment-crd}
For $1 \leq g \leq G$, units within a given stratum, denoted by $\lambda_b = \{i \in \mathcal{M}_g: B_i = b\}$ for $b \in \mathcal{B}$, are assigned with treatment $(Z_{i,g}(\pi_2): i \in \lambda_b)$ that is uniformly distributed over $\{z \in \{0,1\}^{|\lambda_b|}:\sum_{j\in\lambda_b} z_j =  \lfloor \pi_2|\lambda_b|\rfloor \}$ and i.i.d across $b \in \mathcal{B}$.
\end{assumption}

% \begin{definition}\label{def:matched-group}
%     A minimal matched tuples design that matches on a given function $g$ is defined as follows. Suppose there are $n k$ units along with observed covariates $(\chi_{i}:1\leq i\leq nk)$ and an irreducible treated fraction $\pi=\frac{l}{k}$. Order units using the value of $g_i=g(\chi_{i})$ for all $1\leq i \leq nk$ such that $g_{r(1)} \leq \dots \leq g_{r(nk)}$. Construct $n$ blocks as in (\ref{eqn:blocks-lambdaj}) and assignment treatments according to Assumption \ref{ass:assignment-match}.
% \end{definition}
Additionally, I assume that the covariance of outcomes between any pairs of units within a cluster is homogeneous. In other words, the covariance does not depend on the individual-level covariates of units in the same cluster. Formally, the assumption is stated as follows:
\begin{assumption}\label{ass:cond-indep}
For $z \in \{0,1\}$, $ 1 \leq i\neq j\leq N_g$, 
\begin{equation}
    \operatorname{Cov}\left[ Y_{i,g}(z,\pi_2), Y_{j,g}(z,\pi_2)\mid (X_{i,g}:1\leq i \leq N_g)\right] = \operatorname{Cov}\left[ Y_{i,g}(z,\pi_2), Y_{j,g}(z,\pi_2)\right]  ~.
\end{equation}
\end{assumption}
Assumption \ref{ass:cond-indep} is a weaker assumption than assuming that outcomes of units are independent and identically distributed (i.i.d) within a cluster, as it only requires conditional independence between individual covariates and the covariance of outcomes. It is analogous to the standard homoscedasticity assumption, which assumes constant variance of errors in a regression model, except that it is a statement about covariance instead of variance. Under these two additional assumptions I obtain the following optimality result:
\begin{theorem}\label{thm:optimal-second-stage}
Under Assumption \ref{as:assignment2}, \ref{ass:assignment-crd} and \ref{ass:cond-indep}, $V_a(z)$ is minimized when the second-stage design is a finely stratified design that matches on $E\left[Y_{i,g}(z,\pi_2)\mid (X_{i,g}:1\leq i \leq N_g)\right]$ for $z \in \{0,1\}$ and $a\in\{1,2,3,4\}$.
\end{theorem}

\begin{remark}\label{remark:second-stage-optimality}
    By Theorem \ref{thm:optimal-second-stage}, if outcomes are highly correlated within a cluster, it is advisable to match not only on individual covariates but also on neighbors' covariates. For example, in a school-based experiment, treatments within a school could be assigned by matching on each student’s baseline outcome and the average baseline outcome of their close friends.
\end{remark}

Though practitioners may not have knowledge of the index functions in Theorem \ref{thm:optimal-first-stage} and \ref{thm:optimal-second-stage}, optimal stratification can be determined in some special cases. For instance, in experiments where the first-stage design uses only a univariate covariate $C_g$ \citep[see, e.g.,][]{Ichino2012}, and practitioners expect a monotonic relationship between $S_g$ and $C_g$, the optimal stratification is to order the units by $C_g$ and group adjacent units. Similar results apply to the second-stage design. In more general cases where monotonicity does not hold or the baseline covariates are multivariate, a suitable matching algorithm \citep[see, e.g.,][]{bai-inference, cytrybaum2021} that directly matches on vectors of covariates can be asymptotically as efficient if the sample size is sufficiently large. In cases where the sample size is not sufficiently large, \cite{McKenzie2009} and \cite{Bai-optimal} suggest matching on the baseline outcome, when available. If none of the aforementioned options is available, matching in a sub-optimal way can still be effective, as both \cite{Bai-optimal} and simulation results from Section \ref{sec:simulation} demonstrate that matching units sub-optimally can be more effective than completely randomized designs or some sub-optimal stratified block randomization designs. In this case, it could be useful to consider the recommendations in Remarks \ref{remark:first-stage-optimality} and \ref{remark:second-stage-optimality} for the choice of covariates.

\section{Covariate Adjustment}\label{sec:covariate-adjustment}

In the empirical literature, it is common to match or stratify on a small set of covariates expected to be most predictive of outcomes, and to adjust for additional pre-treatment covariates ex-post. Consequently, this section introduces a linearly covariate-adjusted modification of $\hat\theta_2^P$, the size-weighted primary effect estimator. Adjusted estimators for other estimands follow a similar methodology and are thus omitted for brevity.

To begin, I introduce a new set of baseline covariates \(L_g\) that were not used for treatment assignment. These covariates \(L_g\) may include cluster-level aggregates of individual-level outcomes, such as intracluster means and quantiles. For the remainder of Section \ref{sec:covariate-adjustment}, the assumptions specified in Section \ref{sec:setup} are modified such that \(C_g\) is replaced by \((C_g, L_g)\) throughout. In particular, references to Assumption \ref{ass:Q_G} should now be considered to include \((C_g, L_g)\) instead of \(C_g\). Following this, the treatment status is assigned as follows:

\begin{assumption}\label{ass:assignment-match-covariate-adjustment}
Treatments are assigned so that $(W^{(G)}, L^{(G)}) \independent H^{(G)} | S^{(G)}$ and, conditional on $S^{(G)}$,
\[ \{(I\{H_{i}=\pi_2\}): i \in \lambda_j): 1 \leq j \leq n\} \]
are i.i.d.\ and each uniformly distributed over all permutations of $\left\{z\in \{0,1\}^{k}: \sum_{z=1}^k z_{i}= l \right\}$.
\end{assumption}

I consider a linearly covariate-adjusted estimator based on a set of regressors generated by $C_g, N_g, L_g$. To this end, define $\psi_g = \psi( C_g, N_g, L_g)$, where $\psi: \text{supp}((C_g,N_g,L_g)) \to \mathbf R^p$. We impose the following assumptions on $\psi$:

\begin{assumption} \label{ass:psi}
    The function $\psi$ is such that
    \begin{enumerate}[(a)]
        \item No component of $\psi$ is a constant and $E[\var[\psi_g | S_g]]$ is nonsingular.
        \item $\var[\psi_g] < \infty$.
        \item $E[\psi_g | S_g = s]$, $E[\psi_g \psi_g' | S_g = s]$, and $E[\psi_g \bar Y_g^r(z,h) N_g^\ell | S_g = s]$ for $(z,h) \in \{(0,0), (0, \pi_2), (1,\pi_2)\}$ and $r,\ell \in \{0, 1, 2\}$ are Lipschitz.
        \item For some $C < \infty$, $P \{E[\|\psi_g\|^2 \bar Y_g^2(z,h) | S_g] \leq C\} = 1$ for $(z,h) \in \{(0,0), (0, \pi_2), (1,\pi_2)\}$.
    \end{enumerate}
\end{assumption}

I extend the covariate-adjusted estimator from \cite{mp-cluster} to accommodate the finely stratified design with a general treatment fraction $\pi_1$, as discussed in this paper. Let $\hat\mu_{1,j}$ represent the averaged value of $\tilde Y_g^1 \bar N_G$ among treated clusters within the $j$-th tuple, i.e., $g \in \lambda_j$. Similarly, $\hat\mu_{0,j}$ denotes the corresponding value for control clusters. Additionally, $\hat\psi_{1,j}$ and $\hat\psi_{0,j}$ refer to the averaged values of $\psi_g$ for treated and control clusters, respectively. Formaly, define
\begin{align*}
    \hat\mu_{1,j} &= \frac{1}{l} \sum_{g\in\lambda_j} \tilde Y_g^1 \bar N_G I\{H_g = \pi_2\} \\
    \hat\mu_{0,j} &= \frac{1}{k-l} \sum_{g\in\lambda_j} \tilde Y_g^1 \bar N_G I\{H_g = 0\} \\
    \hat\psi_{1,j} &= \frac{1}{l} \sum_{g\in\lambda_j} \psi_g I\{H_g = \pi_2\} \\
    \hat\psi_{0,j} &= \frac{1}{k-l} \sum_{g\in\lambda_j} \psi_g I\{H_g = 0\} ~,
\end{align*}
where $\bar N_G = \sum_{1\leq g \leq G} N_g/G$. Then, I define the linear adjustment coefficient $\hat\beta_2^P$ as the ordinary least squares (OLS) estimator of the slope coefficient in the linear regression of $\hat\mu_{1,j} - \hat\mu_{0,j}$ on a constant and $\hat\psi_{1,j} - \hat\psi_{0,j}$. Finally, I introduce my covariate-adjusted estimator for the size-weighted primary treatment effect as follows:
\begin{align}\label{eqn:covariate-adjusted-estimator}
    \begin{split}
        \hat \theta_2^{P, adj} &= \frac{1}{N_T} \sum_{1\leq g \leq G} I\{H_g = \pi_2 \} (N_g \bar{Y}_{g}^1 -(\psi_g - \bar \psi_G)' \hat\beta_2^P) \\
        &\hspace{3em} - \frac{1}{N_C} \sum_{1\leq g \leq G}I\{H_g = 0 \} (N_g \bar{Y}_{g}^1 - (\psi_g - \bar \psi_G)' \hat\beta_2^P)~,
    \end{split}
\end{align}
where
\begin{equation*}
    \bar \psi_G = \frac{1}{G} \sum_{1 \leq g \leq G} \psi_g~.
\end{equation*}

The following theorem derives the asymptotic behavior of my covariate-adjusted estimator for $\theta_2^P$, and, importantly, it shows that the limiting variance of $\hat \theta_2^{P, adj}$ is no larger than that of $\hat \theta_2^{P}$ in Theorem \ref{thm:matched-group} and can be strictly smaller.

\begin{theorem} \label{thm:adj}
    Suppose Assumption \ref{as:interference} holds, $Q_G$ satisfies Assumptions \ref{ass:Q_G} and \ref{as:Q_G-lip} and the treatment assignment mechanism satisfies Assumptions \ref{as:close}, \ref{as:assignment2} and \ref{ass:assignment-match-covariate-adjustment}, and $\psi$ satisfies Assumption \ref{ass:psi},
    \[ \sqrt G(\hat \theta_2^{P, adj} - \theta_2^P) \stackrel{d}{\to} N(0, V_2^\ast(1)) \]
    as $G \to \infty$, where
    \begin{align*}
        V_2^*(1) &= \frac{1}{\pi_1} \var[ Y_g^*(1,\pi_2)] + \frac{1}{1-\pi_1} \var[ Y_g^*(0,0)]\\
        &\hspace{3em} - \pi_1(1-\pi_1) E\left[\left(\frac{1}{\pi_1}E[ Y_g^*(1,\pi_2)\mid S_g] + \frac{1}{1-\pi_1}E[ Y_g^*(0,0)\mid S_g] \right)^2\right]
    \end{align*}
    with
    \begin{align*}
        Y_g^\ast(z,h) &= \tilde Y_g(z,h) - \frac{(\psi_g - E[\psi_g])' \beta_2^P}{E[N_g]} \quad \text{ for $(z,h) \in \{(0,0), (1,\pi_2)\}$}
    \end{align*}
    and
    \begin{equation} \label{eq:beta}
        \beta_2^P = \pi_1(1-\pi_1) (E[\var[\psi_g \mid S_g]])^{-1}  E\left[  \cov \left[\frac{1}{\pi_1} \tilde Y_g(1, \pi_2)   + \frac{1}{1-\pi_1}  \tilde Y_g(0,0),  \psi_g\mid S_g\right] \right] E[N_g]~.
    \end{equation}
    Moreover, 
    \begin{equation}\label{eq:kappa}
    V_2^*(1) = V_2(1) - \kappa^2~,
    \end{equation} where
    \[
    \kappa^2 = \frac{1}{\pi_1(1-\pi_1)} \frac{1}{E[N_g]^2} E\left[ \var[ \psi_g' \beta^P_2\mid S_g] \right]~.
    \]
    As a consequence, $V_2^*(1) \le V_2(1)$, with equality if and only if $\kappa^2 = 0$.
\end{theorem}

\begin{remark}\label{remark:covariate-adjustment1}
The specific motivation for using the OLS estimator \(\hat\beta_2^P\) stems from its ability to improve efficiency. For matched pairs experiments, this style of covariate adjustment was proposed as early as Section 10.6 of \cite{imbensrubin2015}. Under equal allocation (\(\pi_1 = 1/2\)), its optimality was shown independently by \cite{bai2023} and \cite{cytrynbaum2023covariate}. The generalization to unequal treatment probabilities (\(\pi_1 \neq 1/2\)) considered here was first proposed and analyzed in Section 3.4.3 of \cite{cytrynbaum2023covariate}, where their ``Group OLS'' estimator coincides with our estimator in the special case of individual-level experiments (\(N_g = 1\)) and full treatment saturation (\(\pi_2 = 1\)). In \cite{mp-cluster}, this estimator is adapted to cluster randomized trials, resulting in an estimator closely aligned with \(\hat\theta_2^{P, adj}\). Remark 3.5 of \cite{mp-cluster} discusses the technical distinctions in covariate adjustment between individual- and cluster-level experiments.
\end{remark}

For variance estimation, I employ the same methodology as \(\hat{V}_2(z)\) but with a modification: \(\tilde{Y}_g^z\) is replaced by \(\mathring{Y}_g^z = \tilde{Y}_g^z - \frac{(\psi_g - \bar{\psi}_G)' \hat\beta_2^P}{ \frac{1}{G}\sum_{1 \leq g \leq G} N_g }\). The consistency of this variance estimator follows from combining the arguments used to establish Theorem \ref{thm:variance-estimator-mt} and those used to establish Theorem 3.2 in \cite{bai2023}.

\section{Simulations}\label{sec:simulation}
In this section, I illustrate the results presented in Section \ref{sec:main-results} with a simulation study. To begin with, potential outcomes are generated according to the equation:
\begin{equation*}
    Y_{i,g}(z,h) = \mu_{z,h} + \alpha_{z,h} X_{1,i,g}/(X_{2,i,g} + 0.1)  + \beta_{z,h}\left(C_g - \frac{1}{2}\right) + \gamma\left(N_g - 100\right) + \sigma(C_g, N_g) \epsilon_{i,g}~,
\end{equation*}
for $(z,h) \in \{(0,0),(0,\pi_2),(1,\pi_2)\}$, where
\begin{itemize}
    \item $C_g, N_g$ are i.i.d with $C_g \sim \text{Unif}[0,1]$, and $N_g \sim \text{Unif}\{50,\dots,150\}$, which are mutually independent.
    \item $X_{1, i,g} = N_g u_{i,g}/100$, where $u_{i,g}$ are i.i.d $N(0,0.1)$ across $i,g$. $X_{2,i,g}$ are i.i.d $\text{Unif}[0,1]$ across $i,g$.
    \item $\mu_{1,\pi_2} = \mu_{0,\pi_2} + \tau = \mu_{0,0} + \tau + \omega$\footnote{In Table \ref{table:mse} and \ref{table:reject-probs-ols}, $\tau = \omega = 0$. In Table \ref{table:reject-probs-1}, $\tau = \omega = 0$ for $H_0$ and $\tau = \omega = 0.05$ for $H_1$.}, i.e. primary and spillover effects are additive and homogeneous. 
    \item $\sigma(C_g, N_g) = C_g (N_g-100)/100$ and $\epsilon_{i,g} \sim N(0,10)$, which satisfies Assumption \ref{ass:cond-indep}.
\end{itemize}
All simulations are performed with a sample of $200$ clusters, in which all units are sampled, i.e. $N_g=M_g$.

\begin{table}[ht!]
\centering
\setlength{\tabcolsep}{5pt}
\begin{adjustbox}{max width=0.75\linewidth,center}
\begin{tabular}{lllllllll}
\toprule
 & & \multicolumn{7}{c}{Second-stage}   \\\cmidrule{3-9}
First-stage              & Parameter & \textbf{C}    & \textbf{S-2}   & \textbf{S-4}   & \textbf{S-4O} & \textbf{MT-A}  & \textbf{MT-B}  & \textbf{MT-C}   \\
\midrule
\multirow{4}{*}{\textbf{C}}    & $\theta^P_1$ & 1.0000 & 0.9601 & 0.9270 & 0.9235 & 0.9720 & 0.9323 & \textbf{0.9106} \\
                               & $\theta^P_2$ & 1.0000 & 0.9803 & 0.9404 & \textbf{0.9263} & 0.9939 & 0.9573 & 0.9560 \\
                               & $\theta^S_1$ & 1.0000 & 0.9625 & 0.9187 & 0.9197 & 0.9649 & 0.9410 & \textbf{0.9093} \\
                               & $\theta^S_2$ & 1.0000 & 0.9921 & 0.9432 & \textbf{0.9209} & 0.9875 & 0.9709 & 0.9596 \\
                               &              &        &        &        &                 &        &        &        \\
\multirow{4}{*}{\textbf{S-2}}  & $\theta^P_1$ & 0.8437 & 0.7866 & 0.7859 & \textbf{0.7629} & 0.8473 & 0.7981 & 0.7957 \\
                               & $\theta^P_2$ & 0.8227 & 0.7601 & 0.7877 & \textbf{0.7440} & 0.8361 & 0.7880 & 0.7672 \\
                               & $\theta^S_1$ & 0.8396 & 0.7913 & 0.7754 & \textbf{0.7534} & 0.8473 & 0.8052 & 0.7943 \\
                               & $\theta^S_2$ & 0.8244 & 0.7790 & 0.7806 & \textbf{0.7438} & 0.8456 & 0.7904 & 0.7693 \\
                               &              &        &        &        &                 &        &        &        \\
\multirow{4}{*}{\textbf{S-4}}  & $\theta^P_1$ & 0.7772 & 0.8084 & 0.7730 & 0.7835 & 0.7603 & \textbf{0.7216} & 0.7262 \\
                               & $\theta^P_2$ & 0.7759 & 0.7757 & 0.7330 & 0.7473 & 0.7114 & \textbf{0.6909} & 0.7024 \\
                               & $\theta^S_1$ & 0.7711 & 0.8053 & 0.7656 & 0.7749 & 0.7556 & 0.7357 & \textbf{0.7283} \\
                               & $\theta^S_2$ & 0.7773 & 0.7848 & 0.7330 & 0.7482 & 0.7204 & 0.7100 & \textbf{0.7091} \\
                               &              &        &        &        &                 &        &        &        \\
\multirow{4}{*}{\textbf{S-4O}}  & $\theta^P_1$ & 0.2104 & 0.2102 & 0.2026 & \textbf{0.2010} & 0.2172 & 0.2115 & 0.2035 \\
                               & $\theta^P_2$ & 0.2418 & 0.2428 & 0.2371 & 0.2285 & 0.2339 & 0.2494 & \textbf{0.2241} \\
                               & $\theta^S_1$ & 0.2081 & 0.2136 & 0.2028 & \textbf{0.2002} & 0.2158 & 0.2221 & 0.2004 \\
                               & $\theta^S_2$ & 0.2367 & 0.2489 & 0.2418 & 0.2254 & 0.2396 & 0.2606 & \textbf{0.2226} \\
                               &              &        &        &        &                 &        &        &        \\
\multirow{4}{*}{\textbf{MT-A}} & $\theta^P_1$ & 0.7683 & 0.8172 & 0.7573 & 0.7401 & 0.7347 & 0.7744 & \textbf{0.7097} \\
                               & $\theta^P_2$ & 0.7555 & 0.7693 & 0.7202 & \textbf{0.6726} & 0.7159 & 0.7665 & 0.6769 \\
                               & $\theta^S_1$ & 0.7592 & 0.8157 & 0.7573 & 0.7277 & 0.7310 & 0.7882 & \textbf{0.7035} \\
                               & $\theta^S_2$ & 0.7537 & 0.7763 & 0.7221 & \textbf{0.6644} & 0.7123 & 0.7847 & 0.6771 \\
                               &              &        &        &        &                 &        &        &        \\
\multirow{4}{*}{\textbf{MT-B}} & $\theta^P_1$ & 0.2935 & 0.2806 & \textbf{0.2719} & 0.2970 & 0.2912 & 0.2847 & 0.2797 \\
                               & $\theta^P_2$ & 0.4175 & 0.4013 & \textbf{0.3802} & 0.4120 & 0.4134 & 0.3953 & 0.3880 \\
                               & $\theta^S_1$ & 0.2866 & 0.2935 & \textbf{0.2661} & 0.2941 & 0.2811 & 0.2810 & 0.2746 \\
                               & $\theta^S_2$ & 0.4143 & 0.4181 & \textbf{0.3786} & 0.4106 & 0.4020 & 0.3934 & 0.3841 \\
                               &              &        &        &        &                 &        &        &        \\
\multirow{4}{*}{\textbf{MT-C}} & $\theta^P_1$ & 0.1160 & 0.1140 & \textbf{0.1047} & 0.1125 & 0.1095 & 0.1149 & 0.1069 \\
                               & $\theta^P_2$ & 0.0921 & 0.0873 & 0.0818 & 0.0893 & 0.0842 & 0.0874 & \textbf{0.0755} \\
                               & $\theta^S_1$ & 0.1221 & 0.1183 & 0.1143 & 0.1126 & 0.1076 & 0.1193 & \textbf{0.1045} \\
                               & $\theta^S_2$ & 0.0997 & 0.0930 & 0.0908 & 0.0891 & 0.0829 & 0.0914 & \textbf{0.0757} \\
\bottomrule
\end{tabular}
\end{adjustbox}
\begin{tablenotes}
\footnotesize
\item Note: The rows indicate first-stage designs, and columns indicate second-stage designs.
\end{tablenotes}
\caption{Ratio of MSE under all designs against those under complete randomization in both stages}
\label{table:mse}
\end{table}

\subsection{MSE Properties}\label{sec:mse}
This section examines the performance of optimal matched tuples designs and several other designs via comparison of their MSEs (Mean Squared Errors). For simplicity, the parameters are given as follows: $\alpha_{z,h} = \beta_{z,h} = 1, \gamma = 1/100$ for all $(z,h)\in \{(0,0),(0,\pi_2),(1,\pi_2)\}$. This model configuration is referred to as ``homogeneous model'' since treatment effects are fully captured by $\mu_{z,h}$ and thus are homogeneously additive in this setting. A more complicated ``heterogeneous model'' will be introduced later. According to Theorem \ref{thm:optimal-first-stage}, the optimal index functions for equally-weighted and size-weighted effects in the first stage are
\begin{align}
    & E\left[  \frac{\bar{Y}_g(1,\pi_2)}{\pi_1} + \frac{\bar{Y}_g(0,0) }{1-\pi_1}\mid C_g, N_g \right] \propto  C_g + N_g/100 ~, \label{eqn:optimal-function-simulation-1}\\
    & E\left[  \frac{\tilde{Y}_g(1,\pi_2)}{\pi_1} + \frac{\tilde{Y}_g(0,0) }{1-\pi_1}\mid C_g, N_g \right] \propto N_g (C_g +  N_g/100 ) - \frac{25}{3} N_g \label{eqn:optimal-function-simulation-2}~.
\end{align}
In the second stage, the optimal finely stratified design matches on $X_{1,i,g}/(X_{2,i,g}+0.1)$ according to Theorem \ref{thm:optimal-second-stage}.
This section considers the following experimental designs for both stages:
\begin{enumerate}
    \item \textbf{(C)} $(H_g: 1\leq g \leq G)$ is drawn from a completely randomized design (also known as permuted block design), i.e. uniformly from the assignment space that $\pi_1 G$ (or $\pi_2 N_g$ in the second stage) number of clusters/units get treated.
    \item \textbf{(S-2)} A SBR design, where the experimental sample is divided into two strata using the midpoint of covariate $C_g$ (or $X_{1,i,g}$ in the second stage) as the cutoff. In each stratum, treatment is assigned as in \textbf{C}.
    \item \textbf{(S-4)} As in \textbf{(S-2)}, but with four strata.
    \item \textbf{(S-4O)} The ``optimal'' stratification with four strata. Clusters/units are divided into strata using quartiles of (\ref{eqn:optimal-function-simulation-1}) and (\ref{eqn:optimal-function-simulation-2}) for equally- and size-weighted estimands respectively (or $X_{1,i,g}/(X_{2,i,g}+0.1)$ in the second stage).
    \item \textbf{(MT-A)} Matched tuples design where units are ordered according to $C_g$ (or $X_{1,i,g}$ in the second stage).
    \item \textbf{(MT-B)}  Matched tuples design where units are ordered according to cluster size $N_g$ (or $X_{2,i,g}$ in the second stage).
    \item \textbf{(MT-C)} The optimal matched tuples design where units are ordered according to (\ref{eqn:optimal-function-simulation-1}) and (\ref{eqn:optimal-function-simulation-2}) for equally- and size-weighted estimands respectively (or $X_{1,i,g}/(X_{2,i,g}+0.1)$ in the second stage).
\end{enumerate}

Table \ref{table:mse} shows the ratio of the MSE of each design relative to the MSE of the design with completely randomized assignments (\textbf{C}) in both stages, computed across 1000 Monte Carlo iterations. The rows indicate first-stage designs, and columns indicate second-stage designs. The lowest values in each row are marked in bold. In all designs, treatment effects are set to zero by assigning $\mu_{z,h} = 0$ for all $(z,h)\in {(0,0),(0,\pi_2),(1,\pi_2)}$, and the treated fraction is set to $1/2$ in both stages. As expected from Theorem \ref{thm:optimal-first-stage} and \ref{thm:optimal-second-stage}, the matched-tuples design with complete matching (\textbf{MT-C}) outperforms the other designs in the first stage for all parameters of interest while remaining optimal in the second stage for many cases. However, it is noticeable that the assignment mechanism in the first stage has a greater effect on statistical precision than the second stage.

\begin{table}[ht!]
\centering
\setlength{\tabcolsep}{3pt}
\begin{adjustbox}{max width=0.95\linewidth,center}
\begin{tabular}{lllllllllllllll}
\toprule
       &              & \multicolumn{13}{c}{Second-stage}    \\  \cmidrule{3-15}
       &              & \multicolumn{6}{c}{$H_0: \tau = \omega = 0$}                      &  & \multicolumn{6}{c}{$H_1: \tau = \omega = 0.05$}                      \\  \cmidrule{3-8} \cmidrule{10-15}
First-stage & Parameter &  \textbf{S-2}   & \textbf{S-4}   & \textbf{S-4O} & \textbf{MT-A}  & \textbf{MT-B}  & \textbf{MT-C} & & \textbf{S-2}   & \textbf{S-4}   & \textbf{S-4O} & \textbf{MT-A}  & \textbf{MT-B}  & \textbf{MT-C} \\
\toprule
\multirow{4}{*}{\textbf{S-2}}  & $\theta^P_1$ & 0.044 & 0.066 & 0.063 & 0.044 & 0.050 & 0.050 &  & 0.222 & 0.244 & 0.244 & 0.248 & 0.262 & 0.258 \\
                               & $\theta^P_2$ & 0.045 & 0.062 & 0.059 & 0.049 & 0.062 & 0.058 &  & 0.224 & 0.226 & 0.229 & 0.239 & 0.256 & 0.262 \\
                               & $\theta^S_1$ & 0.046 & 0.061 & 0.065 & 0.043 & 0.052 & 0.050 &  & 0.084 & 0.100 & 0.102 & 0.101 & 0.095 & 0.098 \\
                               & $\theta^S_2$ & 0.046 & 0.066 & 0.066 & 0.046 & 0.056 & 0.061 &  & 0.087 & 0.101 & 0.091 & 0.094 & 0.101 & 0.094 \\
\multicolumn{1}{l}{}           &              &       &       &       &       &       &       &  &       &       &       &       &       &       \\
\multirow{4}{*}{\textbf{S-4}}  & $\theta^P_1$ & 0.050 & 0.048 & 0.060 & 0.058 & 0.036 & 0.051 &  & 0.241 & 0.267 & 0.243 & 0.275 & 0.245 & 0.276 \\
                               & $\theta^P_2$ & 0.056 & 0.055 & 0.062 & 0.051 & 0.037 & 0.056 &  & 0.230 & 0.261 & 0.241 & 0.284 & 0.250 & 0.270 \\
                               & $\theta^S_1$ & 0.054 & 0.053 & 0.062 & 0.056 & 0.037 & 0.048 &  & 0.096 & 0.119 & 0.105 & 0.130 & 0.109 & 0.112 \\
                               & $\theta^S_2$ & 0.058 & 0.055 & 0.061 & 0.054 & 0.033 & 0.056 &  & 0.087 & 0.121 & 0.096 & 0.127 & 0.107 & 0.110 \\
\multicolumn{1}{l}{}           &              &       &       &       &       &       &       &  &       &       &       &       &       &       \\
\multirow{4}{*}{\textbf{S-4O}}  & $\theta^P_1$ & 0.048 & 0.051 & 0.052 & 0.058 & 0.054 & 0.066 &  & 0.692 & 0.729 & 0.691 & 0.708 & 0.685 & 0.716 \\
                               & $\theta^P_2$ & 0.048 & 0.059 & 0.062 & 0.058 & 0.055 & 0.060 &  & 0.608 & 0.629 & 0.588 & 0.630 & 0.582 & 0.639 \\
                               & $\theta^S_1$ & 0.048 & 0.055 & 0.047 & 0.057 & 0.054 & 0.057 &  & 0.220 & 0.268 & 0.247 & 0.282 & 0.222 & 0.241 \\
                               & $\theta^S_2$ & 0.052 & 0.060 & 0.054 & 0.055 & 0.052 & 0.058 &  & 0.220 & 0.228 & 0.214 & 0.246 & 0.192 & 0.208 \\
\multicolumn{1}{l}{}           &              &       &       &       &       &       &       &  &       &       &       &       &       &       \\
\multirow{4}{*}{\textbf{MT-A}} & $\theta^P_1$ & 0.060 & 0.049 & 0.044 & 0.050 & 0.044 & 0.060 &  & 0.270 & 0.271 & 0.260 & 0.252 & 0.240 & 0.256 \\
                               & $\theta^P_2$ & 0.058 & 0.049 & 0.041 & 0.048 & 0.050 & 0.056 &  & 0.254 & 0.260 & 0.268 & 0.237 & 0.236 & 0.259 \\
                               & $\theta^S_1$ & 0.055 & 0.042 & 0.040 & 0.050 & 0.052 & 0.058 &  & 0.109 & 0.101 & 0.100 & 0.101 & 0.105 & 0.106 \\
                               & $\theta^S_2$ & 0.055 & 0.052 & 0.049 & 0.046 & 0.051 & 0.052 &  & 0.115 & 0.100 & 0.097 & 0.096 & 0.100 & 0.105 \\
\multicolumn{1}{l}{}           &              &       &       &       &       &       &       &  &       &       &       &       &       &       \\
\multirow{4}{*}{\textbf{MT-B}} & $\theta^P_1$ & 0.044 & 0.053 & 0.057 & 0.031 & 0.043 & 0.051 &  & 0.565 & 0.582 & 0.586 & 0.553 & 0.530 & 0.586 \\
                               & $\theta^P_2$ & 0.053 & 0.047 & 0.058 & 0.041 & 0.045 & 0.057 &  & 0.402 & 0.419 & 0.444 & 0.403 & 0.378 & 0.431 \\
                               & $\theta^S_1$ & 0.049 & 0.045 & 0.052 & 0.035 & 0.051 & 0.052 &  & 0.197 & 0.203 & 0.216 & 0.174 & 0.184 & 0.198 \\
                               & $\theta^S_2$ & 0.053 & 0.046 & 0.057 & 0.038 & 0.050 & 0.059 &  & 0.148 & 0.158 & 0.180 & 0.131 & 0.135 & 0.148 \\
\multicolumn{1}{l}{}           &              &       &       &       &       &       &       &  &       &       &       &       &       &       \\
\multirow{4}{*}{\textbf{MT-C}} & $\theta^P_1$ & 0.058 & 0.056 & 0.061 & 0.044 & 0.053 & 0.043 &  & 0.920 & 0.939 & 0.917 & 0.917 & 0.919 & 0.933 \\
                               & $\theta^P_2$ & 0.057 & 0.045 & 0.058 & 0.041 & 0.052 & 0.051 &  & 0.955 & 0.975 & 0.955 & 0.950 & 0.941 & 0.968 \\
                               & $\theta^S_1$ & 0.074 & 0.058 & 0.059 & 0.044 & 0.042 & 0.044 &  & 0.399 & 0.429 & 0.427 & 0.416 & 0.400 & 0.411 \\
                               & $\theta^S_2$ & 0.058 & 0.052 & 0.062 & 0.034 & 0.050 & 0.050 &  & 0.430 & 0.465 & 0.471 & 0.472 & 0.444 & 0.504 \\
\bottomrule
\end{tabular}
\end{adjustbox}
\begin{tablenotes}
\footnotesize
\item Note: The rows indicate first-stage designs, and columns indicate second-stage designs.
\end{tablenotes}
\caption{Rejection probabilities under the null and alternative hypothesis}
\label{table:reject-probs-1}
\end{table}

\subsection{Inference}\label{sec:simulation-inference}
In this section, the focus shifts from optimality to studying the finite sample properties of different tests for the following null hypotheses of interest: 
\begin{equation}
    H_0^{P,1}: \theta_1^P(Q_G) = 0, \quad H_0^{P,2}: \theta_2^P(Q_G) = 0, \quad H_0^{S,1}: \theta_1^S(Q_G) = 0, \quad H_0^{S,2}: \theta_2^S(Q_G) = 0~,
\end{equation}
against the alternative hypotheses:
\begin{equation}
    H_1^{P,1}: \theta_1^P(Q_G) = \tau + \omega, \quad H_1^{P,2}: \theta_2^P(Q_G) = \tau + \omega, \quad H_1^{S,1}: \theta_1^S(Q_G) = \omega, \quad H_1^{S,2}: \theta_2^S(Q_G) = \omega~.
\end{equation}

In Table \ref{table:reject-probs-1}, the six assignment mechanisms with covariate-adaptive randomization (Design 2-7 in Section \ref{sec:mse}) for the first and second stages are considered, resulting in a total of 36 different designs. Hypothesis tests are performed at a significance level of 0.05, and rejection probabilities under the null and alternative hypotheses are computed from 1000 Monte Carlo iterations in each case. Tests are constructed as ``adjusted $t$-tests'' using the asymptotic results from Theorem \ref{thm:matched-group}-\ref{thm:variance-estimator-mt}. For stratified designs in the first stage (\textbf{S-2}, \textbf{S-4} and \textbf{S-4O}), tests for equally- and size-weighted effects are performed using the variance estimators $\hat V_3(z)$ and $\hat V_4(z)$ (see (\ref{eqn:V3(z)hat}) and (\ref{eqn:V4(z)hat}) in Appendix \ref{sec:large-strata}). For matched tuples designs in the first stage (\textbf{MT-A}, \textbf{MT-B} and \textbf{MT-C}), tests for equally- and size-weighted effects are performed using the variance estimators $\hat V_1(z)$ and $\hat V_2(z)$. The results show that the rejection probabilities are universally around 0.05 under the null hypothesis, which verifies the validity of tests based on my asymptotic results across all the designs. Under the alternative hypotheses, the rejection probabilities vary substantially across the first-stage designs while remaining relatively stable across the second-stage designs. \textbf{MT-C} stands out as the most powerful design for the first-stage. These findings are consistent with previous section.

\begin{table}[ht!]
\centering
\setlength{\tabcolsep}{3pt}
\begin{adjustbox}{max width=0.95\linewidth,center}
\begin{tabular}{lclllllll}
\toprule
\multirow{2}{*}{Model}  & \multirow{2}{*}{Inference Method}     &  \multirow{2}{*}{Effect}        & \textbf{S-4O} & \textbf{S-4O} & \textbf{S-4O} & \textbf{MT-C} & \textbf{MT-C} & \textbf{MT-C} \\
& &  & \textbf{C} & \textbf{S-4O} & \textbf{MT-C} & \textbf{C} & \textbf{S-4O} & \textbf{MT-C} \\
\toprule
\multirow{9}{*}{\textbf{Homogeneous}}  & OLS robust                         & Primary   & 0.184 & 0.194 & 0.156 & 0.062 & 0.086 & 0.049 \\
& (standard $t$-test)                                            & Spillover   & 0.184 & 0.167 & 0.159 & 0.077 & 0.048 & 0.048 \\
\multicolumn{1}{l}{}                         &                  &          &           &           &          &           &           \\
& OLS cluster                         & Primary   & 0.000    & 0.000     & 0.000     & 0.000    & 0.000     & 0.000     \\
& (clustered $t$-test)                                            & Spillover & 0.000    & 0.000     & 0.000     & 0.000    & 0.000     & 0.000     \\
\multicolumn{1}{l}{}                         &                  &          &           &           &          &           &           \\
& OLS with group                & Primary   & 0.209 & 0.196 & 0.179 & 0.100 & 0.106 & 0.077 \\
& fixed effects (robust)                                             & Spillover & 0.201 & 0.184 & 0.177 & 0.113 & 0.100 & 0.075 \\
\multicolumn{1}{l}{}                         &                  &          &           &           &          &           &           \\
& OLS with group  & Primary   &0.028 & 0.027 & 0.029 & 0.068 & 0.085 & 0.071 \\ 
& fixed effects (clustered)                                            & Spillover & 0.036 & 0.027 & 0.026 & 0.064 & 0.062 & 0.069    \\
\\ \\
\multirow{9}{*}{\textbf{Heterogeneous}} & OLS robust                          & Primary   & 0.118  & 0.118  & 0.175  & 0.061  & 0.048 &  0.080     \\
& (standard $t$-test)                                            & Spillover & 0.225 &   0.213 &   0.162 &   0.135 &   0.144 &   0.069    \\
\multicolumn{1}{l}{}                         &                  &          &           &           &          &           &           \\
& OLS cluster                         & Primary   & 0.000    & 0.001     & 0.000     & 0.000    & 0.000     & 0.000    \\
& (clustered $t$-test)                                              & Spillover & 0.002    & 0.000     & 0.000     & 0.000    & 0.000     & 0.000     \\
\multicolumn{1}{l}{}                         &                  &          &           &           &          &           &           \\
& OLS with group    &   Primary         & 0.118 &   0.115 &   0.172 &   0.079 &   0.057 &   0.125   \\
& fixed effects (robust)                                             & Spillover & 0.250 &   0.253 &   0.166 &   0.273 &   0.265 &   0.150    \\
\multicolumn{1}{l}{}                         &                  &          &           &           &          &           &           \\
& OLS with group  & Primary   & 0.024 &   0.015 &   0.023 &   0.056 &   0.051 &   0.047     \\
& fixed effects (clustered)                                            & Spillover & 0.027 &   0.018 &   0.025 &   0.045 &   0.071 &   0.061    \\
\bottomrule
\end{tabular}
\end{adjustbox}
\caption{Rejection probabilities of various inference methods under the null hypothesis}
\label{table:reject-probs-ols}
\end{table}

Next, the validity of commonly used regression-based inference methods in the empirical literature is tested. These methods are tested under both the ``homogeneous model'' from the previous simulation study in Section \ref{sec:mse} and a ``heterogeneous model'' in which two parameters are modified as follows: $\alpha_{1,\pi_2} = \beta_{1,\pi_2} = 2$, $\alpha_{0,\pi_2} = \beta_{0,\pi_2} = 0.5$, and $\alpha_{0,0} = \beta_{0,0} = 1$. The key difference between the two models is whether the conditional expectations of potential outcomes are identical or different across different exposures $(z,h)$. Four commonly used regression methods are considered in this study:
\begin{enumerate}
    \item OLS robust: regress $Y_{i,g}$ on a constant, individual-level treatment indicator $Z_{i,g}$ and the indicator for untreated units in treated clusters $L_{i,g}$. Tests for primary and spillover effects are performed using standard $t$-tests under robust standard errors to heteroskedasticity.
    \item OLS cluster: run the same regression as ``OLS robust'' but perform $t$-tests with clustered standard errors.
    \item OLS with group fixed effects (robust): regress $Y_{i,g}$ on a constant, $Z_{i,g}$, $L_{i,g}$ and fixed effects for strata or tuples $S_g$. Tests are  performed using standard $t$-tests under robust standard errors to heteroskedasticity.
    \item OLS with group fixed effects (clustered): run the same regression as ``OLS with group fixed effects (robust)'' but perform $t$-tests with clustered standard errors.
\end{enumerate}
Note that due to full sampling, i.e. $N_g = M_g$, regressions without fixed effects (``OLS robust'' and ``OLS cluster'') output the same estimators as the size-weighted estimators $\hat \theta_2^P$ and $\hat \theta_2^S$. Most of the previous empirical analysis on covariate-adaptive two-stage experiments report cluster-robust standard errors in their main results, which could either be ``OLS cluster'' \citep[see for example][]{basse2018} or ``OLS with group fixed effects (clustered)'' \citep[see for example][]{Duflo2003,Ichino2012}. For brevity, Table \ref{table:reject-probs-ols} includes only six designs: those with either \textbf{S-4O} or \textbf{MT-C} in the first stage, and \textbf{C}, \textbf{S-4O}, or \textbf{MT-C} in the second stage. The table reveals that test results can be either conservative or invalid across different regression methods and designs. For stratified designs in the first stage, methods based on ``robust'' standard errors tend to over-reject, while methods based on ``clustered'' standard errors tend to under-reject. For matched tuples designs, ``OLS cluster'' is conservative, and the remaining methods could be invalid as they may over-reject the null hypothesis under some model specifications and parameters of interest. Similar results can also be found in the previous literature on covariate-adaptive randomization. For example, \cite{matched-tuple} demonstrated that inferences based on OLS regressions with strata fixed effects could be invalid. On the other hand, \cite{dechaisemartin2022} documented that in cluster randomized experiments, $t$-test based on clustered standard errors tend to over-reject the null hypothesis when strata fixed effects are included, and under-reject otherwise. Therefore, it can be concluded that, with the exception of ``OLS cluster'' being conservative, the other three inference methods based on regression are generally invalid.

\section{Empirical Application}\label{sec:empirical}
In this section, the inference methods introduced in Section \ref{sec:main-results} are illustrated using data collected in \cite{foos2017}. The experiment conducted by \cite{foos2017} is a randomly assigned spillover experiment in the United Kingdom designed to identify social influence within heterogeneous and homogeneous partisan households. The study first stratified $5190$ two-voter households into three blocks based on the latest recorded party preference of the experimental subject\footnote{Before assigning treatments, the researchers randomly selected one individual per household to potentially receive treatments, whom they mark as ``experimental subjects''. In other words, the second-stage assignment is a complete randomization. Specifically, this two-stage design corresponds to  ``\textbf{S3-C}'' (using the notation from the simulation section).}: ``Labour'' supporter,``rival party'' supporter and those who were ``unattached'' to a party. Then experimental subjects or equivalently their households were randomly assigned to three groups: high partisan intensity treatment, low partisan intensity and control\footnote{The empirical treatment fractions for ``Labour'' supporters are 0.217, 0.217, and 0.566 for the high-intensity, low-intensity, and control groups, respectively. For ``rival party'' supporters, the corresponding fractions are 0.222, 0.215, and 0.563. For ``unattached'' individuals, they are 0.208, 0.226, and 0.566.}. Experimental subjects allocated to treatment groups were called by telephone and encouraged to vote in the PCC election
on November 15, 2012. The ``high partisan intensity'' was formulated in a strongly partisan tone, explicitly mentioning the Labour Party and policies multiple time, while the ``low partisan intensity'' treatment message avoided all statements about party competition.

In the original analysis of \cite{foos2017}, their main focus was on analyzing treatment effects conditional on a wide range of pre-treatment covariates. That said, in the final column of Table 1 in \cite{foos2017}, they report estimators for (unconditional) primary and spillover effects, which are based on calculations of averages over separate experimental subjects and unassigned subjects. In contrast, my estimators do not distinguish experimental subjects from unassigned subjects and take averages solely based on treatment or spillover status. Another difference in my analysis is that estimators are calculated by pooling the two treatment arms, i.e. high and low partisan intensity, to maintain consistency with the setup of the paper\footnote{
Specifically, treated households effectively received a ``random treatment'': high partisan intensity with some probability and low partisan intensity with the complementary probability. The pooled treatment still follows a complete randomization design within each stratum and therefore satisfies all assumptions related to treatment assignment.
}
. In contrast, \cite{foos2017} provide separate estimates for each treatment arm.

\begin{table}[ht!]
\caption{Point estimates and confidence intervals for testing the primary and spillover effects}
\begin{adjustbox}{max width=0.95\linewidth,center}
\begin{tabular}{lllllllllll}
\toprule
         & \multicolumn{2}{c}{adjusted $t$-test}    & \multicolumn{2}{c}{OLS robust}    & \multicolumn{2}{c}{OLS cluster} & \multicolumn{2}{c}{OLS fe robust} & \multicolumn{2}{c}{OLS fe cluster} \\
\toprule
\multirow{2}{*}{Primary}   & \multicolumn{2}{c}{3.0488} & \multicolumn{2}{c}{3.0488} & \multicolumn{2}{c}{3.0488}      & \multicolumn{2}{c}{2.9971} & \multicolumn{2}{c}{2.9971}         \\
              & [0.8339,      & 5.2638]    & [0.9962,       & 5.1014]      & [0.8103,         & 5.2874]         & [0.9633,       & 5.0308]      & [0.7812,           & 5.2129]          \\
                 &                 &               &              &             &                &                &              &             &                  &                 \\
\multirow{2}{*}{Spillover} & \multicolumn{2}{c}{4.5930} & \multicolumn{2}{c}{4.5930}  & \multicolumn{2}{c}{4.5930}       & \multicolumn{2}{c}{4.5413} & \multicolumn{2}{c}{4.5413}         \\
              & [2.3430,      & 6.8431]    & [2.5046,       & 6.6815]      & [2.3216,         & 6.8645]         & [2.4694,       & 6.6132]      & [2.2904,           & 6.7922]     \\  
\bottomrule
\end{tabular}
\end{adjustbox}
\label{table:application}
\begin{tablenotes}
\footnotesize
\item Note: The original paper did not mention the target treated fraction $\pi_1$. I decided to use the empirical treated fraction, $1/G\sum_{1\leq g \leq G} I\{H_g = \pi_2\}$, to calculate the variance estimators.
\end{tablenotes}
\label{table:empirical}
\end{table}

Table \ref{table:application} compares point estimates of treatment effect on turnout percentage and confidence intervals obtained from the four regression methods listed in Section \ref{sec:simulation-inference} with those based on my theoretical results, namely ``adjusted $t$-test''. Since cluster (household) size is fixed, equally-weighted and size-weighted estimators and estimands collapse into one. Moreover, full sampling ($N_g = M_g = 2$) makes the point estimates of ``adjusted $t$-test'' and ``OLS robust/cluster'' equivalent. In the simulation study, it is found that `OLS robust'' and ``OLS fe robust'' tend to over-reject the null hypothesis, which is consistent with the empirical results in Table \ref{table:application} that they both have narrower confidence interval than the ``adjusted $t$-test''. Furthermore, ``OLS cluster'' and ``OLS fe cluster'' are shown to be conservative in the simulation study, and accordingly, they both have wider confidence intervals than the ``adjusted $t$-test'' in Table \ref{table:application}. Therefore, the empirical findings are consistent with the simulation study in Table \ref{table:reject-probs-ols}.

\section{Recommendations for Empirical Practice}\label{sec:conclude}
%Based on my theoretical results as well as the simulation study above, I conclude with some recommendations for for practitioners when conducting inference about the parameters of interest listed in Table \ref{table:estimands} under the small and large strata designs. If the size of strata is large enough (for example more than 50 clusters as \textbf{S-4} in the simulation study), then my recommendation is that practitioners should employ $\hat V_1(1)$ and $\hat V_1(0)$ in (\ref{eqn:V1(z)hat}) for the equally-weighted primary effect $\theta_1^P$ and spillover effect $\theta_1^S$, $\hat V_2(1)$, and employ $\hat V_2(0)$ in (\ref{eqn:V2(z)hat}) for size-weighted primary effect $\theta_2^P$ and spillover effect $\theta_2^S$. If the empirical researcher is not confident about the strata size, or more often the experimental design is a matched-tuples design with one or two observations for each treatment arm, I recommend using $\hat V_3(1), V_3(0)$ and $\hat V_4(1), \hat V_4(0)$ in (\ref{eqn:Vhat3(z)}) for the corresponding equally-weighted and size weighted effects. We have shown that tests based on the regression in equation (\ref{eqn:ols-full}) with HC2 cluster-robust standard errors are valid but potentially conservative, which would result in a loss of power relative to our proposed test. In our simulation studies, we demonstrated that regressions with strata fixed effects or heteroskedasticity-robust standard errors are generally invalid. 

Based on the theoretical results and the supporting simulation study, I conclude with the following recommendations for empirical practice, particularly in conducting inference about the parameters of interest, as listed in Table \ref{table:estimands}. In scenarios where sizes of all strata are considerably large, such as more than 50 clusters as exemplified in simulation \textbf{S-4}, we advise practitioners to utilize $\hat V_3(1)$ and $\hat V_3(0)$, as defined in (\ref{eqn:V3(z)hat}), for estimating the equally-weighted primary effect $\theta_1^P$ and the spillover effect $\theta_1^S$. Similarly, $\hat V_4(1)$ and $\hat V_4(0)$, as detailed in (\ref{eqn:V4(z)hat}), should be employed for the size-weighted primary effect $\theta_2^P$ and the spillover effect $\theta_2^S$. However, when it is unclear whether the strata size is sufficiently large, or more commonly, when the experimental design involves a matched-tuples design with only one or two observations per treatment arm, we recommend the application of $\hat V_1(1), \hat V_1(0)$ and $\hat V_2(1), \hat V_2(0)$ as indicated in (\ref{eqn:Vhat1}) for the corresponding equally-weighted and size-weighted effects.

The results of this study have shown that tests based on the regression specified in equation (\ref{eqn:ols-full}) with HC2 cluster-robust standard errors are valid but potentially conservative, which would result in a loss of power relative to our proposed test. Further, it's critical to note that regressions using strata fixed effects or heteroskedasticity-robust standard errors have generally been found invalid in the simulation study.

Based on the optimality results for the first-stage design, I recommend selecting cluster-level covariates for matching according to the parameters of interest, as elaborated in Remark \ref{remark:first-stage-optimality}, while adhering to the established guidelines from previous studies \citep{McKenzie2009, bai-inference, Bai-optimal, cytrybaum2021}. For the second stage, it is advisable to first evaluate the impact of the design on efficiency, as detailed in Remark \ref{remark:how-second-stage-affects-variance}, and then assess whether the benefits of second-stage randomization outweigh its costs. Should this be the case, implementing a finely stratified second-stage randomization is recommended, taking into account intra-cluster correlation, as discussed in Remark \ref{remark:second-stage-optimality}.

\newpage
\appendix
\small
\section{Inference for Experiments with Large Strata}\label{sec:large-strata}
In this section, I investigate the asymptotic properties of the estimators presented in Section \ref{subsec:estimand-estimator} in the context of two-stage stratified experiments with a fixed number of large strata in the first stage of the experimental design. Specifically, in the first stage, clusters are partitioned into a fixed number of strata such that the number of clusters within each stratum grows as the total number of clusters increases. Formally, denote by $S^{(G)}=(S_1, \dots, S_G)$ the vector of strata on clusters, constructed from the observed, baseline covariates $C_{g}$ and cluster size $N_g$ for $g$th cluster using a function $S:\text{supp}((C_g,N_g)) \rightarrow \mathcal{S}$, where $\mathcal{S}$ is a finite set. Additionally, the second-stage design adheres to the specifications outlined in Section \ref{sec:main-results}.
\begin{example}
    Section \ref{sec:empirical} presents an illustrative empirical example of such a large-strata experiment conducted by \cite{foos2017}. In the first stage of their experiment, 5,190 two-voter households (i.e., clusters of size 2) were categorized into three strata: ``Labour'' supporter, ``rival party'' supporter, and those ``unattached'' to any party. Within each stratum, households were then randomly allocated to either treatment or control groups. In the subsequent stage, one member from the households in the treatment group was given the treatment.
\end{example}

First of all, I provide notations for the quantity of imbalance for each stratum. For $s \in \mathcal{S}$, let
\begin{equation}\label{eqn:D_G(s)}
    D_G(s) = \sum_{1\leq g \leq G} (I\{H_g=\pi_2\} - \pi_1) I\{S_g = s\},
\end{equation}
where $\pi_1 \in (0,1)$ is the “target” proportion of clusters to assign to treatment in each stratum. My requirements on the treatment assignment mechanism for the first stage are summarized in the following assumption:
\begin{assumption}\label{as:assignment1}
The treatment assignment mechanism for the first-stage is such that
\begin{enumerate}
    \item[(a)] $W^{(G)} \perp H^{(G)} \mid  S^{(G)}$,
    \item[(b)] $\left\{ \left\{\frac{D_G(s)}{\sqrt{G}} \right\}_{s \in \mathcal{S}}  \mid S^{(G)}\right\} \xrightarrow{d} N(0, \Sigma_D)$ a.s., where
    \begin{equation*}
        \Sigma_D = \text{diag}\{p(s)\tau(s): s\in\mathcal{S} \}
    \end{equation*}
    with $0\leq \tau(s) \leq \pi_1 (1-\pi_1)$ for all $s\in\mathcal{S}$, and $p(s) = P\left\{S_g = s\right\}$.
\end{enumerate}
\end{assumption}
Assumption \ref{as:assignment1} (a) simply requires that the treatment assignment mechanism is a function only of the vector of strata and an exogenous randomization device. Assumption \ref{as:assignment1} (b) follows Assumption 2.2 (b) of \cite{Bugni2018}. This assumption is commonly satisfied by various experiment designs, such as Bernoulli trials, stratified block randomization, and Efron's biased-coin design, which are widely used in clinical trials and development economics.

The following theorem derives the asymptotic behavior of estimators for equally-weighted effects.
\begin{theorem}\label{thm:clt-large-strata}
Under Assumption \ref{as:interference}-\ref{ass:Q_G}, \ref{as:assignment2} and \ref{as:assignment1},
\begin{align}
    &\sqrt{G} \left( \hat\theta^{P}_1 - \theta^P_1(Q_G) \right) \xrightarrow{d} \mathcal{N}(0, V_3(1))~, \\ &\sqrt{G} \left( \hat\theta^{S}_1 - \theta^S_1(Q_G) \right) \xrightarrow{d} \mathcal{N}(0, V_3(0))~,\\
    &\sqrt{G} \left( \hat\theta^{P}_2 - \theta^P_2(Q_G) \right) \xrightarrow{d} \mathcal{N}(0, V_4(1))~, \\ &\sqrt{G} \left( \hat\theta^{S}_2 - \theta^S_2(Q_G) \right) \xrightarrow{d} \mathcal{N}(0, V_4(0))~,
\end{align}
where
\begin{align}\label{eqn:V3(z)}
    \begin{split}
    V_3(z) &= \frac{1}{\pi_1} \operatorname{Var}\left[\bar{Y}_{g}(z,\pi_2)\right] + \frac{1}{1-\pi_1} \operatorname{Var}\left[\bar{Y}_{g}(0,0)\right] - \pi_1(1-\pi_1)  E\left[\left( \frac{1}{\pi_1} m_{z,\pi_2}\left(S_{g}\right) + \frac{1}{1-\pi_1}m_{0,0}\left(S_{g}\right) \right)^2\right] \\
    &\quad\quad\quad + E\left[\tau\left(S_{g}\right)\left(\frac{1}{\pi_1} m_{z,\pi_2}\left(S_{g}\right) +\frac{1}{1-\pi_1} m_{0,0}\left(S_{g}\right)\right)^{2}\right]~,
    \end{split}
\end{align}
and
\begin{align}\label{eqn:V4(z)}
    \begin{split}
    V_4(z) &= \frac{1}{\pi_1} \var[\tilde Y_g(z,\pi_2)] + \frac{1}{1-\pi_1} \var[\tilde Y_g(0,0)] - \pi_1(1-\pi_1) E\left[\left(\frac{1}{\pi_1}E[\tilde Y_g(z,\pi_2)\mid S_g] + \frac{1}{1-\pi_1}E[\tilde Y_g(0,0)\mid S_g] \right)^2\right] \\
    &\quad\quad\quad + E\left[\tau(S_g) \left(\frac{1}{\pi_1}E[\tilde Y_g(z,\pi_2)\mid S_g] + \frac{1}{1-\pi_1}E[\tilde Y_g(0,0)\mid S_g] \right)^2 \right]~~.
    \end{split}
\end{align}
\end{theorem}
\begin{remark}\label{remark:discuss-variance-expression}
An alternative variance expression, analogous to equation (15) in \cite{Bugni2018}, is:
\begin{align}\label{eqn:V3(z)2}
    \begin{split}
    V_3(z) &= \frac{1}{\pi_1} \operatorname{Var}\left[\check{Y}_{g}(z,\pi_2)\right]+\frac{1}{1-\pi_1} \operatorname{Var}\left[\check{Y}_{g}(0,0)\right] + E\left[\left(m_{z,\pi_2}\left(S_{g}\right) -m_{0,0}\left(S_{g}\right)\right)^{2}\right]\\
    &\quad\quad\quad + E\left[\tau\left(S_{g}\right)\left(\frac{1}{\pi_1} m_{z,\pi_2}\left(S_{g}\right) +\frac{1}{1-\pi_1} m_{0,0}\left(S_{g}\right)\right)^{2}\right]~,
\end{split}
\end{align}
where $\check{Y}_{g}(z,h) = \bar{Y}_g(z,h) - E[\bar{Y}_g(z,h)\mid S_g]$. By comparing (\ref{eqn:V3(z)2}) with the variance expression in \cite{Bugni2018},  we conclude that the asymptotic variance in Theorem \ref{thm:clt-large-strata} corresponds exactly to the asymptotic variance of the difference-in-means estimator for covariate-adaptive experiments with individual-level ``one-stage'' assignment, as in \cite{Bugni2018}. In fact, when $P(N_g=1)=1$ and $\pi_2=1$, $V_3(1)$ collapses to their variance expression.
\end{remark}

The widely used regression method with cluster-robust variance estimator is potentially conservative for matched tuples designs (see Appendix \ref{sec:app-wols}). Therefore, I aim to develop asymptotically exact methods based on my theoretical results. First, I present variance estimators for $V_3(z)$. A natural estimator of $V_3(z)$ may be constructed by replacing population quantities with their sample counterparts. For $z \in \{0, 1\}$, Let
\begin{align*}
    &\bar{Y}_{1, z} = \frac{1}{G_T} \sum_{1\leq g\leq G} \bar Y_{g}^{z} I\left\{ H_g=\pi_2\right\}~, \quad \bar{Y}_{0, z} = \frac{1}{G_C} \sum_{1\leq g\leq G} \bar Y_{g}^{z} I\left\{ H_g=0\right\}~, \\
    &\hat{\mu}_{1, z}(s)=\frac{1}{G_{1}(s)} \sum_{1 \leq g \leq G} \bar Y_g^{z} I\left\{ H_g=\pi_2, S_g=s\right\}~, \quad \hat{\mu}_{0, z}(s)=\frac{1}{G_{0}(s)} \sum_{1 \leq g \leq G} \bar Y_g^{z} I\left\{ H_g=0, S_g=s\right\}~,
\end{align*}
where $G_{1}(s) = |\{1\leq g \leq G: H_g=\pi_2, S_g =s\} |$ and $G_{0}(s) = |\{1\leq g \leq G: H_g=0, S_g =s\} |$. Then, define $G(s) = |\{1\leq g \leq G: S_g =s\} |$. With this notation, the following estimators can be defined:
\begin{align}\label{eqn:V3(z)hat}
    \begin{split}
    \hat V_3(z) &= \frac{1}{\pi_1}\left(\frac{1}{G_T} \sum_{1 \leq g \leq G} \left(\bar Y_g^{z}\right)^2 I\{H_g=\pi_2\}-\sum_{s \in \mathcal{S}} \frac{G(s)}{G} \hat{\mu}_{1,z}(s)^2\right)\\
    &\quad + \frac{1}{1-\pi_1}\left(\frac{1}{G_C} \sum_{1 \leq g \leq G} \left(\bar Y_g^{z}\right)^2I\{H_g=0\}-\sum_{s \in \mathcal{S}} \frac{G(s)}{G} \hat{\mu}_{0,0}(s)^2\right) \\
    &\quad\quad + \sum_{s \in \mathcal{S}} \frac{G(s)}{G}\left(\left(\hat{\mu}_{1,z}(s)-\bar{Y}_{1, z}\right)-\left(\hat{\mu}_{0, 0}(s)-\bar{Y}_{0, 0}\right)\right)^2\\
    &\quad\quad\quad + \sum_{s \in \mathcal{S}} \tau(s) \frac{G(s)}{G}\left(\frac{1}{\pi}\left(\hat{\mu}_{1, z}(s)-\bar{Y}_{1, z}\right)+\frac{1}{1-\pi}\left(\hat{\mu}_{0, 0}(s)-\bar{Y}_{0, 0}\right)\right)^2~.
    \end{split}
\end{align}

The estimator for $V_4(z)$ follows the same approach as $\hat V_3(z)$, while additionally requires estimation for terms associated with $\tilde Y_g(z,h)$. Let $\tilde Y_g^z$ denote the observed adjusted outcome.
\begin{align*}
    \tilde{Y}_g^z = \frac{N_{g}}{ \frac{1}{G}\sum_{1\leq g \leq G} N_g }\left(\bar{Y}_{g}^z-\frac{ \frac{1}{G_g} \sum_{1\leq j \leq G} \bar{Y}_{j}^z I\{H_g = H_j\} N_j}{\frac{1}{G} \sum_{1\leq j \leq G} N_{j}} \right)~,
\end{align*}
where $G_g = \sum_{1\leq j \leq G} I\{H_g = H_j\}$. For $z \in \{0, 1\}$, Let
\begin{align*}
    &\tilde{\mu}_{1, z}(s)=\frac{1}{G_{1}(s)} \sum_{1 \leq g \leq G} \tilde Y_g^{z} I\left\{ H_g=\pi_2, S_g=s\right\}~,\\
    &\tilde{\mu}_{0, z}(s)=\frac{1}{G_{0}(s)} \sum_{1 \leq g \leq G} \tilde Y_g^{z} I\left\{ H_g=0, S_g=s\right\}~.
\end{align*} 
To estimate $V_4(z)$, I propose the exact same estimator as $\hat V_3(z)$ by simply replacing $\bar Y_g^z$ with $\tilde Y_g^z$. Thus, the following estimators can be defined:
\begin{align}\label{eqn:V4(z)hat}
    \begin{split}
        \hat V_4(z) &= \frac{1}{\pi_1}\left(\frac{1}{G_T} \sum_{1 \leq g \leq G} \left(\tilde Y_g^{z}\right)^2 I\{H_g=\pi_2\}-\sum_{s \in \mathcal{S}} \frac{G(s)}{G} \tilde{\mu}_{1,z}(s)^2\right)\\
        &\quad + \frac{1}{1-\pi_1}\left(\frac{1}{G_C} \sum_{1 \leq g \leq G} \left(\tilde Y_g^{z}\right)^2I\{H_g=0\}-\sum_{s \in \mathcal{S}} \frac{G(s)}{G} \tilde{\mu}_{0,0}(s)^2\right) \\
        &\quad\quad + \sum_{s \in \mathcal{S}} \frac{G(s)}{G}\left(\tilde{\mu}_{1,z}(s)-\tilde{\mu}_{0, 0}(s)\right)^2 + \sum_{s \in \mathcal{S}} \tau(s) \frac{G(s)}{G}\left(\frac{1}{\pi}\tilde{\mu}_{1, z}(s)+\frac{1}{1-\pi}\tilde{\mu}_{0, 0}(s)\right)^2~.
    \end{split}
\end{align}
Then, the following consistency result for variance estimators $\hat V_3(z)$ and $\hat V_4(z)$ can be obtained:
\begin{theorem}\label{thm:variance-estimator-ca2}
Under Assumption \ref{as:interference}-\ref{ass:Q_G}, \ref{as:assignment2} and \ref{as:assignment1}, as $n \to \infty$, $\hat V_3(z) \xrightarrow{P} V_3(z)$ and $\hat V_4(z) \xrightarrow{P} V_4(z)$ for $z \in \{0,1\}$.
\end{theorem}
As pointed out by \cite{athey2017} and \cite{matched-tuple}, introducing replicates for each treatment arm in a matched tuples design can improve the finite sample performance for the adjusted $t$-tests based on $\hat V_1(z)$ and $\hat V_2(z)$.\footnote{When there are duplicates, I no longer need to form ``pairs of pairs'' for variance estimation. Instead, I could replace $\hat \rho_n^z(h, h)$ by $$\tilde \rho_n^z(h, h) =  \frac{2}{n} \sum_{1 \leq j \leq \lfloor n / 2 \rfloor} \frac{1}{k^2(h)}\Big ( \sum_{i \in \lambda_{j}} \bar Y_i^z I \{H_i = h\} \Big )~.$$} This motivates the use of variance estimators based on ``large tuples''. To that extent, $\hat V_3(z)$ and $\hat V_4(z)$, which take advantage of all observations within a stratum at the same time, are preferable for experiments with large strata (see Remark \ref{remark:compare-variance-app}). In practice, the choice of variance estimators depends on the sizes of the strata. Specifically, $\hat V_3(z)$ and $\hat V_4(z)$, whose consistency relies on large numbers of observations within each stratum, are suitable for experiments with large strata, while $\hat V_1(z)$ and $\hat V_2(z)$ are suitable for experiments with small strata.\footnote{In practice, most experimental designs either involve stratification on a limited number of categorical variables or matching units into groups of fewer than five. However, when decision-making is complex, choosing \(\hat V_1(z)\) and \(\hat V_2(z)\) is advisable. For instance, if stratification on a few categorical variables results in some strata having insufficient observations for reliable asymptotic analysis, then \(\hat V_1(z)\) and \(\hat V_2(z)\) become essential.} From this perspective, it is useful to divide stratified experiments into ``large strata'' and ``small strata'' and consider two separate sets of variance estimators.

\begin{remark}[Comparison of Variance Estimators]\label{remark:compare-variance-app}
    For experiments with large strata, $\hat V_3(z)$ and $\hat V_4(z)$ could be more efficient than $\hat V_1(z)$ and $\hat V_2(z)$. Consider estimating terms like $E\left[ E[\bar{Y}_{g}(z,h) \mid S_g]^2 \right]$. The estimators $\hat V_3(z)$ and $\hat V_4(z)$ utilize:
    \begin{equation*}
        \hat\nu_z = \sum_{s \in \mathcal{S}} \frac{G(s)}{G} \tilde{\mu}_{1,z}(s)^2 ~.
    \end{equation*}
    In contrast, $\hat V_1(z)$ and $\hat V_2(z)$ use:
    \begin{equation*}
        \hat\omega_z = \hat \rho_n^z(h, h) = \frac{2}{n} \sum_{1 \leq j \leq \lfloor n / 2 \rfloor} \frac{1}{k^2(h)} \left(\sum_{i \in \lambda_{2j-1}} \bar Y_i^z I \{H_i = h\} \right) \left(\sum_{i \in \lambda_{2j}} \bar Y_i^z I \{H_i = h\} \right) ~.
    \end{equation*}
    Consider a simple example with only one stratum ($S_g=1$) and all units treated ($\pi_2=1$), then:
    \begin{align*}
        \hat \nu_1  &= \left( \frac{1}{n} \sum_{i=1}^n \bar Y_{i}^1 \right)^2,\\
        \hat \omega_1 &= \frac{2}{n} \sum_{1\leq i \leq \lfloor n / 2 \rfloor} \bar Y_{i}^1 \bar Y_{i+1}^1 ~.
    \end{align*}
    $\hat \nu_1$ averages all $\bar Y_{i}^1 \bar Y_{j}^1$ for $1 \leq i,j \leq n$, while $\hat \omega_1$ averages $\bar Y_{i}^1 \bar Y_{i+1}^1$ for $1 \leq i \leq \lfloor n / 2 \rfloor$. Assuming $E[Y_{i}^1] = 0$ for all $i$, the finite variance of $\hat \nu_1$ is
    \begin{align*}
        \var(\hat \nu_1) = \frac{1}{n^3} \var \left( (Y_{i}^1)^2 \right) + \frac{n-1}{n^3} \var(Y_{i}^1)^2 ~,
    \end{align*}
    and for $\hat \omega_1$:
    \begin{align*}
        \var(\hat \omega_1) = \frac{2}{n} \var(Y_{i}^1)^2 ~.
    \end{align*}
    Consequently, $\var(\hat \nu_1) = O(1/n^2)$ and $\var(\hat \omega_1) = O(1/n)$, indicating that $\hat V_3(z)$ and $\hat V_4(z)$ could indeed be more efficient.
\end{remark}

\newpage
\section{Proofs of Main Results}
\subsection{Proof for Equally-Weighed Estimator in Theorem \ref{thm:clt-large-strata}}\label{proof:equally-weighted-ca}
To begin with, both estimators can be written as follows.
\begin{align*}
    \hat{\theta}^P_1 &= \frac{1}{G_T}\sum_{1\leq g \leq G} I\{H_g = \pi_2 \} \bar{Y}_g(1,\pi_2)   - \frac{1 }{G_C} \sum_{1\leq g \leq G} I\{H_g = 0 \} \bar{Y}_g(0,0)~, \\
     \hat{\theta}^S_1 &=\frac{1}{G_T}\sum_{1\leq g \leq G} I\{H_g = \pi_2 \} \bar{Y}_g(0,\pi_2)  - \frac{1 }{G_C} \sum_{1\leq g \leq G} I\{H_g = 0 \} \bar{Y}_g(0,0)~.
\end{align*}
By Lemma 5.1 of \cite{Bugni2022} and Assumption \ref{ass:Q_G} (a)-(b), we have $((\bar{Y}_g(1,\pi_2), \bar{Y}_g(0,\pi_2), \bar{Y}_g(0,0)): 1\leq g \leq G)$ being an i.i.d sequence of random variables. Then, by the law of iterated expectation and  Assumption \ref{ass:Q_G} (f) and  \ref{as:assignment2},
\begin{align*}
    E\left[\bar{Y}_g(1,\pi_2)\right] &= E\left[E\left[\frac{1}{M_g^1} \sum_{i \in \mathcal{M}_g} Y_{i,g}(1,\pi_2) Z_{i,g}(\pi_2) \mid B_g, \mathcal{M}_g \right]\right] \\
    &= E\left[\frac{1}{M_g^1}\sum_{i \in \mathcal{M}_g} E\left[  Y_{i,g}(1,\pi_2) Z_{i,g}(\pi_2) \mid B_g, \mathcal{M}_g \right]\right] \\
    &= E\left[\frac{1}{M_g^1} \sum_{i \in \mathcal{M}_g} E\left[Y_{i,g}(1,\pi_2)  \mid B_g, \mathcal{M}_g \right] E\left[ Z_{i,g}(\pi_2) \mid B_g  \right]\right] \\
    &= E\left[\frac{1}{M_g} \sum_{i \in \mathcal{M}_g} E\left[Y_{i,g}(1,\pi_2)  \mid B_g, \mathcal{M}_g \right]\right] = E\left[\frac{1}{N_g}\sum_{1 \leq i \leq N_g} Y_{i,g}(1, \pi_2) \right]~.
\end{align*}
Similarly, 
\begin{equation*}
    E[\bar{Y}_g(0,\pi_2)] = E\left[\frac{1}{N_g}\sum_{1 \leq i \leq N_g} Y_{i,g}(0, \pi_2) \right] \text{, }\quad  E[\bar{Y}_g(0,0)] = E\left[\frac{1}{N_g}\sum_{1 \leq i \leq N_g} Y_{i,g}(0, 0) \right]~.
\end{equation*}
Thus, $\theta^P_1 = E\left[\bar{Y}_g(1,\pi_2)\right] - E[\bar{Y}_g(0,0)]$ and $\theta^S_1 = E\left[\bar{Y}_g(0,\pi_2)\right] - E[\bar{Y}_g(0,0)]$.
By Assumption \ref{as:assignment2} and \ref{as:assignment1}, we have
\begin{equation*}
    H^{(G)} \independent ((\bar{Y}_g(1,\pi_2), \bar{Y}_g(0,\pi_2), \bar{Y}_g(0,0)): 1\leq g \leq G) \mid S^{(G)}~.
\end{equation*}
By Assumption \ref{ass:Q_G} (c)-(d),
\begin{equation*}
    E\left[\bar{Y}_g^2(1,\pi_2)\right] = E\left[\left(\frac{1}{M_g^1} \sum_{i \in \mathcal{M}_g} Y_{i,g}(1,\pi_2) Z_{i,g}(\pi_2) \right)^2\right] \leq E\left[\left(\max_{1\leq i \leq N_g} Y_{i,g}(1,\pi_2) \right)^2 \right] < \infty~.
\end{equation*}
Same conclusions hold for $\bar{Y}_g^2(0,\pi_2)$ and $\bar{Y}_g^2(0,0)$. Then, the result follows directly by Theorem 4.1 of \cite{Bugni2018} and Lemma \ref{lem:E_bounded} and Assumption \ref{ass:Q_G} and \ref{as:assignment1}.
\qed 
\subsection{Proof for Size-Weighed Estimator in Theorem \ref{thm:clt-large-strata}}\label{proof:size-weighted-ca}
To preserve space, I only present proof for primary effect as the proof for spillover effect follows the same argument. Define $\mathbf L_G = \left( \mathbf L_G^{\rm YN1}, \mathbf L_G^{\rm N1}, \mathbf L_G^{\rm YN0}, \mathbf L_G^{\rm N0}\right)$ as follows.
\begin{align*}
    &\mathbf L_G^{\rm YN1} := \frac{1}{G_T}\sum_{1 \leq g \leq G} \left( \bar{Y}_g(1,\pi_2) N_g - E\left[\bar{Y}_g(1,\pi_2) N_g\right]\right) I\{H_g=\pi_2\}~, \\
    &\mathbf L_G^{\rm N1} := \frac{1}{G_T}\sum_{1 \leq g \leq G} \left( N_g -E\left[ N_g\right]\right) I\{H_g=\pi_2\}~, \\
    &\mathbf L_G^{\rm YN0} := \frac{1}{G_C}\sum_{1 \leq g \leq G} \left( \bar{Y}_g(0,0) N_g -E\left[\bar{Y}_g(0,0) N_g\right]\right)I\{H_g=0\}~, \\
    &\mathbf L_G^{\rm N0} := \frac{1}{G_C}\sum_{1 \leq g \leq G} \left( N_g - E\left[ N_g\right]\right) I\{H_g=0\} ~.
\end{align*}
By the law of iterated expectation and  Assumption \ref{ass:Q_G} (f),
\begin{align*}
    E\left[\bar{Y}_g(1,\pi_2) N_g\right] &= E\left[N_g E\left[\bar{Y}_g(1,\pi_2)  \mid N_g\right]\right] = E\left[N_g E\left[\frac{1}{M_g} \sum_{i \in \mathcal{M}_g} Y_{i,g}(1,\pi_2)  \mid N_g\right]\right] \\
    &= E\left[N_g E\left[\frac{1}{N_g} \sum_{1 \leq i \leq N_g} Y_{i,g}(1,\pi_2)  \mid N_g\right]\right] = E\left[ \sum_{1 \leq i \leq N_g} Y_{i,g}(1,\pi_2)  \right].
\end{align*}
Thus,
\begin{equation*}
    \theta^P_2 = \frac{E\left[\bar{Y}_g(1,\pi_2) N_g\right]}{E[N_g]} - \frac{E\left[\bar{Y}_g(0,0) N_g\right]}{E[N_g]} \text{ and } \theta^S_2 = \frac{E\left[\bar{Y}_g(0,\pi_2) N_g\right]}{E[N_g]} - \frac{E\left[\bar{Y}_g(0,0) N_g\right]}{E[N_g]}~.
\end{equation*}
Note that $\frac{G_T}{G} = \frac{D_G}{G} + \pi_1$. Thus,
\begin{equation*}
    \sqrt{G} \mathbf L_G^{\rm YN1} = \left(\frac{D_{G}}{G}+\pi_1\right)^{-1}\left(1-\pi_1-\frac{D_{G}}{G}\right)^{-1} \frac{1}{\sqrt{G}} \sum_{g=1}^{G} \left(\left(1-\pi -\frac{D_G}{G} \right)\left( \bar{Y}_g(1,\pi_2) N_g - \mu_1\right) I\{H_g=\pi_2\} \right)~,
\end{equation*}
where $E\left[\bar{Y}_g(1,\pi_2) N_g\right] = \mu_1$. By Lemma B.1 and B.3 of \cite{Bugni2018}, Lemma \ref{lem:E_bounded} and Assumption \ref{ass:Q_G} and \ref{as:assignment1}, we have
\begin{equation*}
    \sqrt{G} \mathbf L_G^{\rm YN1} = \left(\pi_1 (1-\pi_1)\right)^{-1} \underbrace{\frac{1}{\sqrt{G}} \sum_{1\leq g\leq G} \left(\left(1-\pi_1 \right)\left(\bar{Y}_g(1,\pi_2) N_g - E\left[\bar{Y}_g(1,\pi_2) N_g\right]\right) I\{H_g=\pi_2\} \right)}_{:=L_G^{\rm YN1}}  + o_P(1)~.
\end{equation*}
Similarly,
\begin{align*}
    &\sqrt{G} \mathbf L_G^{\rm N1} = \left(\pi_1 (1-\pi_1)\right)^{-1} \underbrace{\frac{1}{\sqrt{G}} \sum_{1\leq g\leq G} \left(\left(1-\pi_1 \right)\left( N_g -E\left[ N_g\right]\right) I\{H_g=\pi_2\} \right)}_{:=L_G^{\rm N1}} + o_P(1)~,\\
    &\sqrt{G} \mathbf L_G^{\rm YN0} = \left(\pi_1 (1-\pi_1)\right)^{-1} \underbrace{\frac{1}{\sqrt{G}} \sum_{1\leq g\leq G} \left(\pi_1\left( \bar{Y}_g(0,0) N_g - E\left[\bar{Y}_g(0,0) N_g\right]\right) I\{H_g=0\} \right)}_{:=L_G^{\rm YN0}} + o_P(1)~, \\
    &\sqrt{G} \mathbf L_G^{\rm N0} = \left(\pi_1 (1-\pi_1)\right)^{-1} \underbrace{\frac{1}{\sqrt{2n}} \sum_{1\leq i\leq 2n} \left(\pi_1\left( N_g -E\left[ N_g\right]\right) I\{H_g=0\} \right)}_{:=L_G^{\rm N0}}  + o_P(1)~.
\end{align*}
Define
\begin{align*}
    &\Tilde{Y}_g^{\rm N}(z,h) = \bar Y_g(z, h)N_g - E\left[\bar Y_g(z, h)N_g\mid S_g\right]~, \\
    &\Tilde N_g= N_g - E\left[ N_g\mid S_g\right]~, \\
    &m_{z,h}^{\rm YN}(S_{g})=E\left[Y_g(z, h)N_g\mid S_g\right]-E\left[Y_g(z, h)N_g\right]~, \\
    &m^{\rm N}(S_g)=E\left[N_g\mid S_g\right]-E\left[N_g\right]~,
\end{align*}
and consider the following decomposition for $L_G^{\rm YN1}$:
\begin{align*}
    L_G^{\rm YN1} &= R_{n,1} + R_{n,2} + R_{n,3} \\
    &=\frac{\pi_1 (1-\pi_1)}{\sqrt{G}} \sum_{1\leq g \leq G} \frac{1}{\pi_1} \Tilde{Y}_g^{\rm N}(1,\pi_2) I\{H_g=\pi_2\} + \pi_1 (1-\pi_1) \sum_{s\in \mathcal{S}} \frac{D_G(s)}{\sqrt{G}} \frac{1}{\pi_1} m_{1,\pi_2}^{\rm YN}(S_g) \\
    &\quad + \pi_1 (1-\pi_1)\sum_{s \in \mathcal{S}}\sqrt{G}\left(\frac{G(s)}{G} - p(s) \right) m_{1,\pi_2}^{\rm YN}(S_g)~.
\end{align*}
Similarly, we have the same decomposition for $L_G^{\rm YN0}, L_G^{\rm N1}, L_G^{\rm N0}$. Define
\begin{align*}
    &\mathbf{d} := \left(\frac{D_G(s)}{\sqrt{G}}: s\in \mathcal{S} \right)^\prime\\
    &\mathbf{n} := \left(\sqrt{G} \left(\frac{G(s)}{G} - p(s) \right): s\in \mathcal{S} \right)^\prime\\
    &\mathbf{m}^{\rm YN}_{z,h} := \left(E\left[m_{z,h}^{\rm YN}(C_g) \mid S_g=s\right] : s\in \mathcal{S}\right)^\prime\\
    &\mathbf{m}^{\rm N} := \left(E\left[m^{\rm N}(C_g) \mid S_g=s\right] : s\in \mathcal{S}\right)^\prime~.
\end{align*}
Then, we can write
\begin{align*}
\scriptsize
    (\pi_1 (1-\pi_1))^{-1}\begin{pmatrix}
L_G^{\rm YN1}  \\
L_G^{\rm N1} \\
L_G^{\rm YN0} \\
L_G^{\rm N0}
\end{pmatrix}
=  \underbrace{\begin{pmatrix}
1 & 0 & 0 & 0 &\frac{1}{\pi_1} \left(\mathbf{m}_{1,\pi_2}^{\rm YN}\right)^\prime & \left(\mathbf{m}_{1,\pi_2}^{\rm YN}\right)^\prime\\
0 & 1 & 0 & 0 &\frac{1}{\pi_1} \left(\mathbf{m}^{\rm N}\right)^\prime & \left(\mathbf{m}^{\rm N}\right)^\prime \\
0 & 0 & 1 & 0 &-\frac{1}{1-\pi_1} \left(\mathbf{m}_{0,0}^{\rm YN}\right)^\prime & \left(\mathbf{m}_{0,0}^{\rm YN}\right)^\prime \\
0 & 0 & 0 & 1 &-\frac{1}{1-\pi_1} \left(\mathbf{m}^{\rm N}\right)^\prime & \left(\mathbf{m}^{\rm N}\right)^\prime
\end{pmatrix}}_{:=M^\prime} \underbrace{\begin{pmatrix}
\frac{1}{\sqrt{G}} \sum_{g=1}^{G} \frac{1}{\pi_1} \Tilde{Y}_g^{\rm YN}(1,\pi_2) I\{H_g=\pi_2\} \\
\frac{1}{\sqrt{G}} \sum_{g=1}^{G} \frac{1}{\pi_1}  \Tilde{N}_g I\{H_g=\pi_2\} \\
\frac{1}{\sqrt{G}} \sum_{g=1}^{G} \frac{1}{1-\pi_1}  \Tilde{Y}_g^{\rm YN}(0,0) I\{H_g=0\} \\
\frac{1}{\sqrt{G}} \sum_{g=1}^{G} \frac{1}{1-\pi_1} \Tilde{N}_g I\{H_g=0\}\\
\mathbf{d} \\
\mathbf{n}
\end{pmatrix}}_{:= \mathbf{y}_n}
\end{align*}
Following Lemma B.2 from \cite{Bugni2018}, we have
\begin{equation*}
\mathbf{y}_n \xrightarrow{d} \mathcal{N}(0, \Sigma)~,
\end{equation*}
where
\begin{equation*}
    \Sigma = \begin{pmatrix}
    \Sigma_1 & 0 & 0 & 0\\
    0 & \Sigma_0 & 0 & 0 \\
    0 & 0 & \Sigma_D & 0 \\
    0 & 0 & 0 & \Sigma_{N} 
    \end{pmatrix}~,
\end{equation*}
for
\begin{align*}
    &\Sigma_1 = \begin{pmatrix}
    \frac{\text{Var}\left[\Tilde{Y}_g^{\rm YN}(1,\pi_2)\right]}{\pi_1} & \frac{E\left[\Tilde{Y}_g^{\rm YN}(1,\pi_2) N_g\right]}{\pi_1} \\
    \frac{E\left[\Tilde{Y}_g^{\rm YN}(1,\pi_2) N_g\right]}{\pi_1} & \frac{\text{Var}\left[N_g\right]}{\pi_1}
    \end{pmatrix},  &\Sigma_0 = \begin{pmatrix}
    \frac{\text{Var}\left[\Tilde{Y}_g^{\rm YN}(0,0)\right]}{1-\pi_1} & \frac{E\left[\Tilde{Y}_g^{\rm YN}(0,0) N_g\right]}{1-\pi_1} \\
    \frac{E\left[\Tilde{Y}_g^{\rm YN}(0,0) N_g\right]}{1-\pi_1} & \frac{\text{Var}\left[N_g\right]}{1-\pi_1}
    \end{pmatrix}~, \\
    &\Sigma_D = \text{diag}\left(p(s)\tau(s): s\in\mathcal{S} \right),  &\Sigma_N = \text{diag}\left(p(s): s\in\mathcal{S} \right)- \left(p(s): s\in\mathcal{S} \right) \left(p(s): s\in\mathcal{S} \right)^\prime ~.
\end{align*}
Let $\mathbf{m}(S_g) = \left(m_{1,\pi_2}^{\rm YN}(S_g), m_0^{\rm N}(S_g), m_{0,0}^{\rm YN}(S_g), m_0^{\rm N}(S_g) \right)^\prime$. We have
\begin{equation*}
    \mathbb V = M^\prime \Sigma M = \mathbb V_1 + \mathbb V_2 + \mathbb V_3,
\end{equation*}
where
\begin{align*}
    \mathbb V_1 &=  \begin{pmatrix}
    \frac{1}{\pi_1} \text{Var}\left[\Tilde{Y}_g^{\rm YN}(1,\pi_2)\right]   & \frac{1}{\pi_1}E\left[\Tilde{Y}_g^{\rm YN}(1,\pi_2) N_g\right]  & 0 & 0 \\
    \frac{1}{\pi_1} E\left[\Tilde{Y}_g^{\rm YN}(1,\pi_2) N_g\right]  & \frac{1}{\pi_1} \text{Var}\left[N_g\right] &  0 &  0 \\
    0 & 0 & \frac{1}{1-\pi_1} \text{Var}\left[\Tilde{Y}_g^{\rm YN}(0,0)\right] & \frac{1}{1-\pi_1} E\left[\Tilde{Y}_g^{\rm YN}(0,0) N_g\right] \\
    0 & 0 & \frac{1}{1-\pi_1} E\left[\Tilde{Y}_g^{\rm YN}(0,0) N_g\right] &  \frac{1}{1-\pi_1} \text{Var}\left[N_g\right]
    \end{pmatrix}~,\\
    \mathbb V_2 &= \operatorname{Var}\left[\mathbf{m}(S_g)\right]~,\\
    \mathbb V_3 &= E\left[ \tau(S_g) \left(\Lambda \mathbf{m}(S_g) \mathbf{m}(S_g)^\prime \Lambda \right) \right] \text{ with } \Lambda = \text{diag}\left(\frac{1}{\pi_1}, \frac{1}{\pi_1}, -\frac{1}{1-\pi_1},-\frac{1}{1-\pi_1} \right) ~.
\end{align*}
Alternatively,
\begin{align*}
\mathbb V_{11} & =  \frac{1}{\pi_1} \var \left[\bar{Y}_g(1,\pi_2) N_g\right] -  \frac{1-\pi_1}{\pi_1} \var \left[E\left[\bar{Y}_g(1,\pi_2) N_g \mid S_g \right]\right] \\
& \quad + E\left[\frac{\tau(S_g)}{\pi_1^2} \left( E[\bar{Y}_g(1,\pi_2)N_g\mid S_g] - E[\bar{Y}_g(1,\pi_2)N_g] \right)^2 \right] \\
\mathbb V_{12} & = \frac{1}{\pi_1}\cov[\bar{Y}_g(1,\pi_2) N_g, N_g] - \frac{1-\pi_1}{\pi_1}\cov[E[\bar{Y}_g(1,\pi_2) N_g | S_g], E[N_g | S_g]] \\
& \quad + E\left[\frac{\tau(S_g)}{\pi_1^2} \left( E[\bar{Y}_g(1,\pi_2)N_g\mid S_g] - E[\bar{Y}_g(1,\pi_2)N_g] \right) \left( E[N_g\mid S_g] - E[N_g] \right)  \right] \\
\mathbb V_{13} & =  \cov[E[\bar{Y}_g(1,\pi_2) N_g | S_g], E[\bar{Y}_g(0,0) N_g | S_g]] \\
& \quad - E\left[\frac{\tau(S_g)}{\pi_1(1-\pi_1)} \left( E[\bar{Y}_g(1,\pi_2)N_g\mid S_g] - E[\bar{Y}_g(1,\pi_2)N_g] \right) \left( E[\bar{Y}_g(0,0)N_g\mid S_g] - E[\bar{Y}_g(0,0)N_g] \right)  \right]\\
\mathbb V_{14} & =  \cov[E[\bar{Y}_g(1,\pi_2) N_g | S_g], E[N_g | S_g]] \\
& \quad - E\left[\frac{\tau(S_g)}{\pi_1(1-\pi_1)} \left( E[\bar{Y}_g(1,\pi_2)N_g\mid S_g] - E[\bar{Y}_g(1,\pi_2)N_g] \right) \left( E[N_g\mid S_g] - E[N_g] \right)  \right]\\
\mathbb V_{22} & = \frac{1}{\pi_1}\var[N_g] - \frac{1-\pi_1}{\pi_1} \var[E[N_g | S_g]] \\
& \quad + E\left[\frac{\tau(S_g)}{\pi_1^2} \left( E[N_g\mid S_g] - E[N_g] \right)^2  \right]\\
\mathbb V_{23} & =  \cov[E[N_g | S_g], E[\bar{Y}_g(0,0) N_g | S_g]] \\
& \quad - E\left[\frac{\tau(S_g)}{\pi_1(1-\pi_1)} \left( E[N_g\mid S_g] - E[N_g] \right) \left( E[\bar{Y}_g(0,0)N_g\mid S_g] - E[\bar{Y}_g(0,0)N_g] \right)  \right]\\
\mathbb V_{24} & =  \cov[E[N_g | S_g], E[N_g | S_g]] \\
& \quad - E\left[\frac{\tau(S_g)}{\pi_1(1-\pi_1)} \left( E[N_g\mid S_g] - E[N_g] \right)^2  \right]\\
\mathbb V_{33} & = \frac{1}{1-\pi_1} \var[\bar{Y}_g(0,0) N_g] - \frac{\pi_1}{1-\pi_1} \var[E[\bar{Y}_g(0,0) N_g | S_g]] \\
& \quad + E\left[\frac{\tau(S_g)}{(1-\pi_1)^2}\left( E[\bar{Y}_g(0,0)N_g\mid S_g] - E[\bar{Y}_g(0,0)N_g] \right)^2  \right]\\
\mathbb V_{34} & = \frac{1}{1-\pi_1} \cov[\bar{Y}_g(0,0) N_g, N_g] - \frac{\pi_1}{1-\pi_1} \cov[E[\bar{Y}_g(0,0) N_g | S_g], E[N_g | S_g]] \\
& \quad + E\left[\frac{\tau(S_g)}{(1-\pi_1)^2}  \left( E[\bar{Y}_g(0,0)N_g\mid S_g] - E[\bar{Y}_g(0,0)N_g] \right) \left( E[N_g\mid S_g] - E[N_g] \right)  \right]\\
\mathbb V_{44} & = \frac{1}{1-\pi_1} \var[N_g] - \frac{\pi_1}{1-\pi_1}\var[E[N_g | S_g]]\\
& \quad + E\left[\frac{\tau(S_g)}{(1-\pi_1)^2} \left( E[N_g\mid S_g] - E[N_g] \right)^2  \right] ~.
\end{align*}
Therefore,
\begin{equation*}
    \sqrt{G}(\hat{\beta} - \beta) := \sqrt{G} \left(\mathbf L_G^{\rm YN1}, \mathbf L_G^{\rm N1}, \mathbf L_G^{\rm YN0}, \mathbf L_G^{\rm N0} \right)^\prime = (\pi (1-\pi))^{-1}\cdot\left(L_G^{\rm YN1}, L_G^{\rm N1}, L_G^{\rm YN0},L_G^{\rm N0} \right)^\prime + o_P(1) \xrightarrow{d} \mathcal{N}(0, \mathbb V) ~.
\end{equation*}
Let $g(x,y,z,w)= \frac{x}{y}-\frac{z}{w}$. Note that the Jacobian is
\[ D_g(x, y, z, w) = \Big ( \frac{1}{y}, - \frac{x}{y^2}, - \frac{1}{w}, \frac{z}{w^2} \Big )~. \]
By delta method, 
\begin{equation*}
    \sqrt{2n}(\hat{\theta}_2^P - \theta^P_2) = \sqrt{2n}(g(\hat{\beta}) - g(\beta)) \xrightarrow{d} \mathcal{N}(0, V_2(1) ) ~,
\end{equation*}
where
\begin{align*}
    &V_2(1) = D_g^\prime \left(\mathbb V_1 + \mathbb V_2 + \mathbb V_3\right) D_g  
\end{align*}
for
\begin{equation*}
    D_g = \left ( \frac{1}{\pi_1 E[N_g]}, - \frac{E[\bar{Y}_g(1,\pi_2)N_g]}{\pi_1 E[N_g]^2}, - \frac{1}{(1-\pi_1)E[N_g]}, \frac{E[\bar{Y}_g(0,0)N_g]}{(1-\pi_1)E[N_g]^2} \right )^\prime ~.
\end{equation*}
By simple calculation,
\begin{align*}
    &D_g^\prime  \left( \mathbb V_1 + \mathbb V_2\right) D_g =  \frac{1}{\pi_1} \var[\tilde Y_g(z,\pi_2)] + \frac{1}{1-\pi_1} \var[\tilde Y_g(0,0)]\\
    &\quad\quad\quad -  E\left[E\left[\sqrt{\frac{1-\pi_1}{\pi_1}} \tilde Y_g(z,\pi_2) + \sqrt{\frac{\pi_1}{1-\pi_1}}\tilde Y_i(0,0) \Bigg| S_g\right]^2\right] \\
    &D_g^\prime \mathbb V_3 D_g = E\left[ \tau(S_g) \left(\frac{m_{1,\pi_2}^{\rm YN}(S_g)}{\pi_1 E[N_g]} - \frac{E[Y_i(1)N_g] m^{\rm N}(S_g)}{\pi_1 E[N_g]^2} + \frac{m_{0,0}^{\rm YN}(S_g)}{(1-\pi_1)E[N_g] } - \frac{E[Y_i(0)N_g]m^{\rm N}(S_g)}{(1-\pi_1)E[N_g]^2} \right)^2 \right] \\
    &\quad\quad\quad\quad =  E\left[\tau(S_g) \left(\frac{1}{\pi_1}E[\tilde Y_g(z,\pi_2)\mid S_g] + \frac{1}{1-\pi_1}E[\tilde Y_g(0,0)\mid S_g] \right)^2 \right]~.
\end{align*}
Thus, the result follows.
\qed

\subsection{Proof of Theorem \ref{thm:variance-estimator-ca2}}\label{sec:proof-variance-ca}
The conclusion follows by continuous mapping theorem and by showing the following results:
\begin{enumerate}
    \item[(a)] $\frac{G(s)}{G} \xrightarrow{P} p(s)$.
    \item[(b)] $\frac{1}{G_a} \sum_{1 \leq g \leq G} \left(\bar Y_g^{z}\right)^r I\{ H_g=h\} \xrightarrow{P} E[\bar Y_g(z,h)^r]$ for $r, z \in \{0,1\}$ and $(a,h) \in \{(1,\pi_2),(0,0)\}$.
    \item[(c)] $\frac{1}{G_{a}(s)} \sum_{1 \leq g \leq G} \bar Y_g^{z} I\left\{H_g=h, S_g=s\right\} \xrightarrow{P} E[\bar Y_g(z,h) \mid S_g]$ for $z \in \{0,1\}$ and $(a,h) \in \{(1,\pi_2),(0,0)\}$.
    \item[(d)] $\frac{1}{G_a} \sum_{1 \leq g \leq G} \left(\tilde Y_g^{z}\right)^r I\{ H_g=h\} \xrightarrow{P} E[\tilde Y_g(z,h)^r]$ for $r, z \in \{0,1\}$ and $(a,h) \in \{(1,\pi_2),(0,0)\}$.
    \item[(e)] $\frac{1}{G_{a}(s)} \sum_{1 \leq g \leq G} \tilde Y_g^{z} I\left\{H_g=h, S_g=s\right\} \xrightarrow{P} E[\tilde Y_g(z,h) \mid S_g]$ for $z \in \{0,1\}$ and $(a,h) \in \{(1,\pi_2),(0,0)\}$
\end{enumerate}
By following the arguments in Appendix A.2 of \cite{Bugni2018}, Lemma \ref{lem:E_bounded} and Assumption \ref{ass:Q_G} and \ref{as:assignment1},  we conclude that (a), (b) and (c) hold. Next, I first show the results hold for $\tilde Y_g(z,h)$ and then analyze the difference between $\tilde Y_g(z,h)$ and adjusted version $\hat Y_g^z(h)$ defined as follows:
\begin{align}\label{eqn:hatY(h)}
    \begin{split}
    \hat Y_{g}^z(\pi_2) &= \frac{N_{g}}{ \frac{1}{G}\sum_{1\leq g \leq G} N_g }\left(\bar{Y}_{g}(z,\pi_2)-\frac{\frac{1}{G_T} \sum_{1\leq j \leq G} \bar{Y}_{j}(z,\pi_2) I\{H_j = \pi_2\} N_j}{\frac{1}{G} \sum_{1\leq j \leq G} N_j}\right)\\
    \hat Y_{g}^z(0) &= \frac{N_{g}}{ \frac{1}{G}\sum_{1\leq g \leq G} N_g }\left(\bar{Y}_{g}(0,0)-\frac{\frac{1}{G_C}\sum_{1\leq j \leq G} \bar{Y}_{j}(0,0) I\{H_j = 0\} N_j}{\frac{1}{G} \sum_{1\leq j \leq G} N_j}\right)~,
    \end{split}
\end{align}
for which the usual relationship still holds for adjusted outcomes, i.e. $\tilde Y_{g}^z = \sum_{h \in \{0,\pi_2\}} I\{H_g=h\} \hat Y_{g}^z(h)$. Note that
\begin{align*}
    E[\tilde Y_g(z,h)^2] &= E\left[\frac{N_g^2}{E[N_g]^2} \left ( \bar Y_g(z,h) - \frac{E[\bar Y_g(z,h) N_g]}{E[N_g]} \right )^2 \right] \leq 2 E\left[\frac{N_g^2}{E[N_g]^2} \left (  \bar Y_g(z,h)^2 +  \frac{E[\bar Y_g(z,h) N_g]^2}{E[N_g]^2} \right ) \right] \\
    &\leq 2 E\left[N_g^2   \bar Y_g(z,h)^2  \right] + 2 E[\bar Y_g(z,h) N_g]^2 E[N_g^2 ] < \infty~.
\end{align*}
where the first inequality holds by the fact $(a-b)^2 \leq 2 a^2 + 2 b^2$, the second inequality follows by the fact that $E[N_g] \geq 1$, and the last inequality follows by Lemma \ref{lem:E_bounded}. Therefore, again by following the arguments in Appendix A.2 of \cite{Bugni2018},  we conclude that for $r, z \in \{0,1\}$ and $(a,h,c) \in \{(1,\pi_2,T),(0,0,C)\}$,
\begin{align*}
    &\frac{1}{G_c} \sum_{1 \leq g \leq G} \tilde Y_g(z,h)^r I\{ H_g=h\} \xrightarrow{P} E[\tilde Y_g(z,h)^r] \\
    &\frac{1}{G_{a}(s)} \sum_{1 \leq g \leq G} \tilde Y_g(z,h) I\left\{H_g=h, S_g=s\right\} \xrightarrow{P} E[\tilde Y_g(z,h) \mid S_g]~,
\end{align*}
Finally, I show the difference between the above equations with $\tilde Y_g(z,h)$ and $\tilde Y_g^z$ go to zero. Here, I prove this for the following case,
\begin{equation}\label{eqn:variance-ca-diff}
    \frac{1}{G_T} \sum_{1 \leq g \leq G} \left(\tilde Y_g(1,\pi_2)^2 - \left(\tilde Y_g^{1}\right)^2\right)  I\{ H_g=\pi_2\} \xrightarrow{P} 0~;
\end{equation}
an analogous argument establishes the rest. Note that
\begin{align*}
    &\frac{1}{G_T} \sum_{1 \leq g \leq G} \left(\tilde Y_g(1,\pi_2)^2 - \left(\hat Y_g^{1}\right)^2\right)  I\{ H_g=\pi_2\} = \frac{1}{G_T} \sum_{1 \leq g \leq G} \left(\tilde Y_g(1,\pi_2) - \hat Y_g^{1}(\pi_2)\right)\left(\tilde Y_g(1,\pi_2) + \hat Y_g^{1}(\pi_2)\right)  I\{ H_g=\pi_2\} \\
    &= \frac{1}{G_T} \sum_{1 \leq g \leq G}  \left(\frac{1}{E[N_g]} - \frac{1}{ \frac{1}{G}\sum_{1\leq g \leq G} N_g } \right) \bar Y_g(1,\pi_2) N_g \left(\tilde Y_g(1,\pi_2) + \hat Y_g^{1}(\pi_2)\right)  I\{ H_g=\pi_2\} \\
    &\quad - \frac{1}{G_T} \sum_{1 \leq g \leq G} \left( \frac{\frac{1}{G}\sum_{1\leq g \leq G} \bar Y_g(1,\pi_2) I\{H_g=\pi_2\} N_g}{\left(\frac{1}{G} \sum_{1\leq g \leq G} N_g \right)^2} - \frac{E[\Bar{Y}_g(1,\pi_2) N_g]}{E[N_g]^2} \right) N_g \left(\tilde Y_g(1,\pi_2) + \hat Y_g^{1}(\pi_2)\right)  I\{ H_g=\pi_2\}~. 
\end{align*}
I then proceed to prove the following statement
\begin{equation}\label{eqn:variance-ca-1}
    \frac{1}{G_T} \sum_{1 \leq g \leq G}  \bar Y_g(1,\pi_2) N_g \left(\tilde Y_g(1,\pi_2) + \hat Y_g^{1}(\pi_2)\right)  I\{ H_g=\pi_2\} \xrightarrow{P} 2 E[\tilde Y_g(1,\pi_2) \bar Y_g(1,\pi_2) N_g] ~,
\end{equation}
and similar arguments would prove the following statement
\begin{equation*}
    \frac{1}{G_T} \sum_{1 \leq g \leq G}  N_g \left(\tilde Y_g(1,\pi_2) + \hat Y_g^{1}(\pi_2)\right)  I\{ H_g=\pi_2\} \xrightarrow{P} 2 E[N_g \tilde Y_g(1,\pi_2)] ~.
\end{equation*}
Note that
\begin{align*}
    &\frac{1}{G_T} \sum_{1 \leq g \leq G}  \bar Y_g(1,\pi_2) N_g \left(\tilde Y_g(1,\pi_2) + \hat Y_g^{1}(\pi_2)\right)  I\{ H_g=\pi_2\}\\
    &= \frac{1}{G_T} \sum_{1 \leq g \leq G}2  \bar Y_g(1,\pi_2) N_g \tilde Y_g(1,\pi_2)   I\{ H_g=\pi_2\} + \frac{1}{G_T} \sum_{1 \leq g \leq G}  \bar Y_g(1,\pi_2) N_g \left(\hat Y_g^1(\pi_2) -\tilde Y_g(1,\pi_2)\right)  I\{ H_g=\pi_2\}~.
\end{align*}
By weak law of large number, Slutsky's theorem and arguments in the proof of Theorem \ref{thm:clt-large-strata}, we have
\begin{align*}
     \frac{1}{ \frac{1}{G}\sum_{1\leq g \leq G} N_g } &\xrightarrow{P} \frac{1}{E[N_g]} \\
     \frac{\frac{1}{G}\sum_{1\leq g \leq G} \bar Y_g(1,\pi_2) I\{H_g=\pi_2\} N_g}{\left(\frac{1}{G} \sum_{1\leq g \leq G} N_g \right)^2} &\xrightarrow{P} \frac{E[\Bar{Y}_g(1,\pi_2) N_g]}{E[N_g]^2}~.
\end{align*}
Then, by Slutsky's theorem, Lemma \ref{lem:E_bounded} and Lemma B.3 of \cite{Bugni2018},
\begin{align*}
    &\frac{1}{G_T} \sum_{1 \leq g \leq G}  \bar Y_g(1,\pi_2) N_g \left(\hat Y_g^1(\pi_2) -\tilde Y_g(1,\pi_2)\right)  I\{ H_g=\pi_2\}\\
    &= \left(\frac{1}{E[N_g]} - \frac{1}{ \frac{1}{G}\sum_{1\leq g \leq G} N_g } \right) \frac{1}{G_T} \sum_{1 \leq g \leq G}  \bar Y_g(1,\pi_2)^2 N_g^2 I\{ H_g=\pi_2\} \xrightarrow{P} 0~.
\end{align*}
Again, by Lemma B.3 of \cite{Bugni2018}, and
\begin{equation*}
    E[\bar Y_g(1,\pi_2) N_g \tilde Y_g(1,\pi_2)] = \frac{E[\bar Y_g(1,\pi_2)^2 N_g^2]}{E[N_g]} - \frac{E[\bar Y_g(1,\pi_2) N_g] E[\bar Y_g(1,\pi_2) N_g^2]}{E[N_g]^2} < \infty,
\end{equation*}
 We conclude that (\ref{eqn:variance-ca-1}) holds, and then
\begin{equation*}
    \frac{1}{G_T} \sum_{1 \leq g \leq G}  \left(\frac{1}{E[N_g]} - \frac{1}{ \frac{1}{G}\sum_{1\leq g \leq G} N_g } \right) \bar Y_g(1,\pi_2) N_g \left(\tilde Y_g(1,\pi_2) + \hat Y_g^{1}(\pi_2)\right)  I\{ H_g=\pi_2\} \xrightarrow{P} 0~.
\end{equation*}
Therefore, (\ref{eqn:variance-ca-diff}) holds.
\qed

\subsection{Proof of Theorem \ref{thm:matched-group}}
To preserve space, I only present the proof for primary effect  as the proof for spillover effect follows the same argument. First, I analyze the equally-weighted estimator. Note that
\begin{equation*}
    \sqrt{G} (\hat{\theta}_1^P - \theta_1^P) =  (\mathbb L_G^{\rm Y1}, \mathbb L_G^{\rm Y0}) D_h  ~,
\end{equation*}
where $D_h=\left(\frac{1}{\sqrt{\pi_1}}, - \frac{1}{\sqrt{1-\pi_1}}\right)^\prime$ and $\mathbb L_G^{\rm Y1}, \mathbb L_G^{\rm Y0}$ are defined in Lemma \ref{lemma:asymptotics-mp}. Thus, by Lemma \ref{lemma:asymptotics-mp}, 
\begin{equation*}
    \sqrt{G} (\hat{\theta}_1^P - \theta_1^P) \xrightarrow{d} \mathcal{N}(0, D_h^\prime \mathbf{V}^{\rm e} D_h),
\end{equation*}
where
\begin{equation*}
   \mathbf{V}^{\rm e} = \begin{pmatrix}
 E[\var[\bar{Y}_g(1,\pi_2) | S_g]] & 0 \\
0 &  E[\var[\bar{Y}_g(0,0) | S_g]]
\end{pmatrix} + \var\left[\begin{pmatrix}
\sqrt{\pi_1} E[ \bar{Y}_g(1,\pi_2)  | S_g] \\
\sqrt{1-\pi_1} E[\bar{Y}_g(0,0) | S_g]\\
\end{pmatrix}\right]~,
\end{equation*}
By simple calculation,  we conclude that $D_h^\prime \mathbf{V}^{\rm e} D_h=V_3(1)$. In order to calculate the variance of size-weighted estimator, I follow the same argument in the end of Section \ref{proof:equally-weighted-ca}. Note that
\begin{equation*}
    \sqrt{G}(\hat{\beta} - \beta) = \sqrt{G}\left(\frac{\mathbb L_G^{\rm YN1}}{\sqrt{G}_1}, \frac{\mathbb L_G^{\rm N1}}{\sqrt{G}_1},  \frac{\mathbb L_G^{\rm YN0}}{\sqrt{G}_0}, \frac{\mathbb L_G^{\rm N0}}{\sqrt{G}_0}\right) = \left(\frac{1}{\sqrt{\pi_1}}, \frac{1}{\sqrt{\pi_1}}, \frac{1}{\sqrt{1-\pi_1}}, \frac{1}{\sqrt{1-\pi_1}}\right) \begin{pmatrix}
\mathbb L_{G}^{\rm YN1} \\
\mathbb L_{G}^{\rm N1} \\
\mathbb L_{G}^{\rm YN0} \\
\mathbb L_{G}^{\rm N0}
\end{pmatrix}~.
\end{equation*}
By a similar calculation and argument in Section \ref{proof:size-weighted-ca} and Lemma \ref{lemma:asymptotics-mp}, the final results is obtained. \qed

\begin{remark}[Details for Remark \ref{remark:how-second-stage-affects-variance}]\label{remark:how-second-stage-affects-variance-app}
    Apparently, the second term in $V_1(z)$ does not depend on the second-stage design since $Z_{i,g}(h)$ does not enter the second term. At first glance, $m_{z,h}(S_g)$ might seem to depend on $Z_{i,g}(h)$, i.e., the second-stage design. However, consider the following derivation:
    \begin{align*}
        m_{1,\pi_2}\left(S_{g}\right) &= E[\bar{Y}_g(1,\pi_2)\mid S_g] - E[\bar{Y}_g(1,\pi_2)] \\
        &= E\left[ \frac{1}{M_g^{1}} \sum_{i \in \mathcal{M}_g} Y_{i,g}(1,\pi_2)Z_{i,g}(\pi_2) \mid S_g\right] - E\left[\frac{1}{N_g}\sum_{1 \leq i \leq N_g} Y_{i,g}(1, \pi_2) \right] \\
        &= E\left[ \frac{1}{M_g^{1}} \sum_{i \in \mathcal{M}_g} Y_{i,g}(1,\pi_2) \mid S_g\right] E[Z_{i,g}(\pi_2)] - E\left[\frac{1}{N_g}\sum_{1 \leq i \leq N_g} Y_{i,g}(1, \pi_2) \right] \\
        &= E\left[ \frac{1}{M_g} \sum_{i \in \mathcal{M}_g} Y_{i,g}(1,\pi_2) \mid S_g\right]  - E\left[\frac{1}{N_g}\sum_{1 \leq i \leq N_g} Y_{i,g}(1, \pi_2) \right]~,
    \end{align*}
    which does not depend on $Z_{i,g}(\pi_2)$. The second equality is confirmed in Section \ref{proof:equally-weighted-ca}. The third equality is justified by Assumption \ref{as:assignment2}(b), which states that $(Z_{i,g}(\pi_2): 1\leq i \leq N_g) \independent ((Y_{i,g}(1,\pi_2): 1\leq i \leq N_g), \mathcal{M}_g, N_g, S_g)$. The last equality results from Assumption \ref{as:assignment2}(c)
\end{remark}

\subsection{Proof of Theorem \ref{thm:variance-estimator-mt}}\label{sec:proof-variance-mp}
First, note that we can write the variance expression as follows:
\begin{align*}
    V_1(z) &= \frac{1}{\pi_1} \operatorname{Var}\left[\bar{Y}_{g}(z,\pi_2)\right] + \frac{1}{1-\pi_1} \operatorname{Var}\left[\bar{Y}_{g}(0,0)\right] - \pi_1(1-\pi_1)  E\left[\left( \frac{1}{\pi_1}m_{z,\pi_2}\left(S_{g}\right) + \frac{1}{1-\pi_1}m_{0,0}\left(S_{g}\right) \right)^2\right] \\
    &= \frac{1}{\pi_1} E\left[ \operatorname{Var}\left[\bar{Y}_{g}(z,\pi_2)\mid S_g \right] \right] + \frac{1}{1-\pi_1} E \left[ \operatorname{Var}\left[\bar{Y}_{g}(0,0)\mid S_g  \right]\right] +   \operatorname{Var}\left[E \left[\bar{Y}_{g}(z,\pi_2) \mid S_g \right]  \right] \\
    &\quad + \operatorname{Var} \left[E \left[\bar{Y}_{g}(0,0)\mid S_g  \right]\right]  - 2\cdot \operatorname{Cov}\left[E \left[\bar{Y}_{g}(z,\pi_2)\mid S_g  \right] , E \left[\bar{Y}_{g}(0,0)\mid S_g  \right]\right] ~.
\end{align*}
By Slutsky's theorem and Lemma \ref{lemma:var-mp-1}-\ref{lemma:var-mp-3},  we conclude that $\hat V_1(z) \xrightarrow{P} V_1(z)$. Similarly, by Slutsky's theorem and Lemma \ref{lemma:var-mp-size-1}-\ref{lemma:var-mp-size-3},  we conclude that $\hat V_2(z) \xrightarrow{P} V_2(z)$.

\subsection{Proof of Theorem \ref{thm:optimal-first-stage}}
To begin with, observe that it is equivalent to show $g_z^{\rm e}(C_g, N_g)=E\left[  \frac{\bar{Y}_g(z,\pi_2)}{\pi_1} + \frac{\bar{Y}_g(0,0) }{1-\pi_1}\mid C_g, N_g \right]$ maximizes
\begin{align*}
    &E\left[\left(\frac{m_{z,\pi_2}\left(S_{g}\right)}{\pi_1} + \frac{m_{0,0}\left(S_{g}\right)}{1-\pi_1} \right)^2\right]\\
    &= E\left[ \left( \frac{E[\bar{Y}_g(z,\pi_2)\mid S_g] - E[\bar{Y}_g(z,h)]}{\pi_1} + \frac{E[\bar{Y}_g(0,0)\mid S_g] - E[\bar{Y}_g(z,h)]}{1-\pi_1}\right)^2\right]~,
\end{align*}
and $g_z^{\rm s}(C_g, N_g)=E\left[\frac{\tilde Y_g(z,\pi_2)}{\pi_1} + \frac{\tilde Y_g(0,0)}{1-\pi_1} \mid C_g,N_g\right]$ maximizes
\begin{align*}
     &E\left[\left(\frac{1}{\pi_1}E[\tilde Y_g(z,\pi_2)\mid S_g] + \frac{1}{1-\pi_1}E[\tilde Y_g(0,0)\mid S_g] \right)^2\right] ~.
\end{align*}
By Theorem C.2. of \cite{bai-inference}, the result for equally-weighted estimators follow directly. In terms of the size-weighted estimators, first observe that
\begin{align*}
    & E\left[\left(g_z^{\rm s}(C_g, N_g) -  E\left[ \frac{\tilde Y_g(z,\pi_2)}{\pi_1} + \frac{\tilde Y_g(0,0)}{1-\pi_1} \bigg | S_g \right] \right) E\left[ \frac{\tilde Y_g (z,\pi_2)}{\pi_1} + \frac{\tilde Y_g (0,0)}{1-\pi_1} \bigg | S_g \right] \right] \\
    &=E\left[ E\left[\left(g_z^{\rm s}(C_g, N_g) -  E\left[ \frac{\tilde Y_g (z,\pi_2)}{\pi_1} + \frac{\tilde Y_g (0,0)}{1-\pi_1} \bigg | S_g \right] \right)\bigg | S_g\right] E\left[ \frac{\tilde Y_g (z,\pi_2)}{\pi_1} + \frac{\tilde Y_g (0,0)}{1-\pi_1} \bigg | S_g \right] \right] \\
    &= 0~,
\end{align*}
by law of iterated expectation. Therefore,
\begin{align*}
    & E\left[g_z^{\rm s}(C_g, N_g)^2 \right] \\
    &=E\left[\left(g_z^{\rm s}(C_g, N_g) - E\left[ \frac{\tilde Y_g (z,\pi_2)}{\pi_1} + \frac{\tilde Y_g (0,0)}{1-\pi_1} \bigg | S_g \right] + E\left[ \frac{\tilde Y_g (z,\pi_2)}{\pi_1} + \frac{\tilde Y_g (0,0)}{1-\pi_1} \bigg | S_g \right] \right)^2 \right] \\
    &= E\left[\left(g_z^{\rm s}(C_g, N_g) - E\left[ \frac{\tilde Y_g (z,\pi_2)}{\pi_1} + \frac{\tilde Y_g (0,0)}{1-\pi_1} \bigg | S_g \right] \right)^2 \right] + E\left[ E\left[ \frac{\tilde Y_g (z,\pi_2)}{\pi_1} + \frac{\tilde Y_g (0,0)}{1-\pi_1} \bigg | S_g \right]^2 \right] \\
    &\geq  E\left[ E\left[ \frac{\tilde Y_g (z,\pi_2)}{\pi_1} + \frac{\tilde Y_g (0,0)}{1-\pi_1} \bigg | S_g \right]^2 \right]~.
\end{align*}
Thus, it is optimal to match on 
\begin{align*}
     g_z^{\rm s}(C_g, N_g) =& E\left[\frac{\tilde Y_g (z,\pi_2)}{\pi_1} + \frac{\tilde Y_g (0,0)}{1-\pi_1} \mid C_g,N_g\right]  ~.
\end{align*}
\qed

\subsection{Proof of Theorem \ref{thm:optimal-second-stage}}
In this section, I show the optimality result holds for $V_1(z)$ first. To begin with, observe that the second stage design enters the variance formula only through $Z_{i,g}$, or in other words $\bar{Y}_g(z,\pi_2)$. Moreover, the conditional expectations, $m_{1,\pi_2}(C_g), m_{0,\pi_2}(C_g), m_{1,\pi_2}(S_g), m_{0,\pi_2}(S_g)$, do not depend on the stratification strategy. Take $m_{1,\pi_2}(C_g)$ as an example:
\begin{align*}
    m_{1,\pi_2}(C_g) &= E\left[\frac{1}{M_g^{1}} \sum_{i \in \mathcal{M}_g} Y_{i,g}(1,\pi_2)Z_{i,g}(\pi_2) \mid C_g\right] - E\left[\frac{1}{M_g^{1}} \sum_{i \in \mathcal{M}_g} Y_{i,g}(1,\pi_2)Z_{i,g}(\pi_2) \right] \\
    &= E\left[ \frac{1}{M_g^{1}} \sum_{i \in \mathcal{M}_g} E\left[ Y_{i,g}(1,\pi_2)Z_{i,g}(\pi_2) \mid C_g, \mathcal{M}_g, B_g \right] \mid C_g\right] - E\left[\frac{1}{N_g}\sum_{1 \leq i \leq N_g} Y_{i,g}(1, \pi_2) \right] \\
    &=   E\left[ \frac{1}{M_g} \sum_{i \in \mathcal{M}_g} Y_{i,g}(1,\pi_2) \mid C_g \right]  - E\left[\frac{1}{N_g}\sum_{1 \leq i \leq N_g} Y_{i,g}(1, \pi_2) \right]~,
\end{align*}
where the last inequality holds by Assumption \ref{as:assignment2}. Therefore, only the first term is likely to depend on stratification strategy. In addition,
\begin{align*}
    \var\left[\bar{Y}_g(1,\pi_2) \right] = E\left[\bar{Y}_g(1,\pi_2)^2 \right] - E\left[\bar{Y}_g(1,\pi_2) \right]^2~,
\end{align*}
for which I only need to focus on the first term. Let $X_g=(X_{i,g}:1\leq i \leq N_g)$.
\begin{align*}
    E\left[\bar{Y}_g(1,\pi_2)^2 \right] &= E\left[ \left(\frac{1}{M_g^{1}} \sum_{i \in \mathcal{M}_g} Y_{i,g}(1,\pi_2)Z_{i,g}(\pi_2) \right)^2\right]\\
    &= E\left[ \frac{1}{\left(M_g^{1}\right)^2} E\left[ \left( \sum_{i \in \mathcal{M}_g} Y_{i,g}(1,\pi_2)Z_{i,g}(\pi_2) \right)^2\mid X_g, \mathcal{M}_g\right]  \right]~.
\end{align*}
% \begin{align*}
%     &E[Y^\prime Z Z^\prime Y] = E[Y(1)\mid X]^\prime E[Z Z^\prime \mid X ] E[Y(1)\mid X] \\
%     &E[(Y_1 + Y_2)^2] = E[Y_1^2] + 2 E[Y_1 Y_2] \\
%     &\sum_{j=1}^n \frac{1}{4} E[(Y_{\pi(2j-1),g}(1,\pi_2)- Y_{\pi(2j),g}(1,\pi_2))^2 \mid X_g ]  \\
%     &= \sum_{j=1}^n \frac{1}{4} E[Y_{\pi(2j-1),g}(1,\pi_2)^2 +  Y_{\pi(2j),g}(1,\pi_2)^2 - 2 Y_{\pi(2j-1),g}(1,\pi_2) Y_{\pi(2j),g}(1,\pi_2) \mid X_g ] \\
%     & \dots E[Y_{\pi(2j-1),g}(1,\pi_2)  \mid X_g ]   E[Y_{\pi(2j),g}(1,\pi_2)  \mid X_g ]  \\
%     & A_1, \dots, A_{2n}: \max \sum_{j=1}^n E[  A_{\pi(2j-1)} A_{\pi(2j)}] \\
%     & \text{assume we know } E[A_i A_k] \forall i,k
% \end{align*}
In fact, it is equivalent to consider
\begin{align*}
    &E\left[ \frac{1}{\left(M_g^{1}\right)^2} E\left[ \left( \sum_{i \in \mathcal{M}_g} Y_{i,g}(1,\pi_2)Z_{i,g}(\pi_2) \right)^2-  \left( \sum_{i \in \mathcal{M}_g} Y_{i,g}(1,\pi_2)\pi_2 \right)^2\mid X_g, \mathcal{M}_g\right]  \right] \\
    &= E\left[ \frac{1}{\left(M_g^{1}\right)^2} E\left[\sum_{i,j\in\mathcal{M}_g:B_{i,g}=B_{j,g}} Y_{i,g}(1,\pi_2) Y_{j,g}(1,\pi_2) \left(Z_{i,g}(\pi_2) Z_{j,g}(\pi_2)-\pi_2^2\right)  \mid X_g, \mathcal{M}_g\right]  \right] \\
    &= E\left[ \frac{1}{\left(M_g^{1}\right)^2} \sum_{b\in\mathcal{B}} \sum_{B_{i,g}=B_{j,g}=b} E\left[ Y_{i,g}(1,\pi_2) Y_{j,g}(1,\pi_2)  \mid X_g\right] E\left[ Z_{i,g}(\pi_2) Z_{j,g}(\pi_2)-\pi_2^2  \mid X_g\right]  \right] \\
    &= E\left[ \frac{1}{\left(M_g^{1}\right)^2} \sum_{b\in\mathcal{B}} \sum_{i:B_{i,g}=b} E\left[ Y_{i,g}^2(1,\pi_2)  \mid X_g\right] (\pi_2-\pi_2^2)  \right]\\
    &\quad\quad\quad+ E\Bigg[ \frac{1}{\left(M_g^{1}\right)^2} \sum_{b\in\mathcal{B}} \sum_{i\neq j:B_{i,g}=B_{j,g}=b} \bigg( E\left[ Y_{i,g}(1,\pi_2) \mid X_{g}\right] E\left[Y_{j,g}(1,\pi_2)\mid X_{g}\right] \\
    &\quad\quad\quad\quad\quad\quad  + \operatorname{Cov}(Y_{i,g}(1,\pi_2), Y_{j,g}(1,\pi_2)) \bigg) \times E\left[ Z_{i,g}(\pi_2) Z_{j,g}(\pi_2)-\pi_2^2  \mid X_g\right]  \Bigg] ~,
\end{align*}
where the last inequality holds by Assumption \ref{ass:cond-indep}. Note the last term with $\operatorname{Cov}(Y_{i,g}(1,\pi_2), Y_{j,g}(1,\pi_2))$ does not affect the optimization problem and can be dropped since it is invariant across units. By Lemma II.2 of \cite{Bai-optimal}, we only need to consider matched-group design with group size $k$ when $\pi_2 = l/k$ with $l<k$ being positive integers.\footnote{Without loss of generality, I implicitly assume that $N_g/k$ is an integer.} Note that the first term does not depend on stratification, for which we can replace $E\left[ Y_{i,g}^2(1,\pi_2)  \mid X_g\right]$ with $E\left[ Y_{i,g}(1,\pi_2)  \mid X_{g}\right]^2$ without affecting the optimzation problem. Then, by Lemma \ref{lemma:crd-cov} and Assumption \ref{ass:assignment-crd}, we can write the objective above as
\begin{align*}
    &E\left[ \frac{1}{\left(M_g^{1}\right)^2} \sum_{b\in\mathcal{B}} \sum_{i:B_{i,g}=b} E\left[ Y_{i,g}(1,\pi_2)  \mid X_{g}\right]^2  \left(\pi_2-\pi_2^2\right)   \right] \\
    &\quad + E\left[ \frac{1}{\left(M_g^{1}\right)^2} \sum_{b\in\mathcal{B}} \sum_{i\neq j: B_{i,g}=B_{j,g}=b} E\left[ Y_{i,g}(1,\pi_2) \mid X_{g}\right] E\left[Y_{j,g}(1,\pi_2)\mid X_{g}\right] \left( \frac{\pi_2^2-\pi_2}{k-1} \right)  \right] \\
    &= E\left[ \frac{1}{\left(M_g^{1}\right)^2} \sum_{b\in\mathcal{B}} \sum_{i:B_{i,g}=b}\left( E\left[ Y_{i,g}(1,\pi_2)  \mid X_{g}\right] - \bar{\mu}^b(X_{g}) \right)^2  \frac{k \left(\pi_2-\pi_2^2\right)}{k-1}   \right]~,
\end{align*}
where
\begin{align*}
    \bar{\mu}^b(X_{g}) &= \frac{1}{k} \sum_{i:B_{i,g}=b} E\left[ Y_{i,g}(1,\pi_2)  \mid X_{g}\right] ~.
\end{align*}
Therefore, the optimal matching strategy matches on $E\left[ Y_{i,g}(1,\pi_2)  \mid X_{g}\right]$.

Now, let's turn to $V_2(z)$ for $z \in \{0,1\}$. Follow the same argument to conclude that $E[\tilde Y_g(z,\pi_2)\mid S_g]$ is invariant to stratification strategy. Then, only the first term is likely to be affected by stratification.
\begin{align*}
    \var[\tilde Y_g(z,\pi_2) ] = E\left[ \frac{N_g^2}{E[N_g]^2}\left(\bar Y_g(z,\pi_2)^2 - 2 \bar Y_g(z,\pi_2) \frac{E[\bar Y_g(z,\pi_2)N_g]}{E[N_g]} + \frac{E[\bar Y_g(z,\pi_2)N_g]^2}{E[N_g]^2} \right) \right],
\end{align*}
for which we only need to focus on 
\begin{equation*}
    E\left[ N_g^2 \bar Y_g(z,\pi_2)^2 \right] = E\left[ N_g^2 E\left[ \bar Y_g(z,\pi_2)^2 \mid N_g\right] \right],
\end{equation*}
which is also minimized by a matched-group design that matches on $E\left[ Y_{i,g}(z,\pi_2)  \mid X_{g}\right]$.
\qed

\subsubsection{Efficiency Improvement in a Matched-Pair Example}\label{subsubsec:variance-improvement-bound}
Consider a matched-pair design with \( k=2 \) and \( \pi_2=1/2 \). The relevant term for variance improvement is given as:
\begin{align*}
    &E\left[ \frac{1}{M_g^2} \sum_{b\in\mathcal{B}} \sum_{B_{i,g}=B_{j,g}=b}\left( E\left[ Y_{i,g}(1,\pi_2)  \mid X_{g}\right] - E\left[ Y_{j,g}(1,\pi_2)  \mid X_{g}\right] \right)^2 \right] \\
    &= E\left[ \frac{1}{M_g^2} \sum_{b\in\mathcal{B}} \sum_{B_{i,g}=B_{j,g}=b} E\left[\left( E\left[ Y_{i,g}(1,\pi_2)  \mid X_{g}\right] - E\left[ Y_{j,g}(1,\pi_2)  \mid X_{g}\right] \right)^2 \mid \mathcal{M}_g\right] \right]  \\
    &\leq E\left[ \frac{1}{M_g^2} \sum_{b\in\mathcal{B}} \sum_{B_{i,g}=B_{j,g}=b} 2 E\left[E\left[ Y_{i,g}(1,\pi_2)  \mid X_{g}\right]^2 + E\left[ Y_{j,g}(1,\pi_2)  \mid X_{g}\right]^2  \mid \mathcal{M}_g\right] \right] \\
    &= E\left[ \frac{2}{M_g^2} \sum_{i \in \mathcal{M}_g}  E\left[E\left[ Y_{i,g}(1,\pi_2)  \mid X_{g}\right]^2 \mid \mathcal{M}_g\right] \right] ~.
\end{align*}
If \( E\left[E\left[ Y_{i,g}(1,\pi_2)  \mid X_{g}\right]^2 \mid \mathcal{M}_g\right] \leq C \) and \( M_g \geq M \) for some constants \( C \) and \( M \geq 0 \), then the relevant term for variance improvement becomes \( O(1/M) \).

\section{Auxiliary Lemmas}
\begin{lemma}
If cluster size is fixed for all $1\leq g \leq G$, i.e. $N_g = N$, then, $V_1(z) =V_2(z)$ for $z\in\{0,1\}$.
\begin{proof}
Note that when $N_g = N$,
\begin{equation*}
    \tilde Y_g (z,h) = \bar Y_g(z,h) - E[\bar Y_g(z,h)]~.
\end{equation*}
Then,
\begin{align*}
    V_2(z) &= \frac{1}{\pi_1} \var[Y_g(z,\pi_2)] + \frac{1}{1-\pi_1} \var[Y_g(0,0)] \\
    &\quad - E\left[\left(\sqrt{\frac{1-\pi_1}{\pi_1}} m_{z,\pi_2}(S_g) + \sqrt{\frac{ \pi_1}{1-\pi_1}} m_{0,0}(S_g)  \right)^2 \right] \\
    &\quad + E\left[\tau(S_g) \left(\frac{1}{\pi_1}m_{z,\pi_2}(S_g) + \frac{1}{1-\pi_1}m_{0,0}(S_g) \right)^2 \right]~.
\end{align*}
By law of iterated expectation, we have $E\left[m_{z,h}\left(C_{g}\right) \mid S_{g}\right] = m_{z,h}\left(S_{g}\right)$. Thus,
\begin{align*}
    V_1(z) &= \frac{1}{\pi_1} \operatorname{Var}\left[\tilde{Y}_{g}(z,\pi_2)\right]+\frac{1}{1-\pi_1} \operatorname{Var}\left[\tilde{Y}_{g}(0,0)\right] + E\left[\left(m_{z,\pi_2}\left(S_{g}\right)-m_{0,0}\left(S_{g}\right)\right)^{2}\right]\\
    &\quad\quad + E\left[\tau\left(S_{g}\right)\left(\frac{1}{\pi_1} m_{z,\pi_2}\left(S_{g}\right)+\frac{1}{1-\pi_1} m_{0,0}\left(S_{g}\right)\right)^{2}\right] \\
    &= \frac{ E\left[ \bar Y_g^2(z,\pi_2) \right] -   E[\bar Y_g(z,\pi_2)]^2 }{\pi_1} + \frac{ E\left[ \bar Y_g^2(0,0)  \right] -  E[\bar Y_g(0,0)]^2}{1-\pi_1} - 2 E\left[m_{z,\pi_2}\left(S_{g}\right)m_{0,0}\left(S_{g}\right)\right] \\
    &\quad\quad - \frac{1-\pi_1}{\pi_1}\left( E\left[ E[Y_g(z,\pi_2)\mid S_g]^2 \right]  -  E[Y_g(z,\pi_2)]^2\right)  - \frac{\pi_1}{1-\pi_1} \left( E\left[ E[Y_g(0,0)\mid S_g]^2\right]-  E[Y_g(0,0)]^2\right) \\
    &= V_2(z)~.
\end{align*}
\end{proof}
\end{lemma}

\begin{lemma}\label{lemma:crd-cov}
Given a sequence of binary random variables $A^{(n)}=(A_{i}: 1\leq i \leq n)$ with the joint distribution
\begin{equation*}
    P\left(A^{(n)}=a^{(n)}\right) = \frac{1}{\left(\begin{array}{l}
n \\
n \pi
\end{array}\right)} \text{ for all } a^{(n)}= (a_{i}: 1\leq i \leq n) \text{ such that } \sum_{1\leq i \leq n} a_{i} = n \pi~,
\end{equation*}
where $n \pi \in \mathbb{N}$ is an integer, otherwise $P\left(A^{(n)}=a^{(n)}\right)=0$. We have $E[A_i A_j] = \pi^2 -\frac{\pi(1-\pi)}{n-1}$ for all $i\neq j\in [1, n]$.
\end{lemma}
\begin{proof}
Note that
\begin{align*}
    \text{Var}\left[\sum_{1\leq i \leq n} A_{i} \right] &= 0 = \sum_{1\leq i \leq n}  \text{Var}\left[ A_{i} \right] + \sum_{i\neq j} \text{Cov}(A_i, A_j) = n \pi (1-\pi) + n(n-1) \text{Cov}(A_i, A_j)~,
\end{align*}
for any $i\neq j\in[1,n]$, which implies 
\begin{equation*}
    E[A_i A_j] = \text{Cov}(A_i, A_j) + E[A_i] E[A_j] = \pi^2 -\frac{\pi(1-\pi)}{n-1}~.
\end{equation*}
\end{proof}

\begin{lemma}\label{lem:E_bounded}
Suppose Assumption \ref{ass:Q_G} holds, then
\[E[\bar{Y}^r_g(z, \pi_2)|C_g, N_g] \le C \hspace{3mm} a.s.~,\]
for $r \in \{1, 2\}, z\in\{0,1\}$ for some constant $C > 0$,  
\[E\left[\bar{Y}_g^r(z,\pi_2)N_g^\ell\right] < \infty~,\]
for $r \in \{1, 2\}, \ell \in \{0, 1, 2\}, z\in\{0,1\}$, and
\[E\left[E[\bar{Y}_g(z,\pi_2)N_g|S_g]^2\right] < \infty~.,\]
for $z\in\{0,1\}$. In addition, suppose Assumption \ref{as:Q_G-lip} (b) holds, then
\begin{equation*}
    E[\bar{Y}_g(z,\pi_2)^r N_g^\ell \mid S_g] \leq C \hspace{3mm} a.s.~,
\end{equation*}
for $z\in\{0,1\}$.
\end{lemma}
\begin{proof}
We show the first statement for $r = 2$ and $z=1$, since the case $r = 1$ follows similarly. By the Cauchy-Schwarz inequality,
\[\bar{Y}_g(1,\pi_2)^2 = \left(\frac{1}{M_g}\sum_{i\in \mathcal{M}_g}Y_{i,g}(1,\pi_2) Z_{i,g}(\pi_2) \right)^2 \le \frac{1}{M_g}\sum_{i\in \mathcal{M}_g}Y_{i,g}(1,\pi_2)^2~,\]
and hence
\[E[\bar{Y}_g(1,\pi_2)^2|C_g, N_g, X_g] \le   E\left[ \frac{1}{M_g}\sum_{i \in \mathcal{M}_g} Y_{i,g}(1,\pi_2)^2 \mid C_g,N_g,X_g\right] \leq \sum_{1\leq i\leq N_g} E\left[\frac{1\{i\in \mathcal{M}_g\}}{M_g}\mid C_g,N_g,X_g\right] C \le C~,\]
where the second inequality follows from the above derivation, Assumption \ref{ass:Q_G}(e) and the law of iterated expectations, and final inequality follows from Assumption \ref{ass:Q_G}(d).
I show the next statement for $r = \ell = 2$, since the other cases follow similarly. By the law of iterated expectations,
\begin{align*}
E\left[\bar{Y}^2_g(1,\pi_2)N_g^2\right] &= E\left[N_g^2E[\bar{Y}^2_g(1,\pi_2)|C_g,N_g]\right] \\
&\lesssim E\left[N_g^2\right] < \infty~,
\end{align*}
where the final line follows by Assumption \ref{ass:Q_G} (c). Next,
\begin{align*}
E\left[E[\bar{Y}_g(1,\pi_2)N_g|S_g]^2\right] &= E\left[E[N_gE[\bar{Y}_g(1,\pi_2)|C_g, N_g]|S_g]^2\right] \\
&\lesssim E\left[E[N_g|C_g]^2\right] < \infty~,
\end{align*}
where the final line follows from Jensen's inequality and Assumption \ref{ass:Q_G}(c). Finally,
\begin{align*}
    E[\bar{Y}_g(z,\pi_2)^r N_g^\ell \mid S_g] &= E[ N_g^\ell E[\bar{Y}_g(z,\pi_2)^r  \mid C_g. N_g] \mid S_g] \lesssim E[N_g^\ell  \mid S_g] \leq C~,
\end{align*}
where the last inequality follows by Assumption \ref{as:Q_G-lip} (b).
\end{proof}

\begin{lemma}\label{lemma:asymptotics-mp}
Suppose $Q_G$ satisfies Assumptions \ref{ass:Q_G} and \ref{as:Q_G-lip} and the treatment assignment mechanism satisfies Assumptions \ref{ass:assignment-match}-\ref{as:close} and \ref{as:assignment2}. Define
\begin{align*}
\mathbb L_G^{\rm Y1} & = \frac{1}{\sqrt{n l}} \sum_{1 \leq g \leq n k} (\bar{Y}_g(1,\pi_2)   - E[\bar{Y}_g(1,\pi_2) ] )I\{H_g=\pi_2\} \\
\mathbb L_G^{\rm YN1} & = \frac{1}{\sqrt{n l}} \sum_{1 \leq g \leq n k} (\bar{Y}_g(1,\pi_2) N_g  - E[\bar{Y}_g(1,\pi_2) N_g] )I\{H_g=\pi_2\} \\
\mathbb L_G^{\rm N1} & = \frac{1}{\sqrt{n l}} \sum_{1 \leq g \leq 2G} (N_g  - E[N_g] ) I\{H_g=\pi_2\} \\
\mathbb L_G^{\rm Y0} & = \frac{1}{\sqrt{n (k-l)}} \sum_{1 \leq g \leq n k} (\bar{Y}_g(0,0)   - E[\bar{Y}_g(0,0) N_g] )I\{H_g=0\} \\
\mathbb L_G^{\rm YN0} & = \frac{1}{\sqrt{n (k-l)}} \sum_{1 \leq g \leq 2G} (\bar{Y}_g(0,0) N_g  - E[\bar{Y}_g(0,0) N_g] )I\{H_g=0\} \\
\mathbb L_G^{\rm N0} & = \frac{1}{\sqrt{n (k-l)} } \sum_{1 \leq g \leq 2G} (N_g  - E[N_g] )I\{H_g=0\}~.
\end{align*}
Then, as $n\rightarrow \infty$, 
\begin{equation*}
    \left(\mathbb L_G^{\rm Y1},  \mathbb L_G^{\rm YN1}, \mathbb L_G^{\rm N1}, \mathbb L_G^{\rm Y0}, \mathbb L_G^{\rm YN0}, \mathbb L_G^{\rm N0}\right) \xrightarrow{d} \mathcal{N}(0, \mathbf{V})~,
\end{equation*}
where
\begin{equation*}
    \mathbf{V} = \mathbf{V}_1 + \mathbf{V}_2
\end{equation*}
for
\[ \mathbf V_1 = \begin{pmatrix}
\mathbf V_1^1 & 0 \\
0 & \mathbf V_1^0
\end{pmatrix} \]
{\footnotesize
\begin{align*}
\mathbf V_1^1 & = \begin{pmatrix}
E[\var[\bar{Y}_g(1,\pi_2) | S_g]] & E[\cov[\bar{Y}_g(1,\pi_2) , \bar{Y}_g(1,\pi_2) N_g | S_g]] & E[\cov[\bar{Y}_g(1,\pi_2) , N_g | S_g]] \\
E[\cov[\bar{Y}_g(1,\pi_2) , \bar{Y}_g(1,\pi_2) N_g | S_g]] & E[\var[\bar{Y}_g(1,\pi_2) N_g | S_g]]  & E[\cov[\bar{Y}_g(1,\pi_2) N_g, N_g | S_g]] \\
E[\cov[\bar{Y}_g(1,\pi_2) , N_g | S_g]] & E[\cov[\bar{Y}_g(1,\pi_2) N_g, N_g | S_g]] & E[\var[N_g | S_g]]
\end{pmatrix} \\
\mathbf V_1^0 & = \begin{pmatrix}
E[\var[\bar{Y}_g(0,0) | S_g]] & E[\cov[\bar{Y}_g(0,0) , \bar{Y}_g(0,0) N_g | S_g]] & E[\cov[\bar{Y}_g(0,0) , N_g | S_g]] \\
E[\cov[\bar{Y}_g(0,0) , \bar{Y}_g(0,0) N_g | S_g]] & E[\var[\bar{Y}_g(0,0) N_g | S_g]]  & E[\cov[\bar{Y}_g(0,0) N_g, N_g | S_g]] \\
E[\cov[\bar{Y}_g(0,0) , N_g | S_g]] & E[\cov[\bar{Y}_g(0,0) N_g, N_g | S_g]] & E[\var[N_g | S_g]]
\end{pmatrix}
\end{align*}}
{\footnotesize
\[ \mathbf V_2 =  \var\left[\begin{pmatrix}
\sqrt{\pi_1} E[ \bar{Y}_g(1,\pi_2)  | S_g] \\
\sqrt{\pi_1} E[ \bar{Y}_g(1,\pi_2) N_g | S_g]\\
\sqrt{\pi_1} E[  N_g | S_g]\\
\sqrt{1-\pi_1} E[\bar{Y}_g(0,0) | S_g]\\
\sqrt{1-\pi_1} E[\bar{Y}_g(0,0) N_g | S_g]\\
\sqrt{1-\pi_1} E[ N_g | S_g]
\end{pmatrix}\right]~. \]}
\end{lemma}
\begin{proof}
Note
\begin{align*}
    (\mathbb L_G^{\rm Y1}, \mathbb L_G^{\rm YN1}, \mathbb L_G^{\rm N1}, \mathbb L_G^{\rm Y0}, \mathbb L_G^{\rm YN0}, \mathbb L_G^{\rm N0}) &= (\mathbb L_{1, G}^{\rm Y1}, \mathbb L_{1, G}^{\rm YN1}, \mathbb L_{1, G}^{\rm N1}, \mathbb L_{1, G}^{\rm Y0}, \mathbb L_{1, G}^{\rm YN0}, \mathbb L_{1, G}^{\rm N0})\\
    &\quad\quad\quad + (\mathbb L_{2, G}^{\rm Y1}, \mathbb L_{2, G}^{\rm YN1}, \mathbb L_{2, G}^{\rm N1}, \mathbb L_{2, G}^{\rm Y0}, \mathbb L_{2, G}^{\rm YN0}, \mathbb L_{2, G}^{\rm N0})~,
\end{align*}
where
\begin{align*}
\mathbb L_{1, G}^{\rm YN1} & = \frac{1}{\sqrt{nl}} \sum_{1 \leq g \leq 2G} (\bar{Y}_g(1,\pi_2) N_g I\{H_g=\pi_2\} - E[\bar{Y}_g(1,\pi_2) N_g I\{H_g=\pi_2\} | S^{(G)}, H^{(G)}]) \\
\mathbb L_{2, G}^{\rm YN1} & = \frac{1}{\sqrt{nl}} \sum_{1 \leq g \leq 2G} (E[\bar{Y}_g(1,\pi_2) N_g I\{H_g=\pi_2\} | S^{(G)}, H^{(G)}] - E[\bar{Y}_g(1,\pi_2) N_g] I\{H_g=\pi_2\})
\end{align*}
and similarly for the rest. Next, note $(\mathbb L_{1,G}^{\rm Y1}, \mathbb L_{1,G}^{\rm YN1}, \mathbb L_{1,G}^{\rm N1}, \mathbb L_{1,G}^{\rm Y0}, \mathbb L_{1,G}^{\rm YN0}, \mathbb L_{1,G}^{\rm N0}), n \geq 1$ is a triangular array of normalized sums of random vectors. We will apply the Lindeberg central limit theorem for random vectors, i.e., Proposition 2.27 of \cite{van_der_vaart1998asymptotic}, to this triangular array. Conditional on $S^{(G)}, H^{(G)}$, $(\mathbb L_{1,G}^{\rm Y1}, \mathbb L_{1, G}^{\rm YN1}, \mathbb L_{1, G}^{\rm N1}) \perp (\mathbb L_{1,G}^{\rm Y0}, \mathbb L_{1, G}^{\rm YN0}, \mathbb L_{1, G}^{\rm N0})$. Moreover, it follows from $Q_G = Q^{G}$ (by Lemma 5.1 of \cite{Bugni2022} and Assumption \ref{ass:Q_G} (a)-(b)) and Assumption \ref{ass:assignment-match}, \ref{as:assignment2} that
{\tiny
\begin{align*}
    &\var \left [ \left(\mathbb L_{1,G}^{\rm Y1}, \mathbb L_{1, G}^{\rm YN1}, \mathbb L_{1, G}^{\rm N1}\right)^\prime |  S^{(G)}, H^{(G)} \right ] \\
&= \begin{pmatrix}
\frac{1}{nl} \sum_{g =1}^{G} \var[\bar{Y}_g(1,\pi_2)  | S_g] \tilde{H}_g &  \frac{1}{nl} \sum_{g =1}^{G} \cov[\bar{Y}_g(1,\pi_2), \bar{Y}_g(1,\pi_2) N_g| S_g] \tilde{H}_g & \frac{1}{nl} \sum_{g =1}^{G} \cov[\bar{Y}_g(1,\pi_2) , N_g | S_g] \tilde{H}_g \\
\frac{1}{nl} \sum_{g =1}^{G} \cov[\bar{Y}_g(1,\pi_2), \bar{Y}_g(1,\pi_2) N_g| S_g] \tilde{H}_g & \frac{1}{nl} \sum_{g =1}^{G} \var[\bar{Y}_g(1,\pi_2) N_g | S_g] \tilde{H}_g & \frac{1}{nl} \sum_{g =1}^{G} \cov[\bar{Y}_g(1,\pi_2) N_g , N_g | S_g] \tilde{H}_g \\
\frac{1}{nl} \sum_{g =1}^{G} \cov[\bar{Y}_g(1,\pi_2) , N_g | S_g] \tilde{H}_g  & \frac{1}{nl} \sum_{g =1}^{G} \cov[\bar{Y}_g(1,\pi_2) N_g , N_g | S_g] \tilde{H}_g  & \frac{1}{nl} \sum_{g =1}^{G} \var[N_g | S_g] \tilde{H}_g
\end{pmatrix}~,
\end{align*}}
where $\tilde{H}_g = I\{H_g=\pi_2\}$. For the upper left component, we have
\begin{equation} \label{eq:condlip}
\frac{1}{G_T} \sum_{1 \leq g \leq G} \var[\bar{Y}_g(1,\pi_2) | S_g] \tilde{H}_g = \frac{1}{G_T} \sum_{1 \leq g \leq G} E[\bar{Y}^2_g(1,\pi_2)   | S_g] \tilde{H}_g - \frac{1}{G_T} \sum_{1 \leq g \leq G} E[\bar{Y}_g(1,\pi_2)   | S_g]^2 \tilde{H}_g~,
\end{equation}
where $G_T = nl$. Note
\begin{align*}
& \frac{1}{G_T} \sum_{1 \leq g \leq G} E[\bar{Y}^2_g(1,\pi_2)   | S_g] \tilde{H}_g  = \frac{1}{G} \sum_{1 \leq g \leq G} E[\bar{Y}^2_g(1,\pi_2)   | S_g]\\
&\quad\quad\quad+  (1-\pi_2)\left ( \frac{1}{G_T} \sum_{1 \leq g \leq G: \tilde{H}_g  = 1} E[\bar{Y}^2_g(1,\pi_2)   | S_g] - \frac{1}{G_C} \sum_{1 \leq g \leq G: \tilde{H}_g  = 0} E[\bar{Y}^2_g(1,\pi_2)   | S_g] \right )~.
\end{align*}
It follows from the weak law of large numbers, and Lemma \ref{lem:E_bounded}, that
\[ \frac{1}{G} \sum_{1 \leq g \leq G} E[\bar{Y}^2_g(1,\pi_2)   | S_g] \stackrel{P}{\to} E[\bar{Y}^2_g(1,\pi_2)  ]~. \]
On the other hand, it follows from Assumption  \ref{as:close} and \ref{as:Q_G-lip}(b) that
\begin{align*}
& \left | \frac{1}{G_T} \sum_{1 \leq g \leq G: \tilde{H}_g = 1} E[\bar{Y}^2_g(1,\pi_2)   | S_g] - \frac{1}{G_C} \sum_{1 \leq g \leq G: \tilde{H}_g = 0} E[\bar{Y}^2_g(1,\pi_2)   | S_g] \right | \\
&= \frac{1}{G} \left | \frac{1}{\pi_2} \sum_{1 \leq g \leq G: \tilde{H}_g = 1} E[\bar{Y}^2_g(1,\pi_2)   | S_g] - \frac{1}{1-\pi_2} \sum_{1 \leq g \leq G: \tilde{H}_g = 0} E[\bar{Y}^2_g(1,\pi_2)   | S_g] \right | \\
& \leq \frac{1}{G} \sum_{1 \leq j \leq n} k\cdot \max_{i,k\in\lambda_j} |E[\bar{Y}_{i}^2(1,\pi_2)   | S_{i}] - E[\bar{Y}_{k}^2(1,\pi_2)  | S_{k}]| \\
& \lesssim \frac{1}{n} \sum_{1 \leq j \leq n} \max_{i,k\in\lambda_j} |S_{i} - S_{k}|~.
\end{align*}
Therefore,
\[ \frac{1}{G_T} \sum_{1 \leq g \leq G} E[\bar{Y}^2_g(1,\pi_2)   | S_g] \tilde{H}_g \stackrel{P}{\to} E[\bar{Y}^2_g(1,\pi_2)  ]~. \]
Meanwhile,
\begin{align*}
& \frac{1}{G_T} \sum_{1 \leq g \leq G} E[\bar{Y}_g(1,\pi_2)   | S_g]^2 \tilde{H}_g  = \frac{1}{G} \sum_{1 \leq g \leq G} E[\bar{Y}_g(1,\pi_2)   | S_g]^2 \\
&\quad\quad\quad + (1-\pi_2) \left ( \frac{1}{G_T} \sum_{1 \leq g \leq G: \tilde{H}_g = 1} E[\bar{Y}_g(1,\pi_2)   | S_g]^2 - \frac{1}{G_C} \sum_{1 \leq g \leq G: \tilde{H}_g = 0} E[\bar{Y}_g(1,\pi_2)   | S_g]^2 \right )~.
\end{align*}
Jensen's inequality implies $E[E[\bar{Y}_g(1,\pi_2)   | S_g]^2] \leq E[\bar{Y}^2_g(1,\pi_2)  ] < E[\bar{Y}^2_g(1,\pi_2)] < \infty$ by Assumption \ref{ass:Q_G}(d), so it follows from the weak law of large numbers as above that
\[ \frac{1}{G} \sum_{1 \leq g \leq G} E[\bar{Y}_g(1,\pi_2)   | S_g]^2 \stackrel{P}{\to} E[E[\bar{Y}_g(1,\pi_2)   | S_g]^2]~. \]
Next, by Assumption \ref{as:close} and \ref{as:Q_G-lip}, the Cauchy-Schwarz inequality, and the fact that $(a + b)^2 \leq 2 a^2 + 2 b^2$,
{\footnotesize
\begin{align*}
& \left | \frac{1}{G_T} \sum_{1 \leq g \leq G: \tilde{H}_g = 1} E[\bar{Y}_g(1,\pi_2)   | S_g]^2 - \frac{1}{G_C} \sum_{1 \leq g \leq G: \tilde{H}_g = 0} E[\bar{Y}_g(1,\pi_2)   | S_g]^2 \right | \\
&= \frac{1}{G} \left | \frac{1}{\pi_2} \sum_{1 \leq g \leq G: \tilde{H}_g = 1} E[\bar{Y}_g(1,\pi_2)   | S_g]^2 - \frac{1}{1-\pi_2} \sum_{1 \leq g \leq G: \tilde{H}_g = 0} E[\bar{Y}_g(1,\pi_2)   | S_g]^2 \right | \\
&\leq \frac{1}{G} \sum_{1\leq j \leq n} \left(\max_{i,j\in\lambda_j} |E[\bar{Y}_{i}(1,\pi_2)   | S_{i}] - E[\bar{Y}_{k}(1,\pi_2)   | S_{k}]|  \right) \left(\sum_{k\in\lambda_j} E[\bar{Y}_{k}(1,\pi_2)   | S_{k}] \right) \\
&\lesssim  \left(\frac{1}{G} \sum_{1\leq j \leq n}\max_{i,j\in\lambda_j} |E[\bar{Y}_{i}(1,\pi_2)   | S_{i}] - E[\bar{Y}_{k}(1,\pi_2)   | S_{k}]|^2  \right)^{1/2}\left(\frac{1}{G} \sum_{1\leq j \leq n} \left(\sum_{k\in\lambda_j} E[\bar{Y}_{k}(1,\pi_2)   | S_{k}] \right)^2\right)^{1/2} \\
&\lesssim  \left(\frac{1}{G} \sum_{1\leq j \leq n}\max_{i,j\in\lambda_j} |E[\bar{Y}_{i}(1,\pi_2)   | S_{i}] - E[\bar{Y}_{k}(1,\pi_2)   | S_{k}]|^2  \right)^{1/2} \left(\frac{1}{G} \sum_{1\leq j \leq n}\sum_{k\in\lambda_j} E[\bar{Y}_{k}(1,\pi_2)   | S_{k}]^2\right)^{1/2}~.
\end{align*}}
Therefore, it follows from \eqref{eq:condlip} that
\[ \frac{1}{G_T} \sum_{1 \leq g \leq G} \var[\bar{Y}_g(1,\pi_2) | S_g] \tilde{H}_g  \stackrel{P}{\to} E[\var[\bar{Y}_g(1,\pi_2)   | S_g]]~. \]
Similar arguments together with Assumption \ref{as:Q_G-lip}(a)-(b) and Lemma \ref{lem:E_bounded} imply that
\[ \var \left [ \begin{pmatrix}
\mathbb L_{1,G}^{\rm Y1} \\
\mathbb L_{1, G}^{\rm YN1} \\
\mathbb L_{1, G}^{\rm N1}
\end{pmatrix} \Bigg |  S^{(G)}, H^{(G)} \right ] \stackrel{P}{\to} \mathbf V_1^1~. \]
Similarly,
\[ \var \left [ \begin{pmatrix}
\mathbb L_{1, G}^{\rm Y0} \\
\mathbb L_{1, G}^{\rm YN0} \\
\mathbb L_{1, G}^{\rm N0}
\end{pmatrix} \Bigg |  S^{(G)}, H^{(G)} \right ] \stackrel{P}{\to} \mathbf V_1^0~. \]
If $E[\var[\bar{Y}_g(1,\pi_2) N_g | S_g]] = E[\var[N_g | S_g]] = E[\var[\bar{Y}_g(0,0) N_g | S_g]] = 0$, then it follows from Markov's inequality conditional on $S^{(G)}$ and $H^{(G)}$, and the fact that probabilities are bounded and hence uniformly integrable, that $(\mathbb L_{1, G}^{\rm Y1}, \mathbb L_{1, G}^{\rm YN1}, \mathbb L_{1, G}^{\rm N1}, \mathbb L_{1, G}^{\rm Y0}, \mathbb L_{1, G}^{\rm YN0}, \mathbb L_{1, G}^{\rm N0}) \stackrel{P}{\to} 0$. Otherwise, it follows from similar arguments to those in the proof of Lemma S.1.5 of \cite{bai-inference} that
\begin{equation} \label{eq:cond}
\rho(\mathcal L((\mathbb L_{1, G}^{\rm Y1}, \mathbb L_{1, G}^{\rm YN1}, \mathbb L_{1, G}^{\rm N1}, \mathbb L_{1, G}^{\rm Y0}, \mathbb L_{1, G}^{\rm YN0}, \mathbb L_{1, G}^{\rm N0})' | S^{(G)}, H^{(G)}),  N(0, \mathbf V_1)) \stackrel{P}{\to} 0~,
\end{equation}
where $\mathcal L$ denotes the distribution and $\rho$ is any metric that metrizes weak convergence.

Next, I study $(\mathbb L_{2, G}^{\rm Y1}, \mathbb L_{2, G}^{\rm YN1}, \mathbb L_{2, G}^{\rm N1}, \mathbb L_{2, G}^{\rm Y0}, \mathbb L_{2, G}^{\rm YN0}, \mathbb L_{2, G}^{\rm N0})$. It follows from $Q_G = Q^{G}$ (by Lemma 5.1 of \cite{Bugni2022} and Assumption \ref{ass:Q_G} (a)-(b)) and Assumption \ref{ass:assignment-match} that
\[ \begin{pmatrix}
\mathbb L_{2, G}^{\rm Y1} \\
\mathbb L_{2, G}^{\rm YN1} \\
\mathbb L_{2, G}^{\rm N1} \\
\mathbb L_{2, G}^{\rm Y0} \\
\mathbb L_{2, G}^{\rm YN0} \\
\mathbb L_{2, G}^{\rm N0}
\end{pmatrix} = \begin{pmatrix}
\frac{1}{\sqrt{G_T}}\sum_{1 \leq g \leq G} \tilde{H}_g (E[\bar{Y}_g(1,\pi_2)  | S_g] - E[\bar{Y}_g(1,\pi_2)]) \\
\frac{1}{\sqrt{G_T}}\sum_{1 \leq g \leq G} \tilde{H}_g (E[\bar{Y}_g(1,\pi_2) N_g | S_g] - E[\bar{Y}_g(1,\pi_2) N_g]) \\
\frac{1}{\sqrt{G_T}} \sum_{1 \leq g \leq G} \tilde{H}_g (E[N_g | S_g] - E[N_g]) \\
\frac{1}{\sqrt{G_C}} \sum_{1 \leq g \leq G} (1 - \tilde{H}_g) (E[\bar{Y}_g(0,0)  | S_g] - E[\bar{Y}_g(0,0) ]) \\
\frac{1}{\sqrt{G_C}} \sum_{1 \leq g \leq G} (1 - \tilde{H}_g) (E[\bar{Y}_g(0,0) N_g | S_g] - E[\bar{Y}_g(0,0) N_g]) \\
\frac{1}{\sqrt{G_C}} \sum_{1 \leq g \leq G} (1 - \tilde{H}_g) (E[N_g | S_g] - E[N_g])
\end{pmatrix}~. \]
For $\mathbb L_{2, G}^{\rm Y1}$, note it follows from Assumption \ref{ass:assignment-match} and Lemma \ref{lemma:crd-cov} that
\begin{align*}
&\var[\mathbb L_{2, G}^{\rm Y1} | S^{(G)}] = \frac{1}{G_T} \sum_{1 \leq j \leq n} \var \left[\sum_{i=1}^{k} \tilde{H}_i (E[\bar{Y}_g(1,\pi_2)  | S_g] - E[\bar{Y}_g(1,\pi_2)]) \right] \\
&= \frac{1}{G_T} \sum_{1 \leq j \leq n} \pi_1(1-\pi_1)\left( (k-1) \sum_{i\in\lambda_j} (E[\bar{Y}_g(1,\pi_2)  | S_i] - E[\bar{Y}_g(1,\pi_2)])^2 \right.\\
&\quad\quad\quad\left. - \sum_{a\neq b} (E[\bar{Y}_g(1,\pi_2)  | S_a] - E[\bar{Y}_g(1,\pi_2)]) (E[\bar{Y}_g(1,\pi_2)  | S_b] - E[\bar{Y}_g(1,\pi_2)]) \right)\\
&\lesssim \frac{1}{n} \sum_{1 \leq j \leq n}  \sum_{i\in\lambda_j}\sum_{j\neq i} (E[\bar{Y}_g(1,\pi_2)  | S_i] - E[\bar{Y}_g(1,\pi_2)])(E[\bar{Y}_g(1,\pi_2)  | S_i]-E[\bar{Y}_g(1,\pi_2)  | S_j])\\
&\lesssim \frac{1}{n} \sum_{1 \leq j \leq n}  \sum_{i\in\lambda_j} (E[\bar{Y}_g(1,\pi_2)  | S_i] - E[\bar{Y}_g(1,\pi_2)])\left(\max_{i,k\in\lambda_j} \left|E[\bar{Y}_g(1,\pi_2)  | S_i]-E[\bar{Y}_g(1,\pi_2)  | S_k]\right|\right)\\
&\lesssim \left( \frac{1}{n} \sum_{1 \leq j \leq n}\max_{i,k\in\lambda_j} \left|E[\bar{Y}_g(1,\pi_2)  | S_i]-E[\bar{Y}_g(1,\pi_2)  | S_k]\right|^{2}\right)^{1/2}\\
& \lesssim \frac{1}{n} \sum_{1 \leq j \leq G} \max_{i,k\in\lambda_j} |S_i - S_k |^2 \stackrel{P}{\to} 0~.
\end{align*}
Therefore, it follows from Markov's inequality conditional on $S^{(G)}$ and $H^{(G)}$, and the fact that probabilities are bounded and hence uniformly integrable, that
\[ \mathbb L_{2, G}^{\rm Y1} = E[\mathbb L_{2, G}^{\rm Y1} | S^{(G)}] + o_P(1)~. \]
Similarly,
\[ \begin{pmatrix}
\mathbb L_{2, G}^{\rm Y1} \\
\mathbb L_{2, G}^{\rm YN1} \\
\mathbb L_{2, G}^{\rm N1} \\
\mathbb L_{2, G}^{\rm Y0} \\
\mathbb L_{2, G}^{\rm YN0} \\
\mathbb L_{2, G}^{\rm N0}
\end{pmatrix} = \begin{pmatrix}
\frac{1}{\sqrt G} \sqrt{\pi_1} \sum_{1 \leq g \leq G} (E[\bar{Y}_g(1,\pi_2)  | S_g] - E[\bar{Y}_g(1,\pi_2) ]) \\
\frac{1}{\sqrt G} \sqrt{\pi_1} \sum_{1 \leq g \leq G} (E[\bar{Y}_g(1,\pi_2) N_g | S_g] - E[\bar{Y}_g(1,\pi_2) N_g ]) \\
\frac{1}{\sqrt G} \sqrt{\pi_1} \sum_{1 \leq g \leq G} (E[N_g | S_g] - E[N_g]) \\
\frac{1}{\sqrt G} \sqrt{1-\pi_1} \sum_{1 \leq g \leq G} (E[\bar{Y}_g(0,0)  | S_g] - E[\bar{Y}_g(0,0) ]) \\
\frac{1}{\sqrt G} \sqrt{1-\pi_1} \sum_{1 \leq g \leq G} (E[\bar{Y}_g(0,0) N_g | S_g] - E[\bar{Y}_g(0,0) N_g]) \\
\frac{1}{\sqrt G} \sqrt{1-\pi_1} \sum_{1 \leq g \leq G} (E[N_g | S_g] - E[N_g])
\end{pmatrix} + o_P(1)~. \]
It then follows from Assumption \ref{ass:Q_G}(c)-(d) and \ref{as:Q_G-lip}(a) and the central limit theorem that
\[ (\mathbb L_{2, G}^{\rm Y1},\mathbb L_{2, G}^{\rm YN1}, \mathbb L_{2, G}^{\rm N1}, \mathbb L_{2, G}^{\rm Y0}, \mathbb L_{2, G}^{\rm YN0}, \mathbb L_{2, G}^{\rm N0})' \stackrel{d}{\to} N(0, \mathbf V_2)~. \]
Because \eqref{eq:cond} holds and $(\mathbb L_{2, G}^{\rm Y1},\mathbb L_{2, G}^{\rm YN1}, \mathbb L_{2, G}^{\rm N1}, \mathbb L_{2, G}^{\rm Y0}, \mathbb L_{2, G}^{\rm YN0}, \mathbb L_{2, G}^{\rm N0})$ is deterministic conditional on $S^{(G)}, H^{(G)}$, the conclusion of the theorem follows from Lemma S.1.3 in \cite{bai-inference}.
\end{proof}

\section{Lemmas for Proof of Theorem \ref{thm:variance-estimator-mt}}\label{sec:lemma-variance}

\begin{lemma}\label{lemma:assumptions-for-bai-inference}
If Assumption \ref{ass:Q_G}, \ref{ass:assignment-match}, \ref{as:assignment2} and \ref{as:Q_G-lip}(a) hold, then
\begin{enumerate}[(a)]
    \item $E[\bar Y_{g}^{r}(z,h) \mid S_g =s ]$ and $E[\tilde Y_{g}^{r}(z,h) \mid S_g =s ]$ are Lipschitz in $s$ for $(z,h) \in \{(1,\pi_2), (0,\pi_2), (0,0)\}$ and $r\in\{1,2\}$.
    \item $E\left[\bar Y_{g}^{2}(z,h)\right] < \infty$ and $E\left[\tilde Y_{g}^{2}(z,h)\right] < \infty$ for $(z,h) \in \{(1,\pi_2), (0,\pi_2), (0,0)\}$.
    \item $((\bar Y_{g}(1,\pi_2), \bar Y_{g}(0,\pi_2), \bar Y_{g}(0,0)):1\leq g\leq G) \perp H^{(G)} \mid S^{(G)}$ and $((\tilde Y_{g}(1,\pi_2), \tilde Y_{g}(0,\pi_2), \tilde Y_{g}(0,0)):1\leq g\leq G) \perp H^{(G)} \mid S^{(G)}$.
\end{enumerate}
\begin{proof}
    First, (a) is an immediate consequence of Assumption  \ref{as:Q_G-lip}(a). Also, (b) is an immediate consequence of Lemma \ref{lem:E_bounded} with Assumption \ref{ass:Q_G}. Finally, (c) follows directly by inspection and Assumption \ref{ass:assignment-match} and \ref{as:assignment2}.
\end{proof}
\end{lemma}

\begin{lemma}\label{lemma:var-mp-1}
Suppose $Q_G$ satisfies Assumptions \ref{ass:Q_G} and \ref{as:Q_G-lip} and the treatment assignment mechanism satisfies Assumptions \ref{ass:assignment-match}-\ref{as:assignment2}. Then, for $r=1,2$,
\begin{equation*}
    \frac{1}{n k(h)} \sum_{1\leq g \leq G}  \left(\bar Y_g^z\right)^r I \{H_g = h\} \xrightarrow{P} E[ \bar Y_g^r(z,h) ]~.
\end{equation*}
\end{lemma}
\begin{proof}
    I only prove the conclusion for $r = 1$ and the proof for $r = 2$ follows similarly. Note that
    \begin{multline*}
        \frac{1}{n k(h)} \sum_{1\leq g \leq G} \bar Y_g^z I \{H_g = h\} = \frac{1}{n k(h)} \sum_{1\leq g \leq G}  (\bar Y_g(z,h) I \{H_g = h\} - E[\bar Y_g(z,h) I \{H_g = h\} | S^{(G)}, H^{(G)}]) \\
        + \frac{1}{n k(h)} \sum_{1\leq g \leq G} E[\bar Y_g(z,h) I \{H_g = h\} | S^{(G)}, H^{(G)}]~.
    \end{multline*}
    By Lemma \ref{lemma:assumptions-for-bai-inference} (c), Assumption \ref{as:close} and similar arguments to those used in the proof of Lemma \ref{lemma:asymptotics-mp},
    \begin{align*}
        \frac{1}{n k(h)} \sum_{1\leq g \leq G} E[\bar Y_g(z,h) I \{H_g = h\} | S^{(G)}, H^{(G)}] &= \frac{1}{n k(h)} \sum_{1\leq g \leq G} I \{H_g = h\} E[\bar Y_g(z,h)  | S_g] \\
        & \xrightarrow{P} E[E[\bar Y_g(z,h)  | S_g]] = E[\bar Y_g(z,h)]~.
    \end{align*}
    By following the argument in Lemma S.1.5 of \cite{bai-inference}, we conclude that
    \[ \frac{1}{n k(h)} \sum_{1\leq g \leq G}  (\bar Y_g(z,h) I \{H_g = h\} - E[\bar Y_g(z,h) I \{H_g = h\} | S^{(G)}, H^{(G)}]) \stackrel{P}{\to} 0~. \]
    Therefore, the results hold.
\end{proof}

\begin{lemma}\label{lemma:var-mp-2}
Suppose $Q_G$ satisfies Assumptions \ref{ass:Q_G} and \ref{as:Q_G-lip} and the treatment assignment mechanism satisfies Assumptions \ref{ass:assignment-match}-\ref{as:assignment2}. Then, as $n\rightarrow \infty$,
\begin{equation*}
    \hat \rho_n^z(\pi_2, 0) \xrightarrow{P} E[ E[ \bar{Y}_g(z,\pi_2) \mid S_g] E[ \bar{Y}_g(z,0) \mid S_g] ] ~.
\end{equation*}
\end{lemma}
\begin{proof}
    To begin with, by Assumption \ref{ass:assignment-match}, 
    \begin{align*}
        &E[\hat \rho_n^z(\pi_2, 0) \mid S^{(G)} ]\\
        &= \frac{1}{n} \sum_{1 \leq j \leq n} \frac{1}{l(k-l)} E\left[  \Big ( \sum_{i \in \lambda_j} \bar Y_i^z I \{H_i = \pi_2\} \Big ) \Big ( \sum_{i \in \lambda_j} \bar Y_i^z I \{H_i = 0\} \Big ) \mid S^{(G)} \right] \\
        &= \frac{1}{n} \sum_{1 \leq j \leq n} \frac{1}{l(k-l)} \sum_{i \neq m \in \lambda_j} E\left[  \bar Y_i(z,\pi_2)    \mid S_i \right] E\left[   \bar Y_m(z,0)  \mid S_m \right] E\left[  I \{H_i = \pi_2\}   I \{H_m = 0\}  \mid S^{(G)} \right] \\
        &= \frac{1}{n} \sum_{1 \leq j \leq n} \frac{1}{l(k-l)} \sum_{i < m \in \lambda_j} (E\left[  \bar Y_i(z,\pi_2)    \mid S_i \right] E\left[   \bar Y_i(z,0)  \mid S_i \right] + E\left[  \bar Y_m(z,\pi_2)    \mid S_m \right] E\left[   \bar Y_m(z,0)  \mid S_m \right] \\
        &\quad - (E\left[  \bar Y_i(z,\pi_2)    \mid S_i \right] - E\left[  \bar Y_m(z,\pi_2)    \mid S_m \right]) (E\left[   \bar Y_i(z,0)  \mid S_i \right] - E\left[   \bar Y_m(z,0)  \mid S_m \right] ) ) \frac{l(k-l)}{k(k-1)} \\
        &= \frac{1}{n} \sum_{1 \leq j \leq n} \frac{1}{k} \sum_{i \in \lambda_j} E\left[  \bar Y_i(z,\pi_2)    \mid S_i \right] E\left[   \bar Y_i(z,0)  \mid S_i \right] \\
        &\quad - \frac{1}{n} \sum_{1 \leq j \leq n} \frac{1}{k(k-1)} \sum_{i < m \in \lambda_j}(E\left[  \bar Y_i(z,\pi_2)    \mid S_i \right] - E\left[  \bar Y_m(z,\pi_2)    \mid S_m \right]) (E\left[   \bar Y_i(z,0)  \mid S_i \right] - E\left[   \bar Y_m(z,0)  \mid S_m \right] )~.
    \end{align*}
    Then, by Lipschitz condition from Lemma \ref{lemma:assumptions-for-bai-inference}(a), Lemma \ref{lem:E_bounded} and Assumption \ref{as:close}, we conclude that $E[\hat \rho_n^z(\pi_2, 0) \mid S^{(G)} ] \xrightarrow{P} E[ E[ \bar{Y}_g(z,\pi_2) \mid S_g] E[ \bar{Y}_g(z,0) \mid S_g] ]$. To conclude the proof, we need to show
    \begin{equation*}
        \hat \rho_n^z(\pi_2, 0) - E[\hat \rho_n^z(\pi_2, 0) \mid S_g ] \xrightarrow{P} 0 ~.
    \end{equation*}
    Define
    \begin{equation*}
        \hat \rho_{n,j}^z(\pi_2, 0) = \frac{1}{l(k-l)}  \Big ( \sum_{i \in \lambda_j} \bar Y_i^z I \{H_i = \pi_2\} \Big ) \Big ( \sum_{i \in \lambda_j} \bar Y_i^z I \{H_i = 0\} \Big )~.
    \end{equation*}
    Note that
    {\footnotesize
    \begin{align*}
        &\left| E[\hat \rho_{n,j}^z(\pi_2, 0) \mid S^{(G)}] \right| I \left\{ \left|  E[\hat \rho_{n,j}^z(\pi_2, 0) \mid S^{(G)}] \right| > \lambda \right\} \\
        &= \left| \frac{1}{k(k-1)} \sum_{i \neq m \in \lambda_j} E\left[  \bar Y_i(z,\pi_2)    \mid S_i \right] E\left[   \bar Y_m(z,0)  \mid S_m \right] \right|I \left\{ \left|  \frac{1}{k(k-1)} \sum_{i \neq m \in \lambda_j} E\left[  \bar Y_i(z,\pi_2)    \mid S_i \right] E\left[   \bar Y_m(z,0)  \mid S_m \right] \right| > \lambda \right\} \\
        &\leq \sum_{i \neq m \in \lambda_j} \left|E\left[  \bar Y_i(z,\pi_2)    \mid S_i \right] E\left[   \bar Y_m(z,0)  \mid S_m \right] \right| I \left\{ \left|  E\left[  \bar Y_i(z,\pi_2)    \mid S_i \right] E\left[   \bar Y_m(z,0)  \mid S_m \right] \right| > \lambda \right\}~.
    \end{align*}
    }
    Then, the conclusion follows by repeating the same arguments from Lemma C.2 of \cite{matched-tuple}.
\end{proof}

\begin{lemma}\label{lemma:var-mp-3}
Suppose $Q_G$ satisfies Assumptions \ref{ass:Q_G} and \ref{as:Q_G-lip} and the treatment assignment mechanism satisfies Assumptions \ref{ass:assignment-match}-\ref{as:assignment2}. Then, as $n\rightarrow \infty$,
\begin{equation*}
    \hat \rho_n^z(h,h) \xrightarrow{P} E[ E[ \bar{Y}_g(z,h) \mid S_g]^2 ] ~.
\end{equation*}
\end{lemma}
\begin{proof}
    To begin with, by Assumption \ref{ass:assignment-match},
    \begin{align*}
        &E[\hat \rho_n^z(h,h) \mid S^{(G)}]\\
        &= \frac{2}{n} \sum_{1 \leq j \leq \lfloor n / 2 \rfloor} \frac{1}{k^2(h)} E\left[ \Big ( \sum_{i \in \lambda_{2j-1}} \bar Y_i^z I \{H_i = h\} \Big ) \Big ( \sum_{i \in \lambda_{2j}} \bar Y_i^z I \{H_i = h\} \Big ) \mid S^{(G)} \right] \\
        &= \frac{2}{n} \sum_{1 \leq j \leq \lfloor n / 2 \rfloor} \frac{1}{k^2(h)} \frac{k^2(h)}{k^2} \sum_{i \in \lambda_{2j-1}, k \in \lambda_{2j}} E[\bar Y_i^z(z,h)\mid S_i] E[\bar Y_k^z(z,h)\mid S_k] \\
        &= \frac{2}{n} \sum_{1 \leq j \leq \lfloor n / 2 \rfloor} \frac{1}{k^2} \sum_{i \in \lambda_{2j-1}, k \in \lambda_{2j}} \left( \frac{1}{2} E[\bar Y_i^z(z,h)\mid S_i]^2 + \frac{1}{2} E[\bar Y_k^z(z,h)\mid S_k]^2 - \frac{1}{2} (E[\bar Y_i^z(z,h)\mid S_i] - E[\bar Y_k^z(z,h)\mid S_k])^2 \right) \\
        &= \frac{1}{G} \sum_{1 \leq g \leq G} E[\bar Y_g^z(z,h)\mid S_g]^2 - \frac{1}{n k^2} \sum_{1 \leq j \leq \lfloor n / 2 \rfloor} \sum_{i \in \lambda_{2j-1}, k \in \lambda_{2j}}(E[\bar Y_i^z(z,h)\mid S_i] - E[\bar Y_k^z(z,h)\mid S_k])^2 \\
        &\xrightarrow{P} E[ E[ \bar{Y}_g(z,h) \mid S_g]^2 ]~,
    \end{align*}
    where the convergence in probability follows from Lemma \ref{lemma:assumptions-for-bai-inference}(a), Assumption \ref{as:close}, Lemma \ref{lem:E_bounded} and weak law of large numbers. To conclude the proof, we need to show
    \begin{equation*}
        \hat \rho_n^z(h, h) - E[\hat \rho_n^z(h, h) \mid S^{(G)} ] \xrightarrow{P} 0 ~.
    \end{equation*}
    Define
    \begin{equation*}
        \hat \rho_{n,j}^z(h, h) = \frac{1}{k^2(h)}  \Big ( \sum_{i \in \lambda_{2j-1}} \bar Y_i^z I \{H_i = h\} \Big ) \Big ( \sum_{i \in \lambda_{2j}} \bar Y_i^z I \{H_i = h\} \Big ) ~.
    \end{equation*}
    Note that
    {\footnotesize
    \begin{align*}
        &\left| E[\hat \rho_{n,j}^z(\pi_2, 0) \mid S^{(G)}] \right| I \left\{ \left|  E[\hat \rho_{n,j}^z(\pi_2, 0) \mid S^{(G)}] \right| > \lambda \right\} \\
        &= \left| \frac{1}{k^2} \sum_{i \in \lambda_{2j-1}, k \in \lambda_{2j}} E[\bar Y_i^z(z,h)\mid S_i] E[\bar Y_k^z(z,h)\mid S_k] \right|I \left\{ \left|  \frac{1}{k^2} \sum_{i \in \lambda_{2j-1}, k \in \lambda_{2j}} E[\bar Y_i^z(z,h)\mid S_i] E[\bar Y_k^z(z,h)\mid S_k] \right| > \lambda \right\} \\
        &\leq \sum_{i \in \lambda_{2j-1}, k \in \lambda_{2j}} \left|E[\bar Y_i^z(z,h)\mid S_i] E[\bar Y_k^z(z,h)\mid S_k] \right| I \left\{ \left|  E[\bar Y_i^z(z,h)\mid S_i] E[\bar Y_k^z(z,h)\mid S_k] \right| > \lambda \right\}~.
    \end{align*}
    }
    Then, the conclusion follows by repeating the same arguments from Lemma C.3 of \cite{matched-tuple}.
\end{proof}

\begin{lemma}\label{lemma:var-mp-size-1}
Suppose $Q_G$ satisfies Assumptions \ref{ass:Q_G} and \ref{as:Q_G-lip} and the treatment assignment mechanism satisfies Assumptions \ref{ass:assignment-match}-\ref{as:assignment2}. Then, for $r=1,2$,
\begin{equation*}
    \frac{1}{n k(h)} \sum_{1\leq g \leq G}  \left(\tilde Y_g^z\right)^r I \{H_g = h\} \xrightarrow{P} E[ \tilde Y_g^r(z,h) ]~.
\end{equation*}
\end{lemma}
\begin{proof}
    I only prove the conclusion for $r = 1$ and the proof for $r = 2$ follows similarly. Note that
    \begin{align*}
        \frac{1}{n k(h)} \sum_{1\leq g \leq G} \tilde Y_g^z I \{H_g = h\} = \frac{1}{n k(h)} \sum_{1\leq g \leq G} \tilde Y_g(z,h) I \{H_g = h\} + \frac{1}{n k(h)} \sum_{1\leq g \leq G} \left(\hat Y_g^z(h) -\tilde Y_g(z,h) \right) I \{H_g = h\} ~,
    \end{align*}
    where $\hat Y_g^z(h)$ is defined in (\ref{eqn:hatY(h)}). Note that
    \begin{align*}
        &\frac{1}{n k(h)} \sum_{1\leq g \leq G} \left(\hat Y_g^z(h) -\tilde Y_g(z,h) \right) I \{H_g = h\}\\
        &=\left( \frac{1}{\frac{1}{G} \sum_{1\leq g \leq G} N_g} - \frac{1}{E[N_g]} \right)\left( \frac{1}{n k(h)} \sum_{1\leq g \leq G} \Bar{Y}_g(z,h) N_gI \{H_g = h\}\right) \\
        &\quad - \left( \frac{\frac{1}{G}\sum_{1\leq g \leq G} \bar Y_g(z,h) I\{H_g=h\} N_g}{\left(\frac{1}{G} \sum_{1\leq g \leq G} N_g \right)^2} - \frac{E[\Bar{Y}_g(z,h) N_g]}{E[N_g]^2} \right) \left( \frac{1}{n k(h)} \sum_{1\leq g \leq G}  N_gI \{H_g = h\}\right)
    \end{align*}
    By weak law of large number, Lemma \ref{lemma:asymptotics-mp} and Slutsky's theorem, we have 
    \begin{equation*}
        \frac{1}{n k(h)} \sum_{1\leq g \leq G} \left(\hat Y_g^z(h) -\tilde Y_g(z,h) \right) I \{H_g = h\} \xrightarrow{P} 0 ~.
    \end{equation*}
    By Lemma \ref{lemma:assumptions-for-bai-inference} and repeating the arguments in Lemma \ref{lemma:var-mp-1} with $\tilde Y_g(z,h)$ in the place of $\bar Y_g(z,h)$, we have
    \begin{equation*}
        \frac{1}{n k(h)} \sum_{1\leq g \leq G} \tilde Y_g(z,h) I \{H_g = h\} \xrightarrow{P} E[ \tilde Y_g^r(z,h) ]~.
    \end{equation*}
    Thus, the result follows.
\end{proof}

\begin{lemma}\label{lemma:var-mp-size-2}
Suppose $Q_G$ satisfies Assumptions \ref{ass:Q_G} and \ref{as:Q_G-lip} and the treatment assignment mechanism satisfies Assumptions \ref{ass:assignment-match}-\ref{as:assignment2}. Then, as $n\rightarrow \infty$,
\begin{equation*}
    \frac{1}{n} \sum_{1 \leq j \leq n} \frac{1}{l(k-l)} \Big ( \sum_{i \in \lambda_j} \tilde Y_i^z I \{H_i = \pi_2\} \Big ) \Big ( \sum_{i \in \lambda_j} \tilde Y_i^z I \{H_i = 0\} \Big ) \xrightarrow{P} E[ E[ \bar{Y}_g(z,\pi_2) \mid S_g] E[ \bar{Y}_g(z,0) \mid S^{(G)}] ] ~.
\end{equation*}
\end{lemma}
\begin{proof}
    Note that
    \begin{align*}
        &\frac{1}{n} \sum_{1 \leq j \leq n} \frac{1}{l(k-l)} \Big ( \sum_{i \in \lambda_j} \tilde Y_i^z I \{H_i = \pi_2\} \Big ) \Big ( \sum_{i \in \lambda_j} \tilde Y_i^z I \{H_i = 0\} \Big )\\
        &= \frac{1}{n} \sum_{1 \leq j \leq n} \frac{1}{l(k-l)} \sum_{i\neq m \in \lambda_j} \hat Y_i^z(\pi_2) \hat Y_m^z(0) I \{H_i = \pi_2, H_m = 0\} \\
        &= \frac{1}{n} \sum_{1 \leq j \leq n} \frac{1}{l(k-l)} \sum_{i\neq m \in \lambda_j} \tilde Y_i(z,\pi_2) \tilde Y_m(z,0)  I \{H_i = \pi_2, H_m = 0\}\\
        &\quad + \frac{1}{n} \sum_{1 \leq j \leq n} \frac{1}{l(k-l)} \sum_{i\neq m \in \lambda_j} \left(\hat Y_i^z(\pi_2) \hat Y_m^z(0) - \tilde Y_i(z,\pi_2) \tilde Y_m(z,0)\right)  I \{H_i = \pi_2, H_m = 0\}~.
    \end{align*}
    The second term can be written as
    \begin{multline}\label{eqn:rewrting}
        \frac{1}{n} \sum_{1 \leq j \leq n} \frac{1}{l(k-l)} \sum_{i\neq m \in \lambda_j} \left(\hat Y_i^z(\pi_2)  - \tilde Y_i(z,\pi_2)\right)\tilde Y_m(z,0)I \{H_i = \pi_2, H_m = 0\}\\
        + \frac{1}{n} \sum_{1 \leq j \leq n} \frac{1}{l(k-l)} \sum_{i\neq m \in \lambda_j}\left(\hat Y_m^z(0)- \tilde Y_m(z,0)\right)\tilde Y_i(z,\pi_2) I \{H_i = \pi_2, H_m = 0\}\\
        + \frac{1}{n} \sum_{1 \leq j \leq n} \frac{1}{l(k-l)} \sum_{i\neq m \in \lambda_j}\left(\hat Y_i^z(\pi_2)  - \tilde Y_i(z,\pi_2)\right)\left(\hat Y_m^z(0) - \tilde Y_m(z,0))\right)I \{H_i = \pi_2, H_m = 0\} ~.
    \end{multline}
    We show that the first term of (\ref{eqn:rewrting}) converges to zero in probability and the other two terms should follow the same arguments:
    {\footnotesize
    \begin{align*}
        &\frac{1}{n} \sum_{1 \leq j \leq n} \frac{1}{l(k-l)} \sum_{i\neq m \in \lambda_j} \left(\hat Y_i^z(\pi_2)  - \tilde Y_i(z,\pi_2)\right) \tilde Y_m(z,0)  I \{H_i = \pi_2, H_m = 0\} \\
        &= \left( \frac{1}{\frac{1}{G} \sum_{1\leq g \leq G} N_g} - \frac{1}{E[N_g]} \right) \left(\frac{1}{n} \sum_{1 \leq j \leq n} \frac{1}{l(k-l)} \sum_{i\neq m \in \lambda_j} \bar Y_i(z,\pi_2) N_i   \tilde Y_m(z,0) I \{H_i = \pi_2, H_m = 0\} \right) \\
        &\quad - \left( \frac{\frac{1}{G}\sum_{1\leq g \leq G} \bar Y_g(z,h) I\{H_g=h\} N_g}{\left(\frac{1}{G} \sum_{1\leq g \leq G} N_g \right)^2} - \frac{E[\Bar{Y}_g(z,h) N_g]}{E[N_g]^2} \right)\left(\frac{1}{n} \sum_{1 \leq j \leq n} \frac{1}{l(k-l)} \sum_{i\neq m \in \lambda_j} N_i  \tilde Y_m(z,0) I \{H_i = \pi_2, H_m = 0\} \right)
    \end{align*}
    }
    By following the same argument in Lemma S.1.6 from \cite{bai-inference}, we have
    \begin{align*}
        \frac{1}{n} \sum_{1 \leq j \leq n} \frac{1}{l(k-l)} \sum_{i\neq m \in \lambda_j} \bar Y_g(z,\pi_2) N_g   \tilde Y_m(z,0) I \{H_i = \pi_2, H_m = 0\} &\xrightarrow{P} E[E[N_g \Bar Y_g(z,\pi_2)\mid S_g] E[ \tilde Y_m(z,0)\mid S_g]] \\
        \frac{1}{G} \sum_{1 \leq j \leq G} N_{\pi(2 j)} \tilde Y_{\pi(2 j-1)}(0)I \{H_i = \pi_2, H_m = 0\} &\xrightarrow{P} E[E[N_g\mid S_g]E[ \tilde Y_m(z,0)\mid S_g]]~.
    \end{align*}
    By weak law of large number, Lemma \ref{lemma:asymptotics-mp} and Slutsky's theorem, we have 
    \begin{equation*}
        \frac{1}{n} \sum_{1 \leq j \leq n} \frac{1}{l(k-l)} \sum_{i\neq m \in \lambda_j} \left(\hat Y_i^z(\pi_2)  - \tilde Y_i(z,\pi_2)\right)  \tilde Y_m(z,0) I \{H_i = \pi_2, H_m = 0\} \xrightarrow{P} 0~.
    \end{equation*}
    Similarly, the convergence in probability to zero should hold for all three terms in (\ref{eqn:rewrting}).
    Thus, we have
    \begin{equation*}
        \frac{1}{n} \sum_{1 \leq j \leq n} \frac{1}{l(k-l)} \sum_{i\neq m \in \lambda_j} \left(\hat Y_i^z(\pi_2) \hat Y_m^z(0) - \tilde Y_i(z,\pi_2) \tilde Y_m(z,0)\right)  I \{H_i = \pi_2, H_m = 0\} \rightarrow 0~.
    \end{equation*}
    By Lemma \ref{lemma:assumptions-for-bai-inference} and repeating the arguments in Lemma \ref{lemma:var-mp-2} with $\tilde Y_g(z,h)$ in the place of $\bar Y_g(z,h)$, we conclude the result.
\end{proof}

\begin{lemma}\label{lemma:var-mp-size-3}
Suppose $Q_G$ satisfies Assumptions \ref{ass:Q_G} and \ref{as:Q_G-lip} and the treatment assignment mechanism satisfies Assumptions \ref{ass:assignment-match}-\ref{as:assignment2}. Then, as $n\rightarrow \infty$,
\begin{equation*}
    \frac{2}{n} \sum_{1 \leq j \leq \lfloor n / 2 \rfloor} \frac{1}{k^2(h)} \Big ( \sum_{i \in \lambda_{2j-1}} \tilde Y_i^z I \{H_i = h\} \Big ) \Big ( \sum_{i \in \lambda_{2j}} \tilde Y_i^z I \{H_i = h\} \Big )   \xrightarrow{P} E[ E[ \tilde{Y}_g(z,h) \mid S_g]^2 ] ~.
\end{equation*}
\begin{proof}
    Note that
    \begin{align*}
        &\frac{2}{n} \sum_{1 \leq j \leq \lfloor n / 2 \rfloor} \frac{1}{k^2(h)} \Big ( \sum_{i \in \lambda_{2j-1}} \tilde Y_i^z I \{H_i = h\} \Big ) \Big ( \sum_{i \in \lambda_{2j}} \tilde Y_i^z I \{H_i = h\} \Big ) \\
        &= \frac{2}{n} \sum_{1 \leq j \leq \lfloor n / 2 \rfloor} \frac{1}{k^2(h)} \sum_{i \in \lambda_{2j-1}, m\in \lambda_{2j}} \hat Y_i^z(h) \hat Y_m^z(h) I \{H_i = H_m = h\} \\
        &= \frac{2}{n} \sum_{1 \leq j \leq \lfloor n / 2 \rfloor} \frac{1}{k^2(h)} \sum_{i \in \lambda_{2j-1}, m\in \lambda_{2j}} \tilde Y_i(z,h) \tilde Y_m(z,h) I \{H_i = H_m = h\} \\
        &\quad + \frac{2}{n} \sum_{1 \leq j \leq \lfloor n / 2 \rfloor} \frac{1}{k^2(h)} \sum_{i \in \lambda_{2j-1}, m\in \lambda_{2j}} \left(\hat Y_i^z(h) \hat Y_m^z(h) - \tilde Y_i(z,h) \tilde Y_m(z,h) \right) I \{H_i = H_m = h\}
    \end{align*}
    The second term can be written as
    \begin{multline}\label{eqn:rewrting2}
        \frac{2}{n} \sum_{1 \leq j \leq \lfloor n / 2 \rfloor} \frac{1}{k^2(h)} \sum_{i \in \lambda_{2j-1}, k\in \lambda_{2j}} \left(\hat Y_i^z(h)  - \tilde Y_i(z,h)\right)\tilde Y_m(z,h)I \{H_i = H_m = h\}\\
        + \frac{2}{n} \sum_{1 \leq j \leq \lfloor n / 2 \rfloor} \frac{1}{k^2(h)} \sum_{i \in \lambda_{2j-1}, k\in \lambda_{2j}}\left(\hat Y_m^z(h)- \tilde Y_m(z,h)\right)\tilde Y_i(z,h) I \{H_i =  H_m = h\}\\
        + \frac{2}{n} \sum_{1 \leq j \leq \lfloor n / 2 \rfloor} \frac{1}{k^2(h)} \sum_{i \in \lambda_{2j-1}, k\in \lambda_{2j}} \left(\hat Y_i^z(h)  - \tilde Y_i(z,h)\right)\left(\hat Y_m^z(h) - \tilde Y_m(z,h))\right)I \{H_i = H_m = h\} ~.
    \end{multline}
    We show that the first term of (\ref{eqn:rewrting2}) converges to zero in probability and the other two terms should follow the same arguments:
    {\footnotesize
    \begin{align*}
        &\frac{2}{n} \sum_{1 \leq j \leq \lfloor n / 2 \rfloor} \frac{1}{k^2(h)} \sum_{i \in \lambda_{2j-1}, k\in \lambda_{2j}} \left(\hat Y_i^z(h)  - \tilde Y_i(z,h)\right)\tilde Y_m(z,h)I \{H_i = H_m = h\} \\
        &= \left( \frac{1}{\frac{1}{G} \sum_{1\leq g \leq G} N_g} - \frac{1}{E[N_g]} \right) \left(\frac{2}{n} \sum_{1 \leq j \leq \lfloor n / 2 \rfloor} \frac{1}{k^2(h)} \sum_{i \in \lambda_{2j-1}, k\in \lambda_{2j}} \bar Y_i(z,h) N_i   \tilde Y_m(z,h) I \{H_i = H_m = h\} \right) \\
        & - \left( \frac{\frac{1}{G}\sum_{1\leq g \leq G} \bar Y_g(z,h) I\{H_g=h\} N_g}{\left(\frac{1}{G} \sum_{1\leq g \leq G} N_g \right)^2} - \frac{E[\Bar{Y}_g(z,h) N_g]}{E[N_g]^2} \right)\left(\frac{2}{n} \sum_{1 \leq j \leq \lfloor n / 2 \rfloor} \frac{1}{k^2(h)} \sum_{i \in \lambda_{2j-1}, k\in \lambda_{2j}} N_i  \tilde Y_m(z,h) I \{H_i =  H_m = h\} \right)
    \end{align*}
    }
    By following the same argument in Lemma S.1.6 from \cite{bai-inference}, we have
    \begin{align*}
        \frac{2}{n} \sum_{1 \leq j \leq \lfloor n / 2 \rfloor} \frac{1}{k^2(h)} \sum_{i \in \lambda_{2j-1}, k\in \lambda_{2j}} \bar Y_i(z,h) N_i   \tilde Y_m(z,h) I \{H_i = H_m = h\} &\xrightarrow{P} E[ E[Y_g(z,h) N_g \mid S_g] E[ Y_g(z,h) \mid S_g] ] \\
        \frac{2}{n} \sum_{1 \leq j \leq \lfloor n / 2 \rfloor} \frac{1}{k^2(h)} \sum_{i \in \lambda_{2j-1}, k\in \lambda_{2j}} N_i  \tilde Y_m(z,h) I \{H_i =  H_m = h\} &\xrightarrow{P} E[ E[ N_g\mid S_g] E[Y_g(z,h)\mid S_g] ]~.
    \end{align*}
    By weak law of large number, Lemma \ref{lemma:asymptotics-mp} and Slutsky's theorem, we have 
    \begin{equation*}
        \frac{2}{n} \sum_{1 \leq j \leq \lfloor n / 2 \rfloor} \frac{1}{k^2(h)} \sum_{i \in \lambda_{2j-1}, k\in \lambda_{2j}} \left(\hat Y_i^z(h)  - \tilde Y_i(z,h)\right)\tilde Y_m(z,h)I \{H_i = H_m = h\} \xrightarrow{P} 0~.
    \end{equation*}
    Similarly, the convergence in probability to zero should hold for all three terms in (\ref{eqn:rewrting}).
    Thus, we have
    \begin{equation*}
        \frac{2}{n} \sum_{1 \leq j \leq \lfloor n / 2 \rfloor} \frac{1}{k^2(h)} \sum_{i \in \lambda_{2j-1}, m\in \lambda_{2j}} \left(\hat Y_i^z(h) \hat Y_m^z(h) - \tilde Y_i(z,h) \tilde Y_m(z,h) \right) I \{H_i = H_m = h\} \rightarrow 0~.
    \end{equation*}
    By Lemma \ref{lemma:assumptions-for-bai-inference} and repeating the arguments in Lemma \ref{lemma:var-mp-3} with $\tilde Y_g(z,h)$ in the place of $\bar Y_g(z,h)$, we conclude the result.
\end{proof}
\end{lemma}

\newpage
\section{Proof of Theorem \ref{thm:adj}}
\subsection{Limit of Regression Coefficient}
Recall that $\hat\beta^{P}_2$ is an OLS estimator of the slope coefficient in the linear regression of $\hat\mu_{1,j} - \hat\mu_{0,j}$ on a constant and $\hat\psi_{1,j} - \hat\psi_{0,j}$, where
\begin{align*}
    \hat\mu_{1,j} &= \frac{1}{l} \sum_{g\in\lambda_j} \tilde{Y}_g^1 \bar N_G I\{H_g = \pi_2\} \\
    \hat\mu_{0,j} &= \frac{1}{k-l} \sum_{g\in\lambda_j} \tilde{Y}_g^1 \bar N_G I\{H_g = 0\} \\
    \hat\psi_{1,j} &= \frac{1}{l} \sum_{g\in\lambda_j} \psi_g I\{H_g = \pi_2\} \\
    \hat\psi_{0,j} &= \frac{1}{k-l} \sum_{g\in\lambda_j} \psi_g I\{H_g = 0\} ~.
\end{align*}
Note that 
\begin{align*}
    &\frac{1}{n} \sum_{j=1}^{n} (\hat\psi_{1,j} - \hat\psi_{0,j}) (\hat\psi_{1,j} - \hat\psi_{0,j})^{\prime} \\
    &= \frac{1}{l^2} \sum_{g\in \lambda_j} \psi_g \psi_g^\prime I\{H_g=\pi_2\} + \frac{1}{(k-l)^2} \sum_{g\in \lambda_j} \psi_g \psi_g^\prime I\{H_g=0\}\\
    &\hspace{3em} + \frac{1}{l^2} \sum_{g,q\in \lambda_j } \psi_g \psi_q^\prime I\{H_g = H_q = \pi_2\} + \frac{1}{(k-l)^2} \sum_{g,q\in \lambda_j } \psi_g \psi_q^\prime I\{H_g = H_q = 0\} \\
    &\hspace{6em} - \frac{1}{l(k-l)} \sum_{g,q\in \lambda_j } \psi_g \psi_q^\prime I\{H_g = \pi_2,  H_q = 0\} - \frac{1}{l(k-l)} \sum_{g,q\in \lambda_j } \psi_g \psi_q^\prime I\{H_g = 0,  H_q = \pi_2\} ~.
\end{align*}
Following the argument in A.4 of \cite{bai2023}, we have
\begin{align*}
    &\frac{1}{l^2} \sum_{g\in \lambda_j} \psi_g \psi_g^\prime I\{H_g=\pi_2\} + \frac{1}{(k-l)^2} \sum_{g\in \lambda_j} \psi_g \psi_g^\prime I\{H_g=0\}  \xrightarrow{p} \frac{k}{l(k-l)} E[\psi_g \psi_g^\prime]
\end{align*}
and
\begin{align*}
    &\frac{1}{l^2} \sum_{g,q\in \lambda_j } \psi_g \psi_q^\prime I\{H_g = H_q = \pi_2\} + \frac{1}{(k-l)^2} \sum_{g,q\in \lambda_j } \psi_g \psi_q^\prime I\{H_g = H_q = 0\} \\
    &\hspace{3em} - \frac{1}{l(k-l)} \sum_{g,q\in \lambda_j } \psi_g \psi_q^\prime I\{H_g = \pi_2,  H_q = 0\} - \frac{1}{l(k-l)} \sum_{g,q\in \lambda_j } \psi_g \psi_q^\prime I\{H_g = 0,  H_q = \pi_2\}\\
    &\xrightarrow{p} - \frac{k}{l(k-l)} E[E[\psi_g\mid S_g] E[\psi_g^\prime\mid S_g]] ~.
\end{align*}
In other words, we have 
\begin{equation*}
    \frac{1}{n} \sum_{j=1}^{n} (\hat\psi_{1,j} - \hat\psi_{0,j}) (\hat\psi_{1,j} - \hat\psi_{0,j})^{\prime} \xrightarrow{p} \frac{k}{l(k-l)} E[\var[\psi_g \mid S_g]] ~.
\end{equation*}
Similarly, following the argument in A.8 of \cite{mp-cluster},  we have
\begin{align*}
    \frac{1}{n} \sum_{j=1}^{n}  (\hat\psi_{1,j} - \hat\psi_{0,j}) (\hat\mu_{1,j} - \hat\mu_{0,j}) &\xrightarrow{p} \frac{E[N_g]}{l}  \left( E[\psi_g \tilde Y_g(1,\pi_2)]  - E[ E[\psi_g \mid S_g] E[\tilde Y_g(1,\pi_2)  \mid S_g] ] \right) \\
    &\hspace{3em} + \frac{E[N_g]}{k-l}  \left(E[\psi_g \tilde Y_g(0,0) ] - E[ E[\psi_g \mid S_g] E[\tilde Y_g(0,0) \mid S_g] ] \right) \\
    &= E\left[  \cov \left[\frac{1}{l} \tilde Y_g(1, \pi_2)  + \frac{1}{k-l}  \tilde Y_g(0,0),  \psi_g\mid S_g\right] \right] E[N_g]
\end{align*}
Therefore,
\begin{equation*}
    \hat\beta_2^P \xrightarrow{p} \pi_1(1-\pi_1) (E[\var[\psi_g \mid S_g]])^{-1}  E\left[  \cov \left[\frac{1}{\pi_1} \tilde Y_g(1, \pi_2)  + \frac{1}{1-\pi_1}  \tilde Y_g(0,0) ,  \psi_g\mid S_g\right] \right] E[N_g] = \beta^P_2~.
\end{equation*}

\subsection{Asymptotic Normality}
To establish the limiting distribution, first define
\begin{align*}
    \bar\psi_1 = \frac{1}{G_T} \sum_{1 \leq g \leq G} \psi_g I\{H_g = \pi_2\} \\
    \bar\psi_0 = \frac{1}{G_C} \sum_{1 \leq g \leq G} \psi_g I\{H_g = 0\} ~.
\end{align*}
Let $\tilde H_g = I\{H_g = \pi_2\}$. Note that
\begin{align*}
    & \frac{1}{G_T} \sum_{1 \leq g \leq G} (\bar Y_g(1,\pi_2) N_g - (\psi_g - \bar \psi_G)' \hat \beta_2^P) \tilde H_g \\
    & = \frac{1}{G_T} \sum_{1 \leq g \leq G} (\bar Y_g(1,\pi_2) N_g - (\psi_g - \bar \psi_G)' \beta^P_2) \tilde H_g - \frac{1}{G_T} \sum_{1 \leq g \leq G} (\psi_g - \bar \psi_{1, G})' (\hat \beta_2^P - \beta_2^P) \tilde H_g - (\bar \psi_{1, G} - \bar \psi_G)' (\hat \beta_2^P - \beta_2^P) \\
    & = \frac{1}{G_T} \sum_{1 \leq g \leq G} (\bar Y_g(1,\pi_2) N_g - (\psi_g - \bar \psi_G)' \beta_2^P) \tilde H_g - O_P(G^{-1/2}) o_P(1) \\
    & = \frac{1}{G_T} \sum_{1 \leq g \leq G} (\bar Y_g(1,\pi_2) N_g - (\psi_g - \bar \psi_G)' \beta_2^P) \tilde H_g + o_P(G^{-1/2}) \\
    & = \frac{1}{G_T} \sum_{1 \leq g \leq G} (\bar Y_g(1,\pi_2) N_g - (\psi_g - E[\psi_g])' \beta_2^P) \tilde H_g - (\bar \psi_G - E[\psi_g])' \beta_2^P + o_P(G^{-1/2})~.
\end{align*}
where the second equality follows because $\hat \beta_2^P - \beta_2^P = o_P(1)$,
\[ \frac{1}{G_T} \sum_{1 \leq g \leq G} (\psi_g - \bar \psi_{1, G})\tilde H_g = 0~, \]
and
\[ \sqrt{G_T} (\bar \psi_{1, G} - \bar \psi_G) = O_P(1)~. \]
The last equality follows from the arguments that establish (50) in \cite{bai2023}. Define
\begin{align*}
    \tilde \theta_2^{P, adj} = \frac{1}{N_T} \sum_{1\leq g \leq G} I\{H_g = \pi_2 \} (N_g \bar{Y}_{g}^1 -(\psi_g - E[\psi_g])' \beta_2^P)  - \frac{1}{N_C} \sum_{1\leq g \leq G}I\{H_g = 0 \} (N_g \bar{Y}_{g}^1 - (\psi_g - E[\psi_g])' \beta_2^P) ~.
\end{align*}
It follows from previous arguments that
\begin{align*}
    & \sqrt G(\hat \theta_2^{P, adj} - \theta_2^{P}) - \sqrt G(\tilde \theta_2^{P, adj} - \theta_2^{P}) \\
    & = \sqrt G (\bar \psi_G - E[\psi_g])' \beta^\ast \left ( \frac{1}{\frac{1}{G} \sum_{1 \leq g \leq 2G} N_g \tilde H_g} - \frac{1}{\frac{1}{G} \sum_{1 \leq g \leq 2G} N_g (1 - \tilde H_g)} \right ) + o_P(1) \\
    & = o_P(1)~.
\end{align*}
It then follows from the proof of Theorem \ref{thm:matched-group} that $\sqrt G(\hat \theta^{P,adj}_2 - \theta^P_2) \stackrel{d}{\to} N(0, V_2^*(1))$, where
\begin{align*}
    V_2^*(1) &= \frac{1}{\pi_1} \var[ Y_g^*(1,\pi_2)] + \frac{1}{1-\pi_1} \var[ Y_g^*(0,0)]\\
    &\hspace{3em} - \pi_1(1-\pi_1) E\left[\left(\frac{1}{\pi_1}E[ Y_g^*(1,\pi_2)\mid S_g] + \frac{1}{1-\pi_1}E[ Y_g^*(0,0)\mid S_g] \right)^2\right]
\end{align*}
All relevant assumptions for Theorem \ref{thm:matched-group} have their counterparts stated in Theorem \ref{thm:adj}.

\subsection{Variance Improvement}\label{subsec:ca-var-improve}
Recall that
\begin{align*}
    V_2^*(z) &= \frac{1}{\pi_1} \var[Y_g^*(z,\pi_2)] + \frac{1}{1-\pi_1} \var[Y_g^*(0,0)]\\
    &\hspace{3em} - \pi_1(1-\pi_1) E\left[\left(\frac{1}{\pi_1}E[ Y_g^*(z,\pi_2)\mid S_g] + \frac{1}{1-\pi_1}E[ Y_g^*(0,0)\mid S_g] \right)^2\right] \\
    &= \frac{1}{\pi_1} \var[E[ Y_g^*(z,\pi_2)\mid S_g]] + \frac{1}{\pi_1} E[\var[ Y_g^*(z,\pi_2)\mid S_g]] \\
    &\hspace{3em} + \frac{1}{1-\pi_1} \var[E[ Y_g^*(0,0)\mid S_g]] + \frac{1}{1-\pi_1} E[\var[ Y_g^*(0,0)\mid S_g]] \\
    &\hspace{6em} - \frac{1-\pi_1}{\pi_1} E\left[E[ Y_g^*(z,\pi_2)\mid S_g]^2 \right] - 2 E\left[E[ Y_g^*(z,\pi_2)\mid S_g] E[ Y_g^*(0,0)\mid S_g] \right] \\
    &\hspace{9em} - \frac{\pi_1}{1-\pi_1} E\left[E[ Y_g^*(0,0)\mid S_g]^2 \right] \\
    &= \frac{1}{\pi_1} E[\var[Y_g^*(z,\pi_2)\mid S_g]] + \frac{1}{1-\pi_1} E[\var[Y_g^*(0,0)\mid S_g]] + E[E[Y_g^*(z,\pi_2) - Y_g^*(0,0)\mid S_g]^2] ~.
\end{align*}
My goal is to show that $V_2^*(1) \leq V_2(1)$. First note that by definition it follows immediately that 
\begin{equation*}
    E[E[\tilde Y_g(z,\pi_2) - \tilde Y_g(0,0)\mid S_g]^2] = E[E[Y_g^*(1,\pi_2) - Y_g^*(0,0)\mid S_g]^2] ~.
\end{equation*}
It thus remains to show that 
\begin{equation*}
    \frac{1}{\pi_1} E[\var[Y_g^*(1,\pi_2)\mid S_g]] + \frac{1}{1-\pi_1} E[\var[Y_g^*(0,0)\mid S_g]] \leq \frac{1}{\pi_1} \var[E[\tilde Y_g(1,\pi_2)\mid S_g]] + \frac{1}{1-\pi_1} \var[E[\tilde Y_g(0,0)\mid S_g]]~.
\end{equation*}
To that end,
\begin{align*}
    &\frac{1}{\pi_1} E[\var[Y_g^*(1,\pi_2)\mid S_g]] + \frac{1}{1-\pi_1} E[\var[Y_g^*(0,0)\mid S_g]] \\
    &=  \frac{1}{\pi_1} E\left[\var\left[\tilde Y_g(1,\pi_2) - \frac{(\psi_g - E[\psi_g])' \beta^P_2}{E[N_g]} \mid S_g\right]\right]  + \frac{1}{1-\pi_1} E\left[\var\left[\tilde Y_g(0,0) - \frac{(\psi_g - E[\psi_g])' \beta^P_2}{E[N_g]} \mid S_g\right]\right] \\
    &= \frac{1}{\pi_1} E\left[\var\left[\tilde Y_g(1,\pi_2) \mid S_g\right]\right] + \frac{1}{1-\pi_1} E\left[\var\left[\tilde Y_g(0,0) \mid S_g\right]\right] \\
    &\hspace{3em} - 2 E\left[  \cov \left[\frac{1}{\pi_1} \tilde Y_g(1,\pi_2) + \frac{1}{1-\pi_1} \tilde Y_g(0,0),  \frac{(\psi_g - E[\psi_g])' \beta^P_2}{E[N_g]}\mid S_g\right] \right] \\
    &\hspace{6em} + \frac{1}{\pi_1(1-\pi_1)} E\left[ \var\left[\frac{(\psi_g - E[\psi_g])' \beta^P_2}{E[N_g]} \mid S_g\right] \right] \\
    % &= \frac{1}{\pi_1} E\left[\var\left[\tilde Y_g(1,\pi_2) \mid S_g\right]\right] + \frac{1}{1-\pi_1} E\left[\var\left[\tilde Y_g(0,0) \mid S_g\right]\right] \\
    % &\hspace{3em} - \frac{2}{E[N_g]^2} E\left[  \cov \left[\frac{1}{\pi_1} \bar Y_g(1, \pi_2)  N_g + \frac{1}{1-\pi_1}  \bar Y_g(0,0) N_g,  \psi_g' \beta^P_2\mid S_g\right] \right] \\
    % &\hspace{6em} + \frac{2}{E[N_g]^3} E\left[  \cov \left[N_g,  \psi_g' \beta^P_2\mid S_g\right] \right] E\left[\frac{1}{\pi_1}\bar Y_g(1,\pi_2) N_g + \frac{1}{1-\pi_1} \bar Y_g(0,0) N_g\right] \\
    % &\hspace{9em} + \frac{1}{\pi_1(1-\pi_1)} \frac{1}{E[N_g]^2} E\left[ \var[ \psi_g' \beta^P_2\mid S_g] \right] \\
    &= \frac{1}{\pi_1} E\left[\var\left[\tilde Y_g(1,\pi_2) \mid S_g\right]\right] + \frac{1}{1-\pi_1} E\left[\var\left[\tilde Y_g(0,0) \mid S_g\right]\right] \\
    &\hspace{3em} - \frac{1}{\pi_1(1-\pi_1)} \frac{1}{E[N_g]^2} E\left[ \var[ \psi_g' \beta^P_2\mid S_g] \right] ~.\\
    %&\hspace{6em} + \frac{2}{E[N_g]^3} E\left[  \cov \left[N_g,  \psi_g' \beta^P_2\mid S_g\right] \right] E\left[\frac{1}{\pi_1}\bar Y_g(1,\pi_2) N_g + \frac{1}{1-\pi_1} \bar Y_g(0,0) N_g\right]
\end{align*}
The last inequality follows by noting that $\beta^P_2$ is the projection coefficient of $\frac{1}{\pi_1} \tilde Y_g(1, \pi_2)  + \frac{1}{1-\pi_1}  \tilde Y_g(0,0) - E[\frac{1}{\pi_1} \tilde Y_g(1, \pi_2)   + \frac{1}{1-\pi_1}  \tilde Y_g(0,0)  \mid S_g]$ on $(\psi_g - E[\psi_g \mid S_g])/E[N_g]$, 
\begin{equation*}
    E\left[  \cov \left[\frac{1}{\pi_1} \tilde Y_g(1, \pi_2)   + \frac{1}{1-\pi_1}  \tilde Y_g(0,0),  \psi_g' \beta^P_2\mid S_g\right] \right] =  \frac{1}{\pi_1(1-\pi_1) E[N_g]}  E\left[ \var[ \psi_g' \beta^P_2\mid S_g] \right]~.
\end{equation*}
%Note that $E[  \cov [N_g,  \psi_g' \beta^P_2\mid S_g] ] = 0$ when $S_g = (N_g, X_g)$. In this case, the $V_2^*(1) \leq V_2(1)$. However, Assumption \ref{as:Q_G-lip}, \ref{as:close}-\ref{ass:variance-estimator} and \ref{ass:psi} need to be modified as done in \cite{mp-cluster}, as there is a subtle technical difference between matching on cluster size and not.
Therefore, 
\begin{equation*}
    V_2^*(1) = V_2(1) - \kappa^2 ~,
\end{equation*}
where
\[
\kappa^2 = \frac{1}{\pi_1(1-\pi_1)} \frac{1}{E[N_g]^2} E\left[ \var[ \psi_g' \beta^P_2\mid S_g] \right]~.
\]
\newpage

\section{Details for Weighted OLS}\label{sec:app-wols}
In this section, let's consider estimator of the coefficient of $Z_{i,g}$ and $L_{i,g}$ in a weighted least squares regression of $Y_{i,g}$ on a constant and $Z_{i,g}$ and $L_{i,g}$ with weights equal to $\sqrt{N_g/M_g}$. The results for weights equal to $\sqrt{1/M_g}$ (or the unweighted regression) are similar and omitted here.
First, I provide some notatiosn as follows:
\begin{equation*}
    \begin{aligned}
T_{i,g} & :=\left(\sqrt{\frac{N_g}{\left|\mathcal{M}_g\right|}} \quad  \sqrt{\frac{N_g}{\left|\mathcal{M}_g\right|}} Z_{i,g} \quad \sqrt{\frac{N_g}{\left|\mathcal{M}_g\right|}} L_{i,g} \right)^\prime \\
T_{g} & := \left(T_{i,g} : i\in \mathcal{M}_g \right)^{\prime} \\
\hat{\epsilon}_{g} & :=\left(Y_{i,g} - \hat \alpha - \hat \beta_1 Z_{i,g} - \hat \beta_2 L_{i,g}  : i\in \mathcal{M}_g \right)^{\prime} ~,
\end{aligned}
\end{equation*}
where $\hat \alpha, \hat \beta_1$ and $\hat \beta_2$ are the corresponding estimated coefficients.
By doing some algebra, it follows that
{\scriptsize
\begin{equation*}
    \sum_{1 \leq g \leq G} \sum_{i\in\mathcal{M}_g} T_{i,g} T_{i,g}^{\prime}=\left(\begin{array}{ccc}
\sum_{1 \leq g \leq G}  N_g & \sum_{1 \leq g \leq G} N_g \pi_2 I\{H_g=\pi_2\} &  \sum_{1 \leq g \leq G} N_g (1-\pi_2)I\{H_g=\pi_2\}\\
\sum_{1 \leq g \leq G} N_g \pi_2 I\{H_g=\pi_2\} & \sum_{1 \leq g \leq G} N_g\pi_2  I\{H_g=\pi_2\}  & 0 \\
\sum_{1 \leq g \leq G} N_g (1-\pi_2) I\{H_g=\pi_2\} & 0 & \sum_{1 \leq g \leq G} N_g (1-\pi_2) I\{H_g=\pi_2\}  
\end{array}\right)
\end{equation*}}
and
{\footnotesize
\begin{align*}
    \sum_{1 \leq g \leq G} \sum_{i\in\mathcal{M}_g} T_{i,g} \sqrt{\frac{N_g}{\left|\mathcal{M}_g\right|}} Y_{i,g} &= \left( \sum_{1 \leq g \leq G} \frac{N_g}{M_g} \sum_{i\in\mathcal{M}_g} Y_{i,g}  \quad   \sum_{1 \leq g \leq G}  \frac{N_g}{M_g}\sum_{i\in\mathcal{M}_g} Y_{i,g} Z_{i,g} \quad \sum_{1 \leq g \leq G} \frac{N_g}{M_g}  \sum_{i\in\mathcal{M}_g} Y_{i,g} L_{i,g} \right)^\prime \\
    &=\left( \sum_{1 \leq g \leq G} \frac{N_g}{M_g} \sum_{i\in\mathcal{M}_g} Y_{i,g}  \quad   \sum_{1 \leq g \leq G}  I\{H_g=\pi_2\} N_g  \bar Y_g^1 \pi_2 \quad \sum_{1 \leq g \leq G}  I\{H_g=\pi_2\} N_g  \bar Y_g^0 (1-\pi_2) \right)^\prime
\end{align*}}
Note that
\begin{equation*}
    \left(\sum_{1 \leq g \leq G} \sum_{i\in\mathcal{M}_g} T_{i,g} T_{i,g}^{\prime}\right)^{-1}=\left(\begin{array}{ccc}
\frac{1}{N_C} & -\frac{1}{N_C} &  -\frac{1}{N_C}\\
-\frac{1}{N_C} & \frac{1}{N_C} + \frac{1}{N_T \pi_2}  & \frac{1}{N_C} \\
-\frac{1}{N_C} & \frac{1}{N_C} & \frac{1}{N_C} + \frac{1}{N_T(1-\pi_2)}
\end{array}\right)
\end{equation*}
Then, it follows that
\begin{equation*}
    \left(\begin{array}{c}
\hat \alpha  \\
\hat \theta_2^P \\
\hat \theta_2^S 
\end{array}\right) = \left(\sum_{1 \leq g \leq G} \sum_{i\in\mathcal{M}_g} T_{i,g} T_{i,g}^{\prime}\right)^{-1} \left(\sum_{1 \leq g \leq G} \sum_{i\in\mathcal{M}_g} T_{i,g} \sqrt{\frac{N_g}{\left|\mathcal{M}_g\right|}} Y_{i,g} \right) = \left(\frac{1}{N_C} \sum_{1\leq g\leq N_g} I\{H_g=0\} N_g \bar Y_g^1 \quad \hat \theta^P_2 \quad   \hat \theta^S_2\right)^\prime~.
\end{equation*}
Therefore, we conclude that this weighted OLS regression results in the same estimators as $\hat\theta^P_2,   \hat \theta^S_2$. Next, I consider $t$-tests based on cluster-robust variance estimator.
Note that the cluster-robust variance estimator can be written as
\begin{equation*}
    \hat{\mathbf{V}}_{\textsc{CR}} = G\left(\sum_{1 \leq g \leq G}  T_{g}^\prime T_{g} \right)^{-1}\left(\sum_{1 \leq g \leq G}  T_{g}^\prime \hat{\epsilon}_{i,g} \hat{\epsilon}_{i,g}^{\prime} T_{i,g}\right)\left(\sum_{1 \leq g \leq G}  T_{g}^{\prime} T_{g}\right)^{-1}~,
\end{equation*}
where $\sum_{1 \leq g \leq G}  T_{g}^\prime T_{g}$ should be identical to $\sum_{1 \leq g \leq G} \sum_{i\in\mathcal{M}_g} T_{i,g} T_{i,g}^{\prime}$. By doing some algebra, if follows that
\begin{align*}
    &\sum_{1 \leq g \leq G}  T_{g}^\prime \hat{\epsilon}_{i,g} \hat{\epsilon}_{i,g}^{\prime} T_{i,g} = \sum_{1\leq g \leq G} \left(\frac{N_g}{M_g} \right)^2 \left(\begin{array}{c}
\sum_{i\in\mathcal{M}_g} \hat \epsilon_{i,g}  \\
\sum_{i\in\mathcal{M}_g} \hat \epsilon_{i,g} Z_{i,g} \\
\sum_{i\in\mathcal{M}_g} \hat \epsilon_{i,g} L_{i,g}  
\end{array}\right)
\left(\begin{array}{c}
\sum_{i\in\mathcal{M}_g} \hat \epsilon_{i,g}  \\
\sum_{i\in\mathcal{M}_g} \hat \epsilon_{i,g} Z_{i,g} \\
\sum_{i\in\mathcal{M}_g} \hat \epsilon_{i,g} L_{i,g}  
\end{array}\right)^{\prime}~.
\end{align*}
And thus cluster-robust variance estimator can be written as $\sum_{1\leq g \leq G} \tilde \epsilon_g \tilde \epsilon_g^\prime$, where
\begin{equation*}
    \tilde \epsilon_g = \left(\begin{array}{c}
\frac{1}{N_C} \frac{1}{M_g}  \sum_{i\in\mathcal{M}_g}  \hat \epsilon_{i,g}  N_g  I\{H_g=0\}  \\
\frac{1}{N_T } \frac{1}{M_g^1}\sum_{i\in\mathcal{M}_g} \hat \epsilon_{i,g}  N_g Z_{i,g} - \frac{1}{N_C} \frac{1}{M_g}  \sum_{i\in\mathcal{M}_g}  \hat \epsilon_{i,g}  N_g  I\{H_g=0\}  \\
\frac{1}{N_T } \frac{1}{M_g^0}\sum_{i\in\mathcal{M}_g} \hat \epsilon_{i,g}  N_g L_{i,g}  - \frac{1}{N_C} \frac{1}{M_g}  \sum_{i\in\mathcal{M}_g}  \hat \epsilon_{i,g}  N_g  I\{H_g=0\} 
\end{array}\right)~.
\end{equation*}
Take the second diagonal element (primary effect) as an example. Its cluster-robust variance estimator is given by
\begin{align*}
    &\hat V_{\textsc{CR}}(1)=G \sum_{1\leq g \leq G}\left(\frac{1}{N_T } \frac{1}{M_g^1}\sum_{i\in\mathcal{M}_g} \hat \epsilon_{i,g}  N_g Z_{i,g} - \frac{1}{N_C} \frac{1}{M_g}  \sum_{i\in\mathcal{M}_g}  \hat \epsilon_{i,g}  N_g  I\{H_g=0\} \right)^2\\
    &=  \frac{1}{(N_T/G)^2} \frac{1}{G}\sum_{1\leq g\leq G} N_g^2 \left(\bar Y_g(1,\pi_2) - \hat \alpha - \hat \theta_2^P\right)^2 I\{H_g=\pi_2\} + \frac{1}{(N_C/G)^2} \frac{1}{G}\sum_{1\leq g\leq G} N_g^2 \left(\bar Y_g(0,0) - \hat \alpha\right)^2 I\{H_g=0\}~.
\end{align*}
In both finely stratified randomization and ``large strata'' frameworks, by repeating arguments made in the Section \ref{sec:proof-variance-ca} and \ref{sec:proof-variance-mp}, we have the following asymptotic results:
\begin{align*}
    \frac{1}{G}\sum_{1\leq g\leq G} N_g^2 \left(\bar Y_g(1,\pi_2) - \hat \alpha - \hat \theta_2^P\right)^2 I\{H_g=\pi_2\} &\xrightarrow{p} \pi_1 E\left[\left(N_g \bar Y_g(1,\pi_2) - N_g \frac{E[N_g \bar Y_g(1,\pi_2)]}{E[N_g]} \right)^2\right]  \\
    \frac{1}{G}\sum_{1\leq g\leq G} N_g^2 \left(\bar Y_g(0,0) - \hat \alpha\right)^2 &\xrightarrow{p} (1-\pi_1) E\left[\left(N_g \bar Y_g(0,0) - N_g \frac{E[N_g \bar Y_g(0,0)]}{E[N_g]} \right)^2\right],
\end{align*}
which implies
\begin{equation*}
    \hat V_{\textsc{CR}}(1) \xrightarrow{p} \frac{1}{\pi_1}\var[\tilde Y_g(z,\pi_2)] + \frac{1}{1-\pi_1} \var[\tilde Y_g(0,0)]~.
\end{equation*}

\newpage
\section{Additional Simulations Results}
\subsection{Subsampling within Clusters with $M_g < N_g$}
In this section, I repeat the simulation study from Section \ref{sec:simulation}, with the only difference being that \( M_g \) is set to \( 0.5 N_g \) and \( \mathcal{M}_g \) is a random subset of \( \{1, \dots, N_g\} \). In Table \ref{table:mse_ratio_appendix}, I present ratios of MSE for various two-stage designs under the model with \( M_g = 0.5 N_g \) against those with \( M_g = N_g \). Note that under \( M_g = 0.5 N_g \), the asymptotic variance of the four estimators is higher than those under \( M_g = N_g \) in most cases, which is likely due to the effect of smaller sample size. Tables \ref{table:mse_appendix} through \ref{table:reject-probs-ols-appendix} correspond to Tables \ref{table:mse} through \ref{table:reject-probs-ols} in the main text. The conclusions from the main text still hold qualitatively under \( M_g = 0.5 N_g \).

\begin{table}[ht!]
\centering
\setlength{\tabcolsep}{5pt}
\begin{adjustbox}{max width=0.75\linewidth,center}
\begin{tabular}{lllllllll}
\toprule
 & & \multicolumn{7}{c}{Second-stage}   \\\cmidrule{3-9}
First-stage              & Parameter & \textbf{C}    & \textbf{S-2}   & \textbf{S-4}   & \textbf{S-4O} & \textbf{MT-A}  & \textbf{MT-B}  & \textbf{MT-C}   \\
\midrule
\multirow{4}{*}{\textbf{C}} & $\theta^P_1$ & 1.1311 & 0.9961 & 1.1490 & 1.0975 & 1.0282 & 0.9779 & 0.8973 \\
& $\theta^P_2$ & 1.1370 & 0.9958 & 1.1744 & 1.1061 & 1.1038 & 1.0174 & 0.9162 \\
& $\theta^S_1$ & 1.1264 & 1.0002 & 1.1370 & 1.1270 & 1.0356 & 0.9930 & 0.9069 \\
& $\theta^S_2$ & 1.1443 & 0.9915 & 1.1679 & 1.1308 & 1.1145 & 1.0316 & 0.9232 \\
& & & & & & & & \\
\multirow{4}{*}{\textbf{S-2}} & $\theta^P_1$ & 1.0831 & 1.1022 & 0.9960 & 0.9507 & 0.8446 & 1.0029 & 0.9950 \\
& $\theta^P_2$ & 1.1140 & 1.0724 & 1.0202 & 0.9855 & 0.8668 & 1.0452 & 1.0288 \\
& $\theta^S_1$ & 1.0746 & 1.1178 & 0.9895 & 0.9503 & 0.8583 & 0.9826 & 0.9776 \\
& $\theta^S_2$ & 1.0837 & 1.0888 & 1.0032 & 0.9979 & 0.8820 & 1.0361 & 1.0089 \\
& & & & & & & & \\
\multirow{4}{*}{\textbf{S-4}} & $\theta^P_1$ & 1.1002 & 1.0298 & 1.1288 & 1.0669 & 0.9949 & 1.0560 & 1.0894 \\
& $\theta^P_2$ & 1.0792 & 1.0073 & 1.1446 & 1.1216 & 1.0073 & 1.0599 & 1.1048 \\
& $\theta^S_1$ & 1.1002 & 1.0257 & 1.0989 & 1.0196 & 0.9770 & 1.0634 & 1.0656 \\
& $\theta^S_2$ & 1.0500 & 1.0148 & 1.1144 & 1.0748 & 0.9861 & 1.0549 & 1.0916 \\
& & & & & & & & \\
\multirow{4}{*}{\textbf{S-O}} & $\theta^P_1$ & 1.1482 & 1.2698 & 1.1789 & 1.0416 & 1.0616 & 1.2459 & 1.0699 \\
& $\theta^P_2$ & 1.1516 & 1.2765 & 1.1903 & 1.0977 & 1.0537 & 1.2524 & 1.0921 \\
& $\theta^S_1$ & 1.1682 & 1.2737 & 1.1562 & 1.0646 & 1.1392 & 1.1798 & 1.1181 \\
& $\theta^S_2$ & 1.1832 & 1.2388 & 1.1394 & 1.0780 & 1.1380 & 1.1824 & 1.1050 \\
& & & & & & & & \\
\multirow{4}{*}{\textbf{MT-A}} & $\theta^P_1$ & 1.1515 & 1.0560 & 0.9775 & 1.0974 & 1.0252 & 1.0592 & 0.9909 \\
& $\theta^P_2$ & 1.1164 & 1.0489 & 1.0471 & 1.0568 & 1.0346 & 1.0512 & 0.9461 \\
& $\theta^S_1$ & 1.1254 & 1.0116 & 0.9951 & 1.1025 & 1.0496 & 1.0566 & 1.0092 \\
& $\theta^S_2$ & 1.0978 & 0.9980 & 1.0441 & 1.0523 & 1.0717 & 1.0294 & 0.9647 \\
& & & & & & & & \\
\multirow{4}{*}{\textbf{MT-B}} & $\theta^P_1$ & 1.2961 & 0.9776 & 1.3245 & 1.1634 & 1.0426 & 1.3105 & 1.0905 \\
& $\theta^P_2$ & 1.2771 & 0.9580 & 1.2263 & 1.1270 & 0.9956 & 1.2462 & 1.1027 \\
& $\theta^S_1$ & 1.2850 & 0.9764 & 1.2778 & 1.2184 & 1.0650 & 1.1869 & 1.0872 \\
& $\theta^S_2$ & 1.2499 & 0.9526 & 1.2093 & 1.1848 & 1.0172 & 1.1573 & 1.0948 \\
& & & & & & & & \\
\multirow{4}{*}{\textbf{MT-C}} & $\theta^P_1$ & 1.3489 & 1.2758 & 1.3100 & 1.3483 & 1.3233 & 1.4600 & 1.3201 \\
& $\theta^P_2$ & 1.4388 & 1.4396 & 1.5821 & 1.4858 & 1.5985 & 1.6598 & 1.5779 \\
& $\theta^S_1$ & 1.3891 & 1.2927 & 1.2882 & 1.2148 & 1.2832 & 1.3590 & 1.3978 \\
& $\theta^S_2$ & 1.5025 & 1.4200 & 1.4729 & 1.2878 & 1.4961 & 1.5932 & 1.6898 \\
\bottomrule
\end{tabular}
\end{adjustbox}
\caption{Ratio of MSE under $M_g=0.5 N_g$ against those under $M_g = N_g$}
\label{table:mse_ratio_appendix}
\end{table}

\begin{table}[ht!]
\centering
\setlength{\tabcolsep}{5pt}
\begin{adjustbox}{max width=0.75\linewidth,center}
\begin{tabular}{lllllllll}
\toprule
 & & \multicolumn{7}{c}{Second-stage}   \\\cmidrule{3-9}
First-stage              & Parameter & \textbf{C}    & \textbf{S-2}   & \textbf{S-4}   & \textbf{S-4O} & \textbf{MT-A}  & \textbf{MT-B}  & \textbf{MT-C}   \\
\midrule
\multirow{4}{*}{\textbf{C}} & $\theta^P_1$ & \textbf{1.0000} & 1.0678 & 1.1776 & 1.0368 & 1.0867 & 1.1143 & 1.0588 \\
& $\theta^P_2$ & \textbf{1.0000} & 1.0467 & 1.1832 & 1.0618 & 1.0980 & 1.1078 & 1.0684 \\
& $\theta^S_1$ & \textbf{1.0000} & 1.1005 & 1.1998 & 1.0657 & 1.0648 & 1.1496 & 1.0530 \\
& $\theta^S_2$ & \textbf{1.0000} & 1.0565 & 1.2064 & 1.0793 & 1.0570 & 1.1355 & 1.0432 \\
& & & & & & & & \\
\multirow{4}{*}{\textbf{S-2}} & $\theta^P_1$ & 0.9625 & 0.9634 & 0.9359 & 0.9656 & \textbf{0.8380} & 0.9332 & 0.8407 \\
& $\theta^P_2$ & 0.9384 & 0.9214 & 0.8862 & 0.9295 & 0.8447 & 0.8765 & \textbf{0.7875} \\
& $\theta^S_1$ & 0.9529 & 0.9529 & 0.9028 & 0.9575 & 0.8529 & 0.9512 & \textbf{0.8495} \\
& $\theta^S_2$ & 0.9170 & 0.9001 & 0.8463 & 0.9166 & 0.8504 & 0.8916 & \textbf{0.7903} \\
& & & & & & & & \\
\multirow{4}{*}{\textbf{S-4}} & $\theta^P_1$ & 0.9350 & 0.9040 & 0.9774 & 0.9036 & \textbf{0.8641} & 0.9486 & 0.8762 \\
& $\theta^P_2$ & 0.8953 & 0.8650 & 0.9545 & 0.9040 & \textbf{0.7998} & 0.9174 & 0.8209 \\
& $\theta^S_1$ & 0.9111 & 0.8740 & 0.9517 & 0.9139 & 0.8768 & 0.9527 & \textbf{0.8622} \\
& $\theta^S_2$ & 0.8601 & 0.8213 & 0.9075 & 0.8867 & \textbf{0.7975} & 0.8992 & 0.8001 \\
& & & & & & & & \\
\multirow{4}{*}{\textbf{S-O}} & $\theta^P_1$ & 0.2726 & 0.2820 & 0.2629 & 0.2740 & \textbf{0.2527} & 0.2792 & 0.2556 \\
& $\theta^P_2$ & 0.3191 & 0.3200 & 0.3038 & 0.3283 & \textbf{0.2967} & 0.3148 & 0.2970 \\
& $\theta^S_1$ & 0.2820 & 0.2767 & 0.2716 & 0.2758 & 0.2556 & 0.2878 & \textbf{0.2447} \\
& $\theta^S_2$ & 0.3260 & 0.3148 & 0.3101 & 0.3194 & 0.3010 & 0.3154 & \textbf{0.2759} \\
& & & & & & & & \\
\multirow{4}{*}{\textbf{MT-A}} & $\theta^P_1$ & 0.8306 & 0.8604 & 0.8955 & 0.9684 & 0.8905 & 0.8490 & \textbf{0.8269} \\
& $\theta^P_2$ & \textbf{0.7818} & 0.8185 & 0.8575 & 0.9026 & 0.8397 & 0.8402 & 0.7949 \\
& $\theta^S_1$ & \textbf{0.8242} & 0.8604 & 0.8769 & 0.9775 & 0.8664 & 0.8412 & 0.8310 \\
& $\theta^S_2$ & \textbf{0.7565} & 0.8059 & 0.8122 & 0.9044 & 0.8020 & 0.8230 & 0.7924 \\
& & & & & & & & \\
\multirow{4}{*}{\textbf{MT-B}} & $\theta^P_1$ & 0.3759 & 0.3755 & 0.3637 & 0.3696 & 0.3734 & 0.4006 & \textbf{0.3498} \\
& $\theta^P_2$ & 0.5068 & 0.5096 & 0.4988 & 0.4958 & 0.5002 & 0.5395 & \textbf{0.4780} \\
& $\theta^S_1$ & 0.3601 & 0.3771 & 0.3663 & 0.3626 & 0.3559 & 0.3742 & \textbf{0.3461} \\
& $\theta^S_2$ & 0.4734 & 0.5053 & 0.4884 & 0.4820 & \textbf{0.4663} & 0.5057 & 0.4725 \\
& & & & & & & & \\
\multirow{4}{*}{\textbf{MT-C}} & $\theta^P_1$ & 0.1662 & 0.1683 & 0.1556 & 0.1718 & \textbf{0.1521} & 0.1753 & 0.1696 \\
& $\theta^P_2$ & 0.1446 & 0.1450 & \textbf{0.1346} & 0.1547 & 0.1353 & 0.1596 & 0.1468 \\
& $\theta^S_1$ & 0.1694 & 0.1679 & 0.1602 & 0.1653 & \textbf{0.1451} & 0.1979 & 0.1762 \\
& $\theta^S_2$ & 0.1469 & 0.1445 & 0.1397 & 0.1425 & \textbf{0.1274} & 0.1691 & 0.1544 \\
\bottomrule
\end{tabular}
\end{adjustbox}
\caption{Ratio of MSE under all designs against those under complete randomization in both stages}
\label{table:mse_appendix}
\end{table}

\begin{table}[t]
\centering
\setlength{\tabcolsep}{3pt}
\begin{adjustbox}{max width=0.95\linewidth,center}
\begin{tabular}{lllllllllllllll}
\toprule
       &              & \multicolumn{13}{c}{Second-stage}    \\  \cmidrule{3-15}
       &              & \multicolumn{6}{c}{$H_0: \tau = \omega = 0$}                      &  & \multicolumn{6}{c}{$H_1: \tau = \omega = 0.05$}                      \\  \cmidrule{3-8} \cmidrule{10-15}
First-stage & Parameter &  \textbf{S-2}   & \textbf{S-4}   & \textbf{S-4O} & \textbf{MT-A}  & \textbf{MT-B}  & \textbf{MT-C} & & \textbf{S-2}   & \textbf{S-4}   & \textbf{S-4O} & \textbf{MT-A}  & \textbf{MT-B}  & \textbf{MT-C} \\
\toprule
\multirow{4}{*}{\textbf{S-2}}  & $\theta^P_1$ & 0.058 & 0.066 & 0.059 & 0.054 & 0.055 & 0.047 &  & 0.213 & 0.248 & 0.243 & 0.231 & 0.221 & 0.229 \\
                               & $\theta^P_2$ & 0.051 & 0.059 & 0.056 & 0.065 & 0.048 & 0.050 &  & 0.229 & 0.237 & 0.241 & 0.227 & 0.223 & 0.223 \\
                               & $\theta^S_1$ & 0.055 & 0.062 & 0.050 & 0.059 & 0.048 & 0.047 &  & 0.081 & 0.100 & 0.100 & 0.104 & 0.096 & 0.102 \\
                               & $\theta^S_2$ & 0.051 & 0.064 & 0.055 & 0.053 & 0.044 & 0.047 &  & 0.086 & 0.088 & 0.100 & 0.102 & 0.102 & 0.096 \\
\multicolumn{1}{l}{}           &              &       &       &       &       &       &       &  &       &       &       &       &       &       \\
\multirow{4}{*}{\textbf{S-4}}  & $\theta^P_1$ & 0.062 & 0.054 & 0.059 & 0.061 & 0.040 & 0.053 &  & 0.270 & 0.242 & 0.243 & 0.250 & 0.276 & 0.233 \\
                               & $\theta^P_2$ & 0.058 & 0.058 & 0.058 & 0.060 & 0.046 & 0.051 &  & 0.246 & 0.234 & 0.233 & 0.232 & 0.260 & 0.226 \\
                               & $\theta^S_1$ & 0.060 & 0.050 & 0.057 & 0.063 & 0.045 & 0.055 &  & 0.108 & 0.071 & 0.094 & 0.099 & 0.109 & 0.104 \\
                               & $\theta^S_2$ & 0.056 & 0.056 & 0.056 & 0.060 & 0.045 & 0.056 &  & 0.113 & 0.083 & 0.097 & 0.091 & 0.111 & 0.103 \\
\multicolumn{1}{l}{}           &              &       &       &       &       &       &       &  &       &       &       &       &       &       \\
\multirow{4}{*}{\textbf{S-O}}  & $\theta^P_1$ & 0.054 & 0.044 & 0.054 & 0.053 & 0.057 & 0.055 &  & 0.646 & 0.623 & 0.647 & 0.614 & 0.584 & 0.657 \\
                               & $\theta^P_2$ & 0.043 & 0.057 & 0.052 & 0.049 & 0.056 & 0.048 &  & 0.557 & 0.536 & 0.554 & 0.533 & 0.507 & 0.562 \\
                               & $\theta^S_1$ & 0.054 & 0.047 & 0.057 & 0.059 & 0.069 & 0.055 &  & 0.215 & 0.207 & 0.203 & 0.205 & 0.202 & 0.216 \\
                               & $\theta^S_2$ & 0.053 & 0.052 & 0.061 & 0.052 & 0.057 & 0.053 &  & 0.201 & 0.174 & 0.181 & 0.173 & 0.178 & 0.183 \\
\multicolumn{1}{l}{}           &              &       &       &       &       &       &       &  &       &       &       &       &       &       \\
\multirow{4}{*}{\textbf{MT-A}} & $\theta^P_1$ & 0.057 & 0.054 & 0.051 & 0.044 & 0.065 & 0.055 &  & 0.260 & 0.242 & 0.254 & 0.257 & 0.224 & 0.256 \\
                               & $\theta^P_2$ & 0.053 & 0.061 & 0.048 & 0.043 & 0.054 & 0.050 &  & 0.261 & 0.236 & 0.232 & 0.250 & 0.229 & 0.236 \\
                               & $\theta^S_1$ & 0.060 & 0.052 & 0.047 & 0.045 & 0.059 & 0.052 &  & 0.104 & 0.088 & 0.109 & 0.101 & 0.096 & 0.088 \\
                               & $\theta^S_2$ & 0.051 & 0.051 & 0.044 & 0.043 & 0.055 & 0.047 &  & 0.111 & 0.086 & 0.102 & 0.097 & 0.092 & 0.090 \\
\multicolumn{1}{l}{}           &              &       &       &       &       &       &       &  &       &       &       &       &       &       \\
\multirow{4}{*}{\textbf{MT-B}} & $\theta^P_1$ & 0.040 & 0.057 & 0.055 & 0.040 & 0.044 & 0.044 &  & 0.492 & 0.516 & 0.536 & 0.507 & 0.496 & 0.504 \\
                               & $\theta^P_2$ & 0.040 & 0.059 & 0.065 & 0.048 & 0.044 & 0.054 &  & 0.382 & 0.396 & 0.403 & 0.381 & 0.358 & 0.368 \\
                               & $\theta^S_1$ & 0.045 & 0.053 & 0.056 & 0.062 & 0.051 & 0.063 &  & 0.170 & 0.199 & 0.193 & 0.186 & 0.153 & 0.164 \\
                               & $\theta^S_2$ & 0.044 & 0.050 & 0.057 & 0.052 & 0.051 & 0.055 &  & 0.140 & 0.143 & 0.146 & 0.150 & 0.109 & 0.134 \\
\multicolumn{1}{l}{}           &              &       &       &       &       &       &       &  &       &       &       &       &       &       \\
\multirow{4}{*}{\textbf{MT-C}} & $\theta^P_1$ & 0.053 & 0.045 & 0.058 & 0.058 & 0.051 & 0.057 &  & 0.833 & 0.837 & 0.823 & 0.831 & 0.802 & 0.853 \\
                               & $\theta^P_2$ & 0.047 & 0.049 & 0.062 & 0.048 & 0.069 & 0.061 &  & 0.843 & 0.848 & 0.859 & 0.853 & 0.814 & 0.869 \\
                               & $\theta^S_1$ & 0.068 & 0.041 & 0.051 & 0.060 & 0.042 & 0.053 &  & 0.312 & 0.379 & 0.289 & 0.334 & 0.308 & 0.339 \\
                               & $\theta^S_2$ & 0.052 & 0.054 & 0.056 & 0.055 & 0.048 & 0.059 &  & 0.325 & 0.356 & 0.295 & 0.337 & 0.296 & 0.330 \\
\bottomrule
\end{tabular}
\end{adjustbox}
\caption{Rejection probabilities under the null and alternative hypothesis}
\label{table:reject-probs-1-appendix}
\end{table}

\begin{table}[t]
\centering
\setlength{\tabcolsep}{3pt}
\begin{adjustbox}{max width=0.95\linewidth,center}
\begin{tabular}{lclllllll}
\toprule
\multirow{2}{*}{Model}  & \multirow{2}{*}{Inference Method}     &  \multirow{2}{*}{Effect}        & \textbf{S-4O} & \textbf{S-4O} & \textbf{S-4O} & \textbf{MT-C} & \textbf{MT-C} & \textbf{MT-C} \\
& &  & \textbf{C} & \textbf{S-4O} & \textbf{MT-C} & \textbf{C} & \textbf{S-4O} & \textbf{MT-C} \\
\toprule
\multirow{9}{*}{\textbf{Homogeneous}}  & OLS robust                         & Primary   & 0.127 & 0.142 & 0.098 & 0.057 & 0.041 & 0.033 \\
& (standard $t$-test)                                            & Spillover   & 0.121 & 0.112 & 0.103 & 0.040 & 0.022 & 0.030 \\
\multicolumn{1}{l}{}                         &                  &          &           &           &          &           &           \\
& OLS cluster                         & Primary   & 0.000 & 0.000 & 0.000 & 0.000 & 0.000 & 0.000 \\
& (clustered $t$-test)                                            & Spillover & 0.000 & 0.000 & 0.000 & 0.000 & 0.000 & 0.000 \\
\multicolumn{1}{l}{}                         &                  &          &           &           &          &           &           \\
& OLS with group                & Primary   & 0.157 & 0.130 & 0.107 & 0.073 & 0.052 & 0.050 \\
& fixed effects (robust)                                             & Spillover & 0.131 & 0.119 & 0.112 & 0.056 & 0.051 & 0.052 \\
\multicolumn{1}{l}{}                         &                  &          &           &           &          &           &           \\
& OLS with group  & Primary   & 0.024 & 0.036 & 0.026 & 0.065 & 0.060 & 0.062 \\ 
& fixed effects (clustered)                                            & Spillover & 0.029 & 0.027 & 0.024 & 0.052 & 0.058 & 0.063 \\
\\ \\
\multirow{9}{*}{\textbf{Heterogeneous}} & OLS robust                          & Primary   & 0.077 & 0.101 & 0.105 & 0.037 & 0.025 & 0.034 \\
& (standard $t$-test)                                            & Spillover & 0.178 & 0.168 & 0.104 & 0.097 & 0.066 & 0.046 \\
\multicolumn{1}{l}{}                         &                  &          &           &           &          &           &           \\
& OLS cluster                         & Primary   & 0.000 & 0.000 & 0.000 & 0.000 & 0.000 & 0.000 \\
& (clustered $t$-test)                                              & Spillover & 0.000 & 0.000 & 0.000 & 0.000 & 0.000 & 0.000 \\
\multicolumn{1}{l}{}                         &                  &          &           &           &          &           &           \\
& OLS with group    &   Primary         & 0.075 & 0.086 & 0.104 & 0.054 & 0.037 & 0.071 \\
& fixed effects (robust)                                             & Spillover & 0.197 & 0.200 & 0.116 & 0.197 & 0.174 & 0.086 \\
\multicolumn{1}{l}{}                         &                  &          &           &           &          &           &           \\
& OLS with group  & Primary   & 0.018 & 0.028 & 0.022 & 0.057 & 0.047 & 0.049 \\
& fixed effects (clustered)                                            & Spillover & 0.015 & 0.026 & 0.031 & 0.046 & 0.031 & 0.057 \\
\bottomrule
\end{tabular}
\end{adjustbox}
\caption{Rejection probabilities of various inference methods under the null hypothesis}
\label{table:reject-probs-ols-appendix}
\end{table}

\clearpage
\newpage
\subsection{Increasing Number of Clusters}
In this section, I repeat the simulation study from Section \ref{sec:simulation}, with the only difference being that the number cluster $G$ increases from $200$ to $400, 800$ and $1000$. 
\begin{table}[ht!]
\centering
\setlength{\tabcolsep}{5pt}
\begin{adjustbox}{max width=0.75\linewidth,center}
\begin{tabular}{lllllllll}
\toprule
 & & \multicolumn{7}{c}{Second-stage}   \\\cmidrule{3-9}
First-stage              & Parameter & \textbf{C}    & \textbf{S-2}   & \textbf{S-4}   & \textbf{S-4O} & \textbf{MT-A}  & \textbf{MT-B}  & \textbf{MT-C}   \\
\midrule
\multirow{4}{*}{\textbf{C}}    & $\theta^P_1$ & 1.0000 & 0.9369 & 0.9364 & 0.9413 & \textbf{0.8969} & 0.9161 & 0.9274 \\
                               & $\theta^P_2$ & 1.0000 & 0.9579 & 0.9334 & 0.9190 & 0.8963 & 0.9093 & \textbf{0.8935} \\
                               & $\theta^S_1$ & 1.0000 & 0.9401 & 0.9564 & 0.9559 & \textbf{0.9160} & 0.9192 & 0.9356 \\
                               & $\theta^S_2$ & 1.0000 & 0.9649 & 0.9511 & 0.9381 & 0.9162 & 0.9194 & \textbf{0.9030} \\
                               &              &        &        &        &        &        &        &        \\
\multirow{4}{*}{\textbf{S-2}}  & $\theta^P_1$ & 0.7594 & 0.7923 & 0.7789 & 0.8150 & 0.7611 & \textbf{0.7587} & 0.8447 \\
                               & $\theta^P_2$ & 0.7306 & 0.7496 & 0.7503 & 0.7582 & 0.7326 & \textbf{0.7133} & 0.8319 \\
                               & $\theta^S_1$ & 0.7754 & 0.7795 & 0.7846 & 0.8272 & 0.7725 & \textbf{0.7530} & 0.8619 \\
                               & $\theta^S_2$ & 0.7537 & 0.7453 & 0.7553 & 0.7647 & 0.7426 & \textbf{0.7158} & 0.8511 \\
                               &              &        &        &        &        &        &        &        \\
\multirow{4}{*}{\textbf{S-4}}  & $\theta^P_1$ & \textbf{0.6778} & 0.7591 & 0.6963 & 0.7435 & 0.7692 & 0.7616 & 0.7514 \\
                               & $\theta^P_2$ & \textbf{0.6387} & 0.7243 & 0.6480 & 0.7018 & 0.7257 & 0.7018 & 0.7071 \\
                               & $\theta^S_1$ & \textbf{0.6821} & 0.7722 & 0.7062 & 0.7506 & 0.7711 & 0.7693 & 0.7524 \\
                               & $\theta^S_2$ & \textbf{0.6441} & 0.7409 & 0.6559 & 0.7098 & 0.7322 & 0.7152 & 0.7104 \\
                               &              &        &        &        &        &        &        &        \\
\multirow{4}{*}{\textbf{S-4O}} & $\theta^P_1$ & 0.2106 & 0.1949 & 0.2081 & 0.1996 & 0.2147 & 0.2033 & \textbf{0.1918} \\
                               & $\theta^P_2$ & 0.2285 & 0.2308 & 0.2330 & 0.2276 & 0.2411 & 0.2314 & \textbf{0.2176} \\
                               & $\theta^S_1$ & 0.2164 & 0.2026 & 0.2101 & 0.2079 & 0.2188 & 0.2056 & \textbf{0.1968} \\
                               & $\theta^S_2$ & 0.2410 & 0.2385 & 0.2354 & 0.2360 & 0.2425 & 0.2388 & \textbf{0.2222} \\
                               &              &        &        &        &        &        &        &        \\
\multirow{4}{*}{\textbf{MT-A}} & $\theta^P_1$ & 0.7258 & \textbf{0.6914} & 0.7389 & 0.7372 & 0.7464 & 0.7797 & 0.7057 \\
                               & $\theta^P_2$ & 0.6794 & 0.6615 & 0.7061 & 0.6737 & 0.7002 & 0.7399 & \textbf{0.6586} \\
                               & $\theta^S_1$ & 0.7542 & \textbf{0.6878} & 0.7500 & 0.7460 & 0.7677 & 0.7956 & 0.7129 \\
                               & $\theta^S_2$ & 0.7136 & \textbf{0.6582} & 0.7185 & 0.6838 & 0.7216 & 0.7566 & 0.6659 \\
                               &              &        &        &        &        &        &        &        \\
\multirow{4}{*}{\textbf{MT-B}} & $\theta^P_1$ & 0.2624 & 0.2865 & 0.2952 & 0.2686 & 0.2700 & 0.2770 & \textbf{0.2541} \\
                               & $\theta^P_2$ & 0.3609 & 0.3923 & 0.4007 & 0.3709 & 0.3819 & 0.3720 & \textbf{0.3494} \\
                               & $\theta^S_1$ & 0.2579 & 0.2890 & 0.2965 & 0.2732 & 0.2846 & 0.2819 & \textbf{0.2570} \\
                               & $\theta^S_2$ & \textbf{0.3510} & 0.3959 & 0.4064 & 0.3760 & 0.4017 & 0.3808 & 0.3538 \\
                               &              &        &        &        &        &        &        &        \\
\multirow{4}{*}{\textbf{MT-C}} & $\theta^P_1$ & 0.1051 & 0.1037 & 0.1023 & 0.0994 & 0.1052 & 0.1044 & \textbf{0.0929} \\
                               & $\theta^P_2$ & 0.0846 & 0.0819 & 0.0787 & 0.0745 & 0.0722 & 0.0855 & \textbf{0.0711} \\
                               & $\theta^S_1$ & 0.1081 & 0.1007 & 0.1075 & 0.1017 & 0.1085 & 0.1040 & \textbf{0.0948} \\
                               & $\theta^S_2$ & 0.0874 & 0.0779 & 0.0782 & 0.0760 & 0.0773 & 0.0842 & \textbf{0.0727} \\
\bottomrule
\end{tabular}
\end{adjustbox}
\caption{Ratio of MSE under all designs against those under complete randomization in both stages with $G=400$}
\label{table:mse400}
\end{table}

\begin{table}[ht!]
\centering
\setlength{\tabcolsep}{5pt}
\begin{adjustbox}{max width=0.75\linewidth,center}
\begin{tabular}{lllllllll}
\toprule
 & & \multicolumn{7}{c}{Second-stage}   \\\cmidrule{3-9}
First-stage              & Parameter & \textbf{C}    & \textbf{S-2}   & \textbf{S-4}   & \textbf{S-4O} & \textbf{MT-A}  & \textbf{MT-B}  & \textbf{MT-C}   \\
\midrule
\multirow{4}{*}{\textbf{C}}    & $\theta^P_1$ & 1.0000 & \textbf{0.9955} & 1.1978 & 1.1322 & 1.0328 & 1.0124 & 0.9957 \\
                               & $\theta^P_2$ & 1.0000 & 1.0020 & 1.1529 & 1.0984 & 1.0089 & \textbf{0.9641} & 1.0020 \\
                               & $\theta^S_1$ & 1.0000 & \textbf{0.9965} & 1.1980 & 1.1400 & 1.0430 & 1.0020 & 0.9982 \\
                               & $\theta^S_2$ & 1.0000 & 1.0109 & 1.1634 & 1.1106 & 1.0240 & \textbf{0.9558} & 1.0158 \\
                               &              &        &        &        &        &        &        &        \\
\multirow{4}{*}{\textbf{S-2}}  & $\theta^P_1$ & 0.8607 & \textbf{0.7817} & 0.8193 & 0.8848 & 0.8838 & 0.8489 & 0.8245 \\
                               & $\theta^P_2$ & 0.8171 & \textbf{0.7510} & 0.7705 & 0.8458 & 0.8446 & 0.7941 & 0.7625 \\
                               & $\theta^S_1$ & 0.8601 & \textbf{0.7886} & 0.8206 & 0.8745 & 0.8782 & 0.8612 & 0.8236 \\
                               & $\theta^S_2$ & 0.8210 & 0.7709 & 0.7783 & 0.8460 & 0.8475 & 0.8088 & \textbf{0.7693} \\
                               &              &        &        &        &        &        &        &        \\
\multirow{4}{*}{\textbf{S-4}}  & $\theta^P_1$ & 0.8825 & 0.9051 & 0.8238 & 0.7892 & 0.8557 & 0.8355 & \textbf{0.7850} \\
                               & $\theta^P_2$ & 0.8386 & 0.8428 & 0.7624 & 0.7282 & 0.7979 & 0.7964 & \textbf{0.7309} \\
                               & $\theta^S_1$ & 0.8727 & 0.9021 & 0.8305 & \textbf{0.7764} & 0.8519 & 0.8432 & 0.7907 \\
                               & $\theta^S_2$ & 0.8433 & 0.8530 & 0.7718 & \textbf{0.7206} & 0.8046 & 0.8079 & 0.7435 \\
                               &              &        &        &        &        &        &        &        \\
\multirow{4}{*}{\textbf{S-4O}} & $\theta^P_1$ & 0.2310 & 0.2041 & 0.2060 & 0.2153 & 0.2062 & 0.2207 & \textbf{0.2032} \\
                               & $\theta^P_2$ & 0.2564 & 0.2345 & 0.2343 & 0.2417 & \textbf{0.2232} & 0.2518 & 0.2367 \\
                               & $\theta^S_1$ & 0.2414 & 0.2118 & 0.2039 & 0.2170 & \textbf{0.2036} & 0.2194 & 0.2072 \\
                               & $\theta^S_2$ & 0.2671 & 0.2407 & 0.2344 & 0.2463 & \textbf{0.2232} & 0.2547 & 0.2420 \\
                               &              &        &        &        &        &        &        &        \\
\multirow{4}{*}{\textbf{MT-A}} & $\theta^P_1$ & \textbf{0.7313} & 0.8537 & 0.8119 & 0.8049 & 0.7604 & 0.7885 & 0.8531 \\
                               & $\theta^P_2$ & \textbf{0.6727} & 0.7814 & 0.7558 & 0.7502 & 0.6988 & 0.7227 & 0.7958 \\
                               & $\theta^S_1$ & \textbf{0.7437} & 0.8584 & 0.8194 & 0.8027 & 0.7644 & 0.7956 & 0.8542 \\
                               & $\theta^S_2$ & \textbf{0.6917} & 0.7919 & 0.7729 & 0.7512 & 0.7100 & 0.7450 & 0.8049 \\
                               &              &        &        &        &        &        &        &        \\
\multirow{4}{*}{\textbf{MT-B}} & $\theta^P_1$ & 0.3074 & 0.3398 & 0.3206 & 0.3234 & 0.2996 & 0.3054 & \textbf{0.2938} \\
                               & $\theta^P_2$ & 0.4219 & 0.4612 & 0.4404 & 0.4545 & 0.4183 & 0.4151 & \textbf{0.3998} \\
                               & $\theta^S_1$ & 0.3167 & 0.3288 & 0.3171 & 0.3265 & 0.3006 & 0.3052 & \textbf{0.2943} \\
                               & $\theta^S_2$ & 0.4383 & 0.4542 & 0.4417 & 0.4580 & 0.4236 & 0.4189 & \textbf{0.4035} \\
                               &              &        &        &        &        &        &        &        \\
\multirow{4}{*}{\textbf{MT-C}} & $\theta^P_1$ & 0.1272 & 0.1120 & 0.1182 & 0.1181 & 0.1147 & 0.1239 & \textbf{0.1113} \\
                               & $\theta^P_2$ & 0.1030 & 0.0873 & 0.0930 & 0.0862 & 0.0867 & 0.1039 & \textbf{0.0802} \\
                               & $\theta^S_1$ & 0.1288 & 0.1132 & 0.1192 & 0.1180 & 0.1179 & 0.1287 & \textbf{0.1095} \\
                               & $\theta^S_2$ & 0.1045 & 0.0923 & 0.0958 & 0.0878 & 0.0903 & 0.1033 & \textbf{0.0792} \\
\bottomrule
\end{tabular}
\end{adjustbox}
\caption{Ratio of MSE under all designs against those under complete randomization in both stages with $G=800$}
\label{table:mse800}
\end{table}

\begin{table}[ht!]
\centering
\setlength{\tabcolsep}{5pt}
\begin{adjustbox}{max width=0.75\linewidth,center}
\begin{tabular}{lllllllll}
\toprule
 & & \multicolumn{7}{c}{Second-stage}   \\\cmidrule{3-9}
First-stage              & Parameter & \textbf{C}    & \textbf{S-2}   & \textbf{S-4}   & \textbf{S-4O} & \textbf{MT-A}  & \textbf{MT-B}  & \textbf{MT-C}   \\
\midrule
\multirow{4}{*}{\textbf{C}}    & $\theta^P_1$ & 1.0000 & 1.0261 & 1.0119 & \textbf{0.8974} & 1.0175 & 0.9374 & 0.9806 \\
                               & $\theta^P_2$ & 1.0000 & 1.0226 & 1.0010 & \textbf{0.8975} & 0.9820 & 0.9470 & 0.9877 \\
                               & $\theta^S_1$ & 1.0000 & 1.0560 & 1.0366 & \textbf{0.9179} & 1.0326 & 0.9760 & 1.0015 \\
                               & $\theta^S_2$ & 1.0000 & 1.0454 & 1.0195 & \textbf{0.9091} & 0.9926 & 0.9798 & 1.0027 \\
                               &              &        &        &        &        &        &        &        \\
\multirow{4}{*}{\textbf{S-2}}  & $\theta^P_1$ & 0.7743 & 0.9256 & 0.8678 & 0.8789 & 0.8183 & \textbf{0.7698} & 0.8305 \\
                               & $\theta^P_2$ & 0.7290 & 0.8828 & 0.8291 & 0.8297 & 0.7911 & \textbf{0.7140} & 0.7824 \\
                               & $\theta^S_1$ & 0.7962 & 0.9451 & 0.8740 & 0.8841 & 0.8377 & \textbf{0.7754} & 0.8427 \\
                               & $\theta^S_2$ & 0.7518 & 0.9050 & 0.8271 & 0.8284 & 0.8086 & \textbf{0.7190} & 0.7897 \\
                               &              &        &        &        &        &        &        &        \\
\multirow{4}{*}{\textbf{S-4}}  & $\theta^P_1$ & 0.8211 & 0.7965 & 0.7692 & 0.7757 & \textbf{0.7574} & 0.7600 & 0.7865 \\
                               & $\theta^P_2$ & 0.7503 & 0.7374 & 0.7501 & 0.7323 & 0.7017 & \textbf{0.6958} & 0.7212 \\
                               & $\theta^S_1$ & 0.8435 & 0.8254 & 0.7831 & 0.7869 & 0.7757 & \textbf{0.7697} & 0.8064 \\
                               & $\theta^S_2$ & 0.7678 & 0.7586 & 0.7592 & 0.7394 & 0.7195 & \textbf{0.7009} & 0.7370 \\
                               &              &        &        &        &        &        &        &        \\
\multirow{4}{*}{\textbf{S-4O}} & $\theta^P_1$ & 0.2185 & 0.2104 & 0.2041 & 0.2094 & \textbf{0.2007} & 0.2051 & 0.2080 \\
                               & $\theta^P_2$ & 0.2489 & 0.2442 & 0.2283 & 0.2348 & \textbf{0.2135} & 0.2252 & 0.2307 \\
                               & $\theta^S_1$ & 0.2222 & 0.2069 & \textbf{0.2037} & 0.2144 & 0.2089 & 0.2051 & 0.2116 \\
                               & $\theta^S_2$ & 0.2464 & 0.2465 & 0.2305 & 0.2424 & \textbf{0.2230} & 0.2245 & 0.2336 \\
                               &              &        &        &        &        &        &        &        \\
\multirow{4}{*}{\textbf{MT-A}} & $\theta^P_1$ & 0.7618 & \textbf{0.6901} & 0.7937 & 0.7355 & 0.7084 & 0.7585 & 0.7045 \\
                               & $\theta^P_2$ & 0.7037 & \textbf{0.6487} & 0.7538 & 0.6837 & 0.6585 & 0.7258 & 0.6676 \\
                               & $\theta^S_1$ & 0.7712 & \textbf{0.6907} & 0.8058 & 0.7537 & 0.7370 & 0.7793 & 0.7196 \\
                               & $\theta^S_2$ & 0.7159 & \textbf{0.6457} & 0.7565 & 0.7017 & 0.6808 & 0.7512 & 0.6785 \\
                               &              &        &        &        &        &        &        &        \\
\multirow{4}{*}{\textbf{MT-B}} & $\theta^P_1$ & 0.2925 & 0.2883 & 0.2906 & 0.2810 & 0.2853 & 0.2779 & \textbf{0.2694} \\
                               & $\theta^P_2$ & 0.3986 & 0.3923 & 0.3964 & 0.3853 & 0.3904 & 0.3788 & \textbf{0.3745} \\
                               & $\theta^S_1$ & 0.2984 & 0.2939 & 0.3018 & 0.2859 & 0.2861 & 0.2891 & \textbf{0.2763} \\
                               & $\theta^S_2$ & 0.4042 & 0.3952 & 0.4071 & 0.3908 & 0.3944 & 0.3857 & \textbf{0.3826} \\
                               &              &        &        &        &        &        &        &        \\
\multirow{4}{*}{\textbf{MT-C}} & $\theta^P_1$ & 0.1104 & 0.1160 & 0.1044 & 0.1070 & \textbf{0.1027} & 0.1187 & 0.1071 \\
                               & $\theta^P_2$ & 0.0846 & 0.0912 & 0.0853 & \textbf{0.0779} & 0.0786 & 0.0938 & 0.0808 \\
                               & $\theta^S_1$ & 0.1166 & 0.1140 & \textbf{0.1044} & 0.1103 & 0.1053 & 0.1260 & 0.1093 \\
                               & $\theta^S_2$ & 0.0922 & 0.0871 & 0.0869 & 0.0819 & \textbf{0.0806} & 0.1015 & 0.0820 \\
\bottomrule
\end{tabular}
\end{adjustbox}
\caption{Ratio of MSE under all designs against those under complete randomization in both stages with $G=1000$}
\label{table:mse1000}
\end{table}

\begin{table}[ht!]
\centering
\setlength{\tabcolsep}{3pt}
\begin{adjustbox}{max width=0.95\linewidth,center}
\begin{tabular}{lllllllllllllll}
\toprule
       &              & \multicolumn{13}{c}{Second-stage}    \\  \cmidrule{3-15}
       &              & \multicolumn{6}{c}{$H_0: \tau = \omega = 0$}                      &  & \multicolumn{6}{c}{$H_1: \tau = \omega = 0.05$}                      \\  \cmidrule{3-8} \cmidrule{10-15}
First-stage & Parameter &  \textbf{S-2}   & \textbf{S-4}   & \textbf{S-4O} & \textbf{MT-A}  & \textbf{MT-B}  & \textbf{MT-C} & & \textbf{S-2}   & \textbf{S-4}   & \textbf{S-4O} & \textbf{MT-A}  & \textbf{MT-B}  & \textbf{MT-C} \\
\toprule
\multirow{4}{*}{\textbf{S-2}}  & $\theta^P_1$ & 0.048 & 0.043 & 0.044 & 0.043 & 0.059 & 0.057 &  & 0.433 & 0.450 & 0.430 & 0.452 & 0.402 & 0.403 \\
                               & $\theta^P_2$ & 0.047 & 0.042 & 0.039 & 0.045 & 0.050 & 0.052 &  & 0.422 & 0.431 & 0.413 & 0.444 & 0.370 & 0.399 \\
                               & $\theta^S_1$ & 0.049 & 0.037 & 0.048 & 0.042 & 0.058 & 0.058 &  & 0.155 & 0.138 & 0.140 & 0.145 & 0.124 & 0.130 \\
                               & $\theta^S_2$ & 0.050 & 0.032 & 0.056 & 0.051 & 0.051 & 0.053 &  & 0.149 & 0.135 & 0.136 & 0.141 & 0.127 & 0.139 \\
\multicolumn{1}{l}{}           &              &       &       &       &       &       &       &  &       &       &       &       &       &       \\
\multirow{4}{*}{\textbf{S-4}}  & $\theta^P_1$ & 0.062 & 0.046 & 0.060 & 0.054 & 0.045 & 0.056 &  & 0.424 & 0.446 & 0.454 & 0.438 & 0.428 & 0.457 \\
                               & $\theta^P_2$ & 0.058 & 0.062 & 0.061 & 0.056 & 0.046 & 0.057 &  & 0.423 & 0.434 & 0.438 & 0.420 & 0.423 & 0.439 \\
                               & $\theta^S_1$ & 0.056 & 0.052 & 0.058 & 0.056 & 0.050 & 0.056 &  & 0.158 & 0.150 & 0.148 & 0.147 & 0.146 & 0.147 \\
                               & $\theta^S_2$ & 0.057 & 0.058 & 0.063 & 0.052 & 0.041 & 0.061 &  & 0.152 & 0.158 & 0.151 & 0.150 & 0.147 & 0.139 \\
\multicolumn{1}{l}{}           &              &       &       &       &       &       &       &  &       &       &       &       &       &       \\
\multirow{4}{*}{\textbf{S-4O}} & $\theta^P_1$ & 0.051 & 0.057 & 0.055 & 0.046 & 0.064 & 0.059 &  & 0.932 & 0.938 & 0.942 & 0.932 & 0.941 & 0.945 \\
                               & $\theta^P_2$ & 0.057 & 0.059 & 0.053 & 0.055 & 0.054 & 0.058 &  & 0.882 & 0.885 & 0.886 & 0.876 & 0.873 & 0.896 \\
                               & $\theta^S_1$ & 0.050 & 0.045 & 0.060 & 0.042 & 0.064 & 0.059 &  & 0.442 & 0.412 & 0.431 & 0.390 & 0.409 & 0.418 \\
                               & $\theta^S_2$ & 0.058 & 0.050 & 0.048 & 0.046 & 0.054 & 0.063 &  & 0.375 & 0.341 & 0.357 & 0.310 & 0.342 & 0.370 \\
\multicolumn{1}{l}{}           &              &       &       &       &       &       &       &  &       &       &       &       &       &       \\
\multirow{4}{*}{\textbf{MT-A}} & $\theta^P_1$ & 0.059 & 0.050 & 0.054 & 0.063 & 0.046 & 0.060 &  & 0.431 & 0.469 & 0.458 & 0.454 & 0.439 & 0.458 \\
                               & $\theta^P_2$ & 0.056 & 0.059 & 0.054 & 0.067 & 0.041 & 0.060 &  & 0.440 & 0.469 & 0.457 & 0.447 & 0.419 & 0.450 \\
                               & $\theta^S_1$ & 0.063 & 0.056 & 0.055 & 0.060 & 0.060 & 0.058 &  & 0.133 & 0.155 & 0.162 & 0.152 & 0.145 & 0.144 \\
                               & $\theta^S_2$ & 0.056 & 0.054 & 0.057 & 0.062 & 0.053 & 0.058 &  & 0.129 & 0.157 & 0.155 & 0.152 & 0.149 & 0.147 \\
\multicolumn{1}{l}{}           &              &       &       &       &       &       &       &  &       &       &       &       &       &       \\
\multirow{4}{*}{\textbf{MT-B}} & $\theta^P_1$ & 0.054 & 0.033 & 0.049 & 0.053 & 0.047 & 0.054 &  & 0.835 & 0.832 & 0.856 & 0.837 & 0.821 & 0.862 \\
                               & $\theta^P_2$ & 0.050 & 0.042 & 0.048 & 0.068 & 0.054 & 0.062 &  & 0.656 & 0.675 & 0.674 & 0.675 & 0.657 & 0.675 \\
                               & $\theta^S_1$ & 0.045 & 0.042 & 0.053 & 0.054 & 0.047 & 0.056 &  & 0.336 & 0.321 & 0.324 & 0.315 & 0.296 & 0.332 \\
                               & $\theta^S_2$ & 0.051 & 0.045 & 0.042 & 0.067 & 0.048 & 0.060 &  & 0.248 & 0.228 & 0.236 & 0.226 & 0.203 & 0.241 \\
\multicolumn{1}{l}{}           &              &       &       &       &       &       &       &  &       &       &       &       &       &       \\
\multirow{4}{*}{\textbf{MT-C}} & $\theta^P_1$ & 0.039 & 0.057 & 0.048 & 0.041 & 0.056 & 0.058 &  & 0.996 & 0.996 & 0.997 & 0.996 & 0.999 & 0.999 \\
                               & $\theta^P_2$ & 0.054 & 0.048 & 0.045 & 0.043 & 0.051 & 0.050 &  & 1.000 & 0.998 & 0.998 & 0.999 & 1.000 & 1.000 \\
                               & $\theta^S_1$ & 0.040 & 0.054 & 0.046 & 0.046 & 0.047 & 0.056 &  & 0.705 & 0.673 & 0.677 & 0.672 & 0.674 & 0.677 \\
                               & $\theta^S_2$ & 0.051 & 0.047 & 0.050 & 0.049 & 0.036 & 0.047 &  & 0.741 & 0.741 & 0.751 & 0.740 & 0.734 & 0.743 \\
\bottomrule
\end{tabular}
\end{adjustbox}
\caption{Rejection probabilities under the null and alternative hypothesis with $G=400$}
\label{table:reject-probs-400}
\end{table}

\begin{table}[ht!]
\centering
\setlength{\tabcolsep}{3pt}
\begin{adjustbox}{max width=0.95\linewidth,center}
\begin{tabular}{lllllllllllllll}
\toprule
       &              & \multicolumn{13}{c}{Second-stage}    \\  \cmidrule{3-15}
       &              & \multicolumn{6}{c}{$H_0: \tau = \omega = 0$}                      &  & \multicolumn{6}{c}{$H_1: \tau = \omega = 0.05$}                      \\  \cmidrule{3-8} \cmidrule{10-15}
First-stage & Parameter &  \textbf{S-2}   & \textbf{S-4}   & \textbf{S-4O} & \textbf{MT-A}  & \textbf{MT-B}  & \textbf{MT-C} & & \textbf{S-2}   & \textbf{S-4}   & \textbf{S-4O} & \textbf{MT-A}  & \textbf{MT-B}  & \textbf{MT-C} \\
\toprule
\multirow{4}{*}{\textbf{S-2}}  & $\theta^P_1$ & 0.049 & 0.047 & 0.039 & 0.051 & 0.055 & 0.044 &  & 0.713 & 0.693 & 0.720 & 0.711 & 0.691 & 0.679 \\
                               & $\theta^P_2$ & 0.049 & 0.049 & 0.044 & 0.042 & 0.057 & 0.049 &  & 0.694 & 0.681 & 0.692 & 0.695 & 0.662 & 0.666 \\
                               & $\theta^S_1$ & 0.053 & 0.047 & 0.041 & 0.048 & 0.065 & 0.049 &  & 0.248 & 0.229 & 0.251 & 0.226 & 0.242 & 0.239 \\
                               & $\theta^S_2$ & 0.046 & 0.043 & 0.042 & 0.046 & 0.061 & 0.053 &  & 0.232 & 0.218 & 0.238 & 0.224 & 0.235 & 0.236 \\
\multicolumn{1}{l}{}           &              &       &       &       &       &       &       &  &       &       &       &       &       &       \\
\multirow{4}{*}{\textbf{S-4}}  & $\theta^P_1$ & 0.058 & 0.047 & 0.046 & 0.055 & 0.055 & 0.052 &  & 0.727 & 0.716 & 0.731 & 0.712 & 0.711 & 0.729 \\
                               & $\theta^P_2$ & 0.052 & 0.048 & 0.053 & 0.053 & 0.051 & 0.054 &  & 0.718 & 0.714 & 0.714 & 0.692 & 0.683 & 0.708 \\
                               & $\theta^S_1$ & 0.055 & 0.048 & 0.053 & 0.062 & 0.056 & 0.054 &  & 0.264 & 0.240 & 0.253 & 0.236 & 0.247 & 0.269 \\
                               & $\theta^S_2$ & 0.056 & 0.052 & 0.060 & 0.057 & 0.052 & 0.051 &  & 0.251 & 0.248 & 0.243 & 0.213 & 0.251 & 0.243 \\
\multicolumn{1}{l}{}           &              &       &       &       &       &       &       &  &       &       &       &       &       &       \\
\multirow{4}{*}{\textbf{S-4O}} & $\theta^P_1$ & 0.062 & 0.055 & 0.045 & 0.060 & 0.053 & 0.054 &  & 0.999 & 0.998 & 1.000 & 1.000 & 0.999 & 0.999 \\
                               & $\theta^P_2$ & 0.062 & 0.067 & 0.054 & 0.060 & 0.057 & 0.044 &  & 0.994 & 0.990 & 0.994 & 0.998 & 0.997 & 0.997 \\
                               & $\theta^S_1$ & 0.057 & 0.051 & 0.049 & 0.053 & 0.048 & 0.058 &  & 0.685 & 0.695 & 0.725 & 0.702 & 0.703 & 0.716 \\
                               & $\theta^S_2$ & 0.064 & 0.061 & 0.055 & 0.058 & 0.053 & 0.042 &  & 0.610 & 0.591 & 0.636 & 0.602 & 0.580 & 0.626 \\
\multicolumn{1}{l}{}           &              &       &       &       &       &       &       &  &       &       &       &       &       &       \\
\multirow{4}{*}{\textbf{MT-A}} & $\theta^P_1$ & 0.048 & 0.043 & 0.051 & 0.036 & 0.050 & 0.051 &  & 0.755 & 0.730 & 0.723 & 0.737 & 0.732 & 0.760 \\
                               & $\theta^P_2$ & 0.047 & 0.044 & 0.055 & 0.036 & 0.054 & 0.059 &  & 0.747 & 0.721 & 0.732 & 0.719 & 0.718 & 0.750 \\
                               & $\theta^S_1$ & 0.046 & 0.040 & 0.053 & 0.038 & 0.053 & 0.055 &  & 0.248 & 0.286 & 0.241 & 0.277 & 0.258 & 0.251 \\
                               & $\theta^S_2$ & 0.046 & 0.044 & 0.056 & 0.038 & 0.053 & 0.062 &  & 0.240 & 0.279 & 0.231 & 0.270 & 0.242 & 0.252 \\
\multicolumn{1}{l}{}           &              &       &       &       &       &       &       &  &       &       &       &       &       &       \\
\multirow{4}{*}{\textbf{MT-B}} & $\theta^P_1$ & 0.048 & 0.050 & 0.055 & 0.052 & 0.048 & 0.051 &  & 0.988 & 0.992 & 0.991 & 0.986 & 0.989 & 0.987 \\
                               & $\theta^P_2$ & 0.046 & 0.057 & 0.060 & 0.055 & 0.048 & 0.046 &  & 0.919 & 0.925 & 0.932 & 0.925 & 0.922 & 0.927 \\
                               & $\theta^S_1$ & 0.051 & 0.052 & 0.060 & 0.064 & 0.042 & 0.053 &  & 0.563 & 0.558 & 0.555 & 0.552 & 0.543 & 0.542 \\
                               & $\theta^S_2$ & 0.046 & 0.054 & 0.059 & 0.062 & 0.039 & 0.052 &  & 0.395 & 0.416 & 0.393 & 0.405 & 0.385 & 0.392 \\
\multicolumn{1}{l}{}           &              &       &       &       &       &       &       &  &       &       &       &       &       &       \\
\multirow{4}{*}{\textbf{MT-C}} & $\theta^P_1$ & 0.052 & 0.047 & 0.052 & 0.051 & 0.042 & 0.038 &  & 1.000 & 1.000 & 1.000 & 1.000 & 1.000 & 1.000 \\
                               & $\theta^P_2$ & 0.045 & 0.041 & 0.053 & 0.044 & 0.056 & 0.048 &  & 1.000 & 1.000 & 1.000 & 1.000 & 1.000 & 1.000 \\
                               & $\theta^S_1$ & 0.046 & 0.044 & 0.059 & 0.043 & 0.049 & 0.042 &  & 0.929 & 0.924 & 0.920 & 0.926 & 0.922 & 0.937 \\
                               & $\theta^S_2$ & 0.048 & 0.049 & 0.051 & 0.036 & 0.051 & 0.054 &  & 0.957 & 0.963 & 0.968 & 0.959 & 0.947 & 0.979 \\
\bottomrule
\end{tabular}
\end{adjustbox}
\caption{Rejection probabilities under the null and alternative hypothesis with $G=800$}
\label{table:reject-probs-800}
\end{table}

\begin{table}[ht!]
\centering
\setlength{\tabcolsep}{3pt}
\begin{adjustbox}{max width=0.95\linewidth,center}
\begin{tabular}{lllllllllllllll}
\toprule
       &              & \multicolumn{13}{c}{Second-stage}    \\  \cmidrule{3-15}
       &              & \multicolumn{6}{c}{$H_0: \tau = \omega = 0$}                      &  & \multicolumn{6}{c}{$H_1: \tau = \omega = 0.05$}                      \\  \cmidrule{3-8} \cmidrule{10-15}
First-stage & Parameter &  \textbf{S-2}   & \textbf{S-4}   & \textbf{S-4O} & \textbf{MT-A}  & \textbf{MT-B}  & \textbf{MT-C} & & \textbf{S-2}   & \textbf{S-4}   & \textbf{S-4O} & \textbf{MT-A}  & \textbf{MT-B}  & \textbf{MT-C} \\
\toprule
\multirow{4}{*}{\textbf{S-2}}  & $\theta^P_1$ & 0.046 & 0.047 & 0.046 & 0.052 & 0.049 & 0.061 &  & 0.804 & 0.778 & 0.762 & 0.788 & 0.793 & 0.803 \\
                               & $\theta^P_2$ & 0.042 & 0.050 & 0.055 & 0.058 & 0.062 & 0.052 &  & 0.790 & 0.748 & 0.744 & 0.764 & 0.777 & 0.796 \\
                               & $\theta^S_1$ & 0.051 & 0.044 & 0.052 & 0.050 & 0.049 & 0.061 &  & 0.310 & 0.280 & 0.259 & 0.270 & 0.274 & 0.292 \\
                               & $\theta^S_2$ & 0.049 & 0.046 & 0.052 & 0.055 & 0.049 & 0.051 &  & 0.297 & 0.275 & 0.248 & 0.280 & 0.258 & 0.279 \\
\multicolumn{1}{l}{}           &              &       &       &       &       &       &       &  &       &       &       &       &       &       \\
\multirow{4}{*}{\textbf{S-4}}  & $\theta^P_1$ & 0.041 & 0.059 & 0.045 & 0.059 & 0.057 & 0.050 &  & 0.815 & 0.836 & 0.825 & 0.811 & 0.793 & 0.833 \\
                               & $\theta^P_2$ & 0.043 & 0.058 & 0.040 & 0.057 & 0.054 & 0.052 &  & 0.812 & 0.831 & 0.816 & 0.806 & 0.786 & 0.814 \\
                               & $\theta^S_1$ & 0.044 & 0.058 & 0.043 & 0.063 & 0.062 & 0.044 &  & 0.306 & 0.312 & 0.307 & 0.288 & 0.326 & 0.300 \\
                               & $\theta^S_2$ & 0.056 & 0.059 & 0.044 & 0.070 & 0.055 & 0.053 &  & 0.296 & 0.318 & 0.290 & 0.297 & 0.309 & 0.306 \\
\multicolumn{1}{l}{}           &              &       &       &       &       &       &       &  &       &       &       &       &       &       \\
\multirow{4}{*}{\textbf{S-4O}}  & $\theta^P_1$ & 0.051 & 0.039 & 0.053 & 0.047 & 0.057 & 0.053 &  & 0.999 & 1.000 & 1.000 & 1.000 & 1.000 & 1.000 \\
                               & $\theta^P_2$ & 0.055 & 0.035 & 0.048 & 0.045 & 0.048 & 0.056 &  & 0.999 & 0.999 & 0.997 & 0.998 & 0.998 & 0.999 \\
                               & $\theta^S_1$ & 0.050 & 0.045 & 0.048 & 0.048 & 0.053 & 0.051 &  & 0.812 & 0.793 & 0.788 & 0.806 & 0.750 & 0.791 \\
                               & $\theta^S_2$ & 0.055 & 0.036 & 0.053 & 0.040 & 0.051 & 0.052 &  & 0.708 & 0.705 & 0.681 & 0.699 & 0.669 & 0.703 \\
\multicolumn{1}{l}{}           &              &       &       &       &       &       &       &  &       &       &       &       &       &       \\
\multirow{4}{*}{\textbf{MT-A}} & $\theta^P_1$ & 0.049 & 0.049 & 0.061 & 0.058 & 0.048 & 0.058 &  & 0.830 & 0.809 & 0.837 & 0.823 & 0.809 & 0.813 \\
                               & $\theta^P_2$ & 0.044 & 0.051 & 0.060 & 0.054 & 0.053 & 0.058 &  & 0.812 & 0.803 & 0.830 & 0.813 & 0.799 & 0.804 \\
                               & $\theta^S_1$ & 0.051 & 0.056 & 0.053 & 0.058 & 0.056 & 0.061 &  & 0.314 & 0.316 & 0.301 & 0.291 & 0.273 & 0.284 \\
                               & $\theta^S_2$ & 0.048 & 0.050 & 0.052 & 0.051 & 0.061 & 0.053 &  & 0.320 & 0.314 & 0.296 & 0.276 & 0.277 & 0.286 \\
\multicolumn{1}{l}{}           &              &       &       &       &       &       &       &  &       &       &       &       &       &       \\
\multirow{4}{*}{\textbf{MT-B}} & $\theta^P_1$ & 0.052 & 0.041 & 0.048 & 0.051 & 0.058 & 0.058 &  & 0.998 & 0.995 & 0.999 & 0.995 & 0.996 & 0.999 \\
                               & $\theta^P_2$ & 0.060 & 0.048 & 0.054 & 0.050 & 0.059 & 0.052 &  & 0.969 & 0.966 & 0.977 & 0.963 & 0.964 & 0.970 \\
                               & $\theta^S_1$ & 0.052 & 0.048 & 0.056 & 0.049 & 0.056 & 0.060 &  & 0.661 & 0.657 & 0.673 & 0.632 & 0.594 & 0.673 \\
                               & $\theta^S_2$ & 0.054 & 0.049 & 0.055 & 0.057 & 0.053 & 0.054 &  & 0.486 & 0.484 & 0.474 & 0.476 & 0.431 & 0.494 \\
\multicolumn{1}{l}{}           &              &       &       &       &       &       &       &  &       &       &       &       &       &       \\
\multirow{4}{*}{\textbf{MT-C}} & $\theta^P_1$ & 0.054 & 0.050 & 0.052 & 0.056 & 0.055 & 0.044 &  & 1.000 & 1.000 & 1.000 & 1.000 & 1.000 & 1.000 \\
                               & $\theta^P_2$ & 0.045 & 0.052 & 0.053 & 0.047 & 0.059 & 0.037 &  & 1.000 & 1.000 & 1.000 & 1.000 & 1.000 & 1.000 \\
                               & $\theta^S_1$ & 0.051 & 0.036 & 0.058 & 0.053 & 0.048 & 0.049 &  & 0.962 & 0.964 & 0.968 & 0.963 & 0.955 & 0.975 \\
                               & $\theta^S_2$ & 0.043 & 0.046 & 0.069 & 0.056 & 0.047 & 0.034 &  & 0.987 & 0.988 & 0.988 & 0.982 & 0.978 & 0.990 \\
\bottomrule
\end{tabular}
\end{adjustbox}
\caption{Rejection probabilities under the null and alternative hypothesis with $G=1000$}
\label{table:reject-probs-1000}
\end{table}

\clearpage
\newpage
\bibliography{reference}

\end{document}